\documentclass[11pt,a4paper]{article}
\usepackage{graphicx}
\usepackage{ifpdf}
\usepackage{amsthm,amsfonts,amssymb,amsmath,euscript,comment,bbm}
\usepackage{mathabx}
\usepackage[utf8]{inputenc}
\usepackage[T1]{fontenc}
\usepackage{color}
\usepackage{textcomp}
\usepackage{bbm,bm}
\usepackage{slashed}
\usepackage[colorlinks=true,linkcolor=red,citecolor=blue]{hyperref}
\usepackage{mathrsfs}
\usepackage{dsfont}
\usepackage{geometry}
\usepackage{epsfig}
\usepackage{tikz}
\usetikzlibrary{arrows.meta}
\usepackage{float}
\usetikzlibrary{decorations.pathmorphing}
 \usetikzlibrary{decorations.markings}
\usepackage{tikz-3dplot}
\numberwithin{equation}{section}
\renewcommand{\c}{\cdot}
\newcommand{\pr}{\partial}
\def\ub{{\underline{u}}}
\def\cb{\mathbf{c}}
\newcommand{\x}{\mathbf{x}}

\newcommand{\mo}{\mathcal{O}}
\newcommand{\omb}{\underline{\omega}}
\newcommand{\etab}{\underline{\eta}}
\def\sld{\slashed{d}}
\renewcommand{\div}{\sdiv}
\def\Rl{R_{l}}
\def\Rr{R_{r}}
\def\NNN{\mathbb{N}}
\def\Fb{{\underline{F}}}
\def\slgc{\widecheck{\slg}}
\def\ins{\slashed{\in}}

\def\Des{\slashed{\De}}

\def\Nb{\underline{N}}
\def\UU{\mathcal{U}}

\def\slg{\slashed{g}}
\def\y{\mathbf{y}}
\def\sdivs{\slashed{\sdiv}}
\def\curls{\slashed{\curl}}
\def\cuvss{{H_{u_0}^{(\de,\ub)}}}
\def\ucuvss{{\Hb_\de^{(u_0,u)}}}

\def\af{a^\frac{1}{2}}

\def\RRR{\mathbb{R}}
\def\Rk{\mathfrak{R}}
\def\Rkb{{\underline{\Rk}}}
\def\Ok{\mathfrak{O}}
\def\Hb{\underline{H}}
\def\aa{{\underline{\a}}}
\def\sk{\mathfrak{s}}
\def\th{\theta}

\def\nabs{\slashed{\nab}}
\def\slg{{\slashed{g}}}
\def\slgo{\overset{\circ}{\slg}}
\def\afd{a^{-\frac{1}{2}}}
\def\trcht{\widetilde{\trch}}
\def\trchbt{\widetilde{\trchb}}

\newcommand{\mr}{\mathcal{R}}
\newcommand{\ur}{\underline{\mr}}

\newcommand{\M}{\mathcal{M}}
\newcommand{\D}{\mathbf{D}}
\newcommand{\cuv}{{H_u^{(0,\ub)}}}
\newcommand{\ucuv}{{\Hb_\ub^{(u_\infty,u)}}}
\newcommand{\cuvs}{{H_u^{(\de,\ub)}}}
\newcommand{\ucuvs}{{\Hb_\ub^{(u_0,u)}}}
\newcommand{\dk}{{\mathfrak{d}}}

\newcommand{\g}{\mathbf{g}}

\newcommand{\R}{{\mathbf{R}}}
\newcommand{\K}{\mathbf{K}}
\def\ep{\varepsilon}
\def\eps{\epsilon}
\def\hot{\widehat{\otimes}}

\def\Lb{\underline{L}}

\def\II{\mathcal{I}}
\def\la{\lambda}

\def\MM{\mathcal{M}}
\def\E{\mathbf{E}}
\def\P{\mathbf{P}}
\def\C{\mathbf{C}}
\def\J{\mathbf{J}}
\def\Q{\mathbf{Q}}
\def\V{\mathbf{V}}
\def\chib{{\underline{\chi}}}
\def\om{\omega}
\def\ze{\zeta}
\def\Om{\Omega}
\def\bb{{\underline{\b}}}
\def\Si{\Sigma}
\def\si{\sigma}
\def\mub{{\underline{\mu}}}
\def\dkb{{\underline{\slashed{\dk}}}}
\def\dkbb{\slashed{\dk}}
\def\ga{\gamma}
\def\Ga{\Gamma}
\def\xib{\underline{\xi}}
\def\a{\alpha}
\def\b{\beta}
\def\hch{\widehat{\chi}}
\def\hchb{\widehat{\chib}}
\def\trch{\tr\chi}
\def\trchb{\tr\chib}
\def\trchc{\widecheck{\trch}}
\def\trchbc{\widecheck{\trchb}}

\def\de{\delta}
\def\De{\Delta}
\def\nab{\nabla}
\def\ov{\overline}

\def\bbb{\underline{b}}
\def\Gab{\Ga_b}
\def\Gag{\Ga_g}
\def\les{\lesssim}
\def\RR{\mr}

\def\slD{\slashed{D}}
\newcommand{\blue}{\textcolor{blue}}

\DeclareMathOperator{\curl}{curl}

\DeclareMathOperator{\grad}{grad}
\DeclareMathOperator{\supp}{supp}

\DeclareMathOperator{\tr}{tr}
\DeclareMathOperator{\sdiv}{div}
\def\bdiv{\mathbf{Div}}
\def\Ric{\mathbf{Ric}}

\def\gt{\widetilde{g}}
\def\kt{\widetilde{k}}

\def\etabf{\bm{\eta}}
\def\gBL{g_{BL}}
\newtheorem{thm}{Theorem}[section]
\newtheorem{prop}[thm]{Proposition}
\newtheorem{lem}[thm]{Lemma}
\newtheorem{cor}[thm]{Corollary}
\newtheorem{rk}[thm]{Remark}
\newtheorem{df}[thm]{Definition}

\newtheorem{que}[thm]{Question}
\newcommand{\red}{\textcolor{red}}
\allowdisplaybreaks[3]
\title{Formation of multiple Black Holes from Cauchy data}
\author{Dawei Shen and Jingbo Wan}
\begin{document}
\maketitle
\begin{abstract}
We construct a family of asymptotically flat Cauchy initial data for the Einstein vacuum equations that contain no trapped surfaces, yet whose future development admits multiple causally independent trapped surfaces. Assuming the weak cosmic censorship conjecture, this implies the formation of multiple black holes in finite time. The initial data are obtained by surgically modifying a vacuum geometrostatic manifold equipped with the Brill-Lindquist metric: the data agree exactly with the Brill-Lindquist metric outside a collection of balls centered at its multiple poles, and each ball is replaced by a constant-time slice of a well-prepared dynamical spacetime. The construction is based on Christodoulou's short-pulse framework, a stability analysis in the finite future of the short-pulse region, the geometry of geometrostatic manifolds with small mass and large separation, and the obstruction-free annular gluing method developed by Mao-Oh-Tao. The absence of trapped surfaces in the initial data is verified via a standard mean curvature comparison argument.

{\centering\subsubsection*{\small Keywords}}
\small\noindent short-pulse method; Minkowski stability; obstruction-free gluing; Brill-Lindquist metric; $N$--body problem.
\end{abstract}
\tableofcontents
\section{Introduction}
\subsection{Einstein vacuum constraint equations}
Let $(\Si,g)$ be a $3$--dimensional Riemannian manifold and $k$ be a symmetric 2-tensor on $\Si$. The \emph{Einstein vacuum constraint equations} take the form:
\begin{align}
    \begin{split}\label{constrainteq}
        R(g)+(\tr_g k)^2 - |k|^2_g &= 0,\\
        \sdiv_g k - \nab \tr_g k &= 0,
    \end{split}
\end{align}
where $\nab$ denotes the Levi-Civita connection of $g$, and $R(g)$ its scalar curvature. By the local existence theorem of Choquet-Bruhat and Choquet-Bruhat-Geroch \cite{cb,cbg}, any triplet $(\Si, g, k)$ satisfying \eqref{constrainteq} gives rise to a unique, maximal global hyperbolic development $(\M, \g)$ solving the Einstein vacuum equations:
\begin{equation}\label{EVE}
    \Ric(\g)_{\mu\nu} = 0,
\end{equation}
in which $(\Si,g)$ embeds isometrically, with $k$ as its second fundamental form.

The global behavior of $(\M,\g)$ depends sensitively on the nature of the initial data. For small perturbations of Minkowski space, the groundbreaking work of Christodoulou and Klainerman \cite{Ch-Kl} proved that all future-directed causal geodesics are complete, implying the absence of singularities. In contrast, for large data, gravitational collapse can lead to the formation of \emph{trapped surfaces}, i.e. spacelike $2$--surfaces whose outgoing and incoming null expansions are both negative.

A key result in this context is the celebrated incompleteness theorem of Penrose:
\begin{thm}[Penrose \cite{Penrose}]\label{Penrosesingularity}
Let $(\M, \g)$ be a spacetime solving \eqref{EVE}, with a non-compact Cauchy hypersurface. If $\M$ contains a compact trapped surface, then $(\M, \g)$ is future causally geodesically incomplete.
\end{thm}
Assuming the \emph{weak cosmic censorship conjecture},\footnote{See \cite{Chr99,PenroseWCC}; see also Conjecture 1 in \cite{An24}. The conjecture asserts that for generic asymptotically flat initial data, the maximal development of Einstein’s field equations possesses a complete future null infinity $\II^+$ and hides any formed singularities within a black hole region causally disconnected from $\mathcal{I}^+$.} such singularities must lie within a black hole region, invisible from null infinity. This connects the problem of black hole formation in general relativity to the dynamical formation of trapped surfaces. While the formation of a single trapped surface has been established in various settings, the possibility of dynamically forming \emph{multiple} trapped surfaces---each causally disjoint from the others---has remained a major open problem in mathematical relativity.

In this work, we construct the first example of smooth, asymptotically flat vacuum initial data for which future development contains multiple causally disjoint trapped surfaces. This provides a rigorous framework for studying multi-centered gravitational collapse and offers new insights into the $N$--body problem in general relativity.
\subsection{State of the art}
This work builds on three major advances in mathematical general relativity: the dynamical formation of trapped surfaces, the stability of Minkowski spacetime, and the construction of initial data via geometric gluing methods. By combining these tools, we construct smooth, asymptotically flat vacuum initial data whose future development exhibits multiple causally disjoint trapped surfaces, a phenomenon not previously realized in the vacuum setting. In what follows, we review the key results in each of these areas that form the foundation of our construction.
\subsubsection{Trapped surface formation}
A major breakthrough in the study of Einstein vacuum equations \eqref{EVE} was achieved by Christodoulou \cite{Chr}, who showed that trapped surfaces can form dynamically from regular initial data in vacuum. His short-pulse method demonstrated that a sufficiently concentrated burst of gravitational radiation leads to the formation of a trapped surface.

This method was further developed by Klainerman-Rodnianski \cite{kr} via parabolic rescaling and later generalized by An \cite{An12}. An-Luk \cite{AnLuk} established a \emph{scale-critical threshold} for the formation of trapped surfaces, identifying a severe condition for collapse. An \cite{An} also provided a simplified proof using a signature function $s_2$. In a different direction, Klainerman-Luk-Rodnianski \cite{KLR} introduced a fully anisotropic mechanism that relaxes the lower bounds required in Christodoulou’s original setup. More recently, Chen-Klainerman \cite{Chen-K} proposed a new approach using a non-integrable frame.

These results are mostly formulated on null hypersurfaces and analyzed via double null foliations. They all rigorously establish the formation of a \emph{single} trapped surface from localized incoming gravitational waves. In contrast, Li-Yu \cite{LY} constructed asymptotically flat Cauchy data with precise asymptotics and no trapped surfaces initially, and showed that the resulting spacetime contains a trapped surface. This represents one of the few results demonstrating trapped surface formation from genuinely spacelike data.

In contrast to dynamical formation, there are striking results on the \emph{existence} of black holes based purely on the Riemannian geometry of the initial data. Schoen-Yau \cite{SYbh} showed that sufficient matter condensation on a spacelike hypersurface under the dominant energy condition implies the existence of a marginally outer trapped surface (MOTS), via blow-up of solutions to the Schoen-Yau-Jang equation. Yau \cite{Yau02} extended this to cases with lower bounds on boundary mean curvature. See also \cite{HKKZ} for generalizations to higher dimensions. These results show that scalar and mean curvature bounds alone can guarantee black hole regions without reference to evolution.

Trapped surface formation has also been studied in the Einstein equations coupled with matter fields (see, e.g. \cite{AnAth,AMY,LL,ShenWan,Yucmp,Yu,Zhao}). In our recent work \cite{ShenWan}, we established a scale-critical result for the Einstein-Maxwell-charged scalar field system, showing a charging process and the formation of a trapped surface without any symmetry assumptions. The $|u|^p$--weighted estimate method developed there will also be used in the present paper.

All of the above works concern the formation of a \emph{single} trapped surface. In this paper, we construct asymptotically flat Cauchy initial data for Einstein vacuum equations that are initially free of trapped surfaces but whose evolution leads to the formation of \emph{multiple}, causally independent trapped surfaces. Our approach combines the short-pulse method with a multi-localized, geometric, obstruction-free gluing construction. This extends the scale-critical analysis of gravitational collapse to a new regime featuring spatially separated collapse regions.

To the best of our knowledge, this is the \emph{first rigorous result} demonstrating the formation of multiple trapped surfaces from a single evolution of vacuum Cauchy data. Our construction provides a new geometric framework to the study of $N$–body problems in general relativity, where multiple localized gravitational interactions are essential.

\subsubsection{Stability of Minkowski}
For initial data prescribed on a spacelike hypersurface, the study of \eqref{EVE} in the small data regime has been highly successful. In 1993, Christodoulou and Klainerman \cite{Ch-Kl} proved the global stability of Minkowski for the Einstein vacuum equations, a milestone in the domain of mathematical general relativity.

In 2003, Klainerman and Nicol\`o \cite{kn} gave a second proof of this result in the exterior of an outgoing null cone using double null foliation. Bieri \cite{Bieri} gave a new proof of the global stability of Minkowski requiring one less derivative and fewer vector fields compared to \cite{Ch-Kl}. Lindblad-Rodnianski \cite{lr1,lr2} gave a new proof of the stability of Minkowski using wave coordinates and showing that the Einstein equations verify the so-called weak null structure in that gauge. Huneau \cite{huneau} proved the nonlinear stability of Minkowski spacetime with a translation Killing field using generalized wave coordinates. Using the framework of Melrose’s $b$--analysis, Hintz-Vasy \cite{hv} reproved the stability of Minkowski spacetime. Later, Graf \cite{graf} proved the global stability of Minkowski in the context of the spacelike-characteristic Cauchy initial data. Using the $r^p$--weighted estimates introduced by Dafermos-Rodnianski \cite{Da-Ro}, \cite{Shen22} reproved the stability of Minkowski in the exterior region. In the framework of \cite{hv}, Hintz \cite{Hintz} reproved the stability of Minkowski in the exterior region. More recently, the stability of Minkowski under very weak decay and regularity assumptions has been proved in \cite{Shen23,Shen24}.

There are also numerous Minkowski stability results for Einstein's equations coupled with matter fields. We refer the reader to the introductions of \cite{Shen22,Shen23,Smulevici} for further results on the stability of Minkowski spacetime.

In this work, we apply the systematic $r^p$--weighted estimates framework developed in \cite{Shen22} for the stability of Minkowski spacetime to construct a transition region to the future of the short-pulse region, designed to facilitate a subsequent gluing construction. A similar idea of preparing a gluing-compatible transition region also appears in the work of Li–Yu \cite{LY} in the context of forming a single trapped surface from Cauchy data.
\subsubsection{Construction of initial data}
The Einstein vacuum constraint equations \eqref{constrainteq} form a nonlinear undetermined system for the initial data $(g,k)$. A central challenge in their analysis is to construct physically meaningful solutions with the desired geometric or asymptotic properties. Multiple frameworks have been developed for this purpose.

A classical approach to solving the constraint equations \eqref{constrainteq} is the conformal method, introduced by Lichnerowicz \cite{Lich}. However, as shown by Bartnik \cite{Bartnik}, elliptic estimates degenerate in weighted Sobolev spaces with strong decay, leading to large cokernels. To address this challenge, Fang–Szeftel–Touati \cite{FST, FSTKerr} developed a localized construction within the conformal method framework, based on a reduction to transverse-traceless perturbations and a nonlinear fixed point scheme. Their method allows for small deformations of Minkowski and Kerr slices, while carefully managing the analytic difficulties posed by linear obstructions from the cokernels of the underlying elliptic operators.

To overcome these difficulties in a more flexible way, gluing methods were introduced. Corvino’s pioneering work \cite{Cor} and its extensions by Corvino and Schoen \cite{CSglue} produce new initial data by gluing a general interior solution to a model solution (e.g. Kerr) in an exterior annular region. This framework was used by Li-Yu \cite{LY} to construct initial data for the formation of a single trapped surface, where a prescribed asymptotically flat end is matched with a carefully prepared interior. Carlotto and Schoen \cite{CS} later introduced a conical gluing method that allows for nontrivial deformations supported within cones, thereby enabling localized yet non-compactly supported modifications of the initial data and greatly increasing geometric flexibility. This approach was further refined by Mao-Tao \cite{MaoTao} using conic-type operators, inspired in part by solution operator techniques developed in the Yang-Mills setting by Oh and Tataru \cite{OhTataru}. However, these constructions are sensitive to the linearized obstructions of the underlying equations and remain rigid, making them insufficiently flexible to accommodate more general or geometric configurations.

Earlier contributions to multi-localized constructions in the spacelike Cauchy setting include the work of Chruściel–Corvino–Isenberg \cite{CCI}, who developed a method to construct $N$–body initial data by gluing multiple asymptotically flat ends in vacuum. More recently, Anderson–Corvino–Pasqualotto \cite{ACP} constructed a class of time-symmetric initial data sets modeling multiple almost isolated sources of gravitational radiation. Their construction uses a gluing scheme based on the Brill–Lindquist metric introduced in \cite{BL} and produces families of data sets that remain well controlled under large spatial separation.

A significant advance came from obstruction-free gluing. In the characteristic setting, Czimek-Rodnianski \cite{CR} developed a new approach to the gluing problem that exploits the nonlinear structure of the Einstein equations to bypass the 10-dimensional space of obstructions previously known in both spacelike and null settings. This builds on earlier developments in null gluing by Aretakis \cite{Are1, Are2} and a series of works with Czimek and Rodnianski \cite{ACR23, ACR2, ACR3}, which still operated within the constraints of linear obstructions. A recent major application of this earlier gluing framework is Kehle-Unger \cite{KU}, in which characteristic gluing is used to construct black hole spacetimes that settle to extremal Reissner-Nordstr\"om after finite advanced time, thereby disproving the third law of black hole thermodynamics. The obstruction-free result of \cite{CR} goes further, showing that exact charge matching is not necessary and offering significantly more flexibility, provided sufficient control over global geometric quantities.

In the Cauchy setting, Mao-Oh-Tao \cite{MOT} developed a purely spacelike, obstruction-free gluing method that enables controlled gluing in annular regions between two vacuum initial data sets. A key technical innovation in their work is the construction of Bogovskii-type solution operators with prescribed support-extending the classical result of Bogovskii \cite{Bogovskii}, which provides a right inverse to the divergence operator. These tools enable controlled gluing in annular regions. Their main result shows that, under smallness and sharp nonlinear positivity assumptions on the mismatch of ADM energy-momentum vectors, two localized initial data sets can be glued while prescribing the full set of physical charges up to controlled error. The construction is sharp in regularity and decay and enjoys stability under perturbation. Unlike the more rigid Corvino-Schoen \cite{CSglue} framework used in \cite{LY}, the method in \cite{MOT} allows multiple insertions of geometric structures with nearly independent asymptotics. See also the recent work of Hintz in \cite{Hin1, Hin2, Hin3, Hin4, Hin5} for an alternative microlocal approach to gluing problems.

Our construction of Cauchy data for the formation of multiple black holes is based on the space-like obstruction-free gluing framework of Mao-Oh-Tao \cite{MOT}. In contrast to earlier work such as \cite{LY}, which was limited to producing a single trapped surface by matching one interior to a well-chosen Kerr metric, the method of \cite{MOT} provides the necessary flexibility to insert multiple localized geometric structures, each with independently controlled asymptotics.
\subsection{Statement of the main theorem}
A simple version of our main theorem is stated as follows.
\begin{thm}\label{maintheoremintro}
Let $N\in\mathbb{N}$ and let $\gBL$ be the Brill-Lindquist metric defined by
\begin{align}\label{gBLintro}
    \gBL:=\left(1+\sum_{I=1}^N\frac{m_I}{2|\x-\cb_I|}\right)^4e,
\end{align}
where $e$ denotes the Euclidean metric. Assume that $\gBL$ satisfies the following small-mass and large-separation condition:
\begin{equation}\label{Mdintro}
    M:=\sum_{I=1}^Nm_I\ll 1,\qquad\quad d_I:=\min_{J\ne I}|\cb_I-\cb_J|\gg m_I^{-1},\quad\forall\; I\in\{1,2,\dots,N\}.
\end{equation}
Then, there exists a $2N$--parameter $\{(\de_I,a_I)\}_{I=1}^N$ family of triplets $(\Si,g,k)$ solving \eqref{constrainteq}, free of trapped surfaces, and satisfying the following properties for all $I\in\{1,2,\dots,N\}$:
\begin{enumerate}
\item We have
\begin{align*}
    (g,k)&=(\gBL,0)\qquad \mbox{ in }\;B_{32}^c(\cb_I),\\
    (g,k)&=(e,0)\qquad\quad\,\mbox{ in }\;B_{1-2\de_I}(\cb_I),
\end{align*}
where $B_r(\cb_I):=\{\x/\,|\x-\cb_I|<r\}$ and $0<\de_I\ll 1$ is a fixed constant.
\item The annulus $B_1(\cb_I)\setminus\ov{B_{1-2\de_I}(\cb_I)}$ is called a \emph{short-pulse annulus}, which is obtained from the evolution of a short-pulse characteristic initial data of length $\de_I$ and size $a_I$, see Theorems \ref{shortpulsecone} and \ref{interiorsolution}.
\item The annulus $B_2(\cb_I)\setminus\ov{B_{1}(\cb_I)}$ is called a \emph{barrier annulus}, which is obtained by taking a constant-time slice in the \emph{transition region}, see Theorems \ref{mainstability} and \ref{interiorsolution}.
\item We have detailed controls of $(g-e,k)$ on the annulus $B_{64}(\cb_I)\setminus\ov{B_1(\cb_I)}$, see \eqref{64-1control}.
\item Trapped surfaces will form in $D^+(B_1(\cb_I))$.\footnote{Here, $D^+(B_1(\cb_I))$ denotes the future domain of dependence of $B_1(\cb_I)$.}
\end{enumerate}
See Figure \ref{IntroID} for a geometric illustration of $(\Si,g,k)$ in the particular case $N=3$.
\end{thm}
\begin{rk}
Since the balls $\{B_{1}(\cb_I)\}_{I=1,\ldots,N}$ are mutually disjoint, the corresponding trapped surfaces are causally unrelated. By the Penrose incompleteness theorem (Theorem~\ref{Penrosesingularity}), this implies the occurrence of multiple, causally independent instances of geodesic incompleteness within finite time. Assuming the weak cosmic censorship conjecture, we conclude that $N$ black holes must form within finite time.
\end{rk}
\begin{rk}
Theorem \ref{maintheoremintro} shows that for any $N\in\mathbb{N}$, there exists a family of triplets $(\Si,g,k)$, which solves \eqref{constrainteq} and evolves to $N$ trapped surfaces. The proof also applies for the formation of countably many trapped surfaces under the following small-mass and large-separation condition:
\begin{equation}\label{Mdnumintro}
    M:=\sum_{I=1}^\infty m_I\ll 1,\qquad\quad d_I:=\min_{J\ne I}|\cb_I-\cb_J|\gg m_I^{-1},\quad \forall\; I.
\end{equation}
See Remark \ref{rkcountable} for more explanations.
\end{rk}
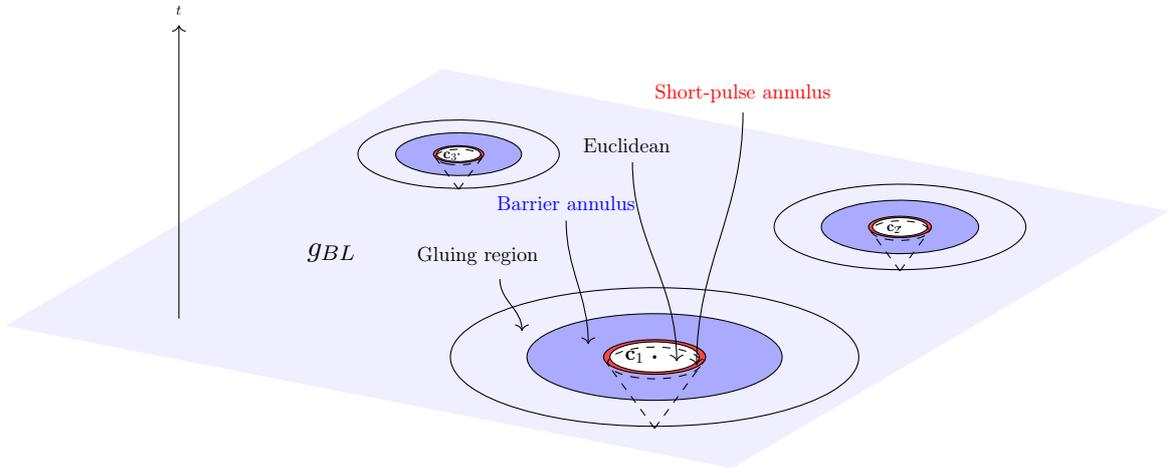
\begin{figure}[H]
\centering
\tdplotsetmaincoords{70}{120} 
\begin{tikzpicture}[scale=0.46, tdplot_main_coords]
\fill[blue!20, opacity=0.3] (-13,-12,0) -- (12,-12,0) -- (12,12,0) -- (-13,12,0) -- cycle; 
\node at (2,-7,0) {\scalebox{1}{$\gBL$}};
\newcommand{\AnnulusUnit}[4]{%
    \begin{scope}[shift={(#1,#2,0)}, scale=#3]
\begin{scope}
    \fill[blue!40, opacity=0.8] (0,0,0) ellipse (5 and 5);
    \clip (0,0,0) ellipse (2 and 2);
    \fill[white] (0,0,0) ellipse (2 and 2); 
\end{scope}
\begin{scope}
    \fill[red!90, opacity=0.8] (0,0,0) ellipse (2 and 2);
    \clip (0,0,0) ellipse (1.75 and 1.75);
    \fill[white] (0,0,0) ellipse (1.75 and 1.75); 
\end{scope}
\fill[black] (0,0,0) circle (2pt);
\node[left] at (0,0,0) {\scalebox{0.7}{$\cb_{#4}$}};
\draw (0,0,0) ellipse (2 and 2); 
\draw (0,0,0) ellipse (1.75 and 1.75); 
\draw (0,0,0) ellipse (5 and 5);
\draw (0,0,0) ellipse (8 and 8);
\draw[->, thin, rounded corners=4pt] (0,-1,8) to[out=-90, in=90] (0,1,0);
\node[above] at (0.1,-1.2,8) {\scalebox{0.7}{Euclidean}};
\draw[->, thin, rounded corners=4pt] (0,4,11) to[out=-90, in=90] (0,1.9,0);
\node[above] at (0,4,11) {\scalebox{0.7}{\red{Short-pulse annulus}}};
\draw[->, thin, rounded corners=4pt] (0,-4,5) to[out=-90, in=90] (0,-3,0);
\node[above] at (0,-4,5) {\scalebox{0.7}{\blue{Barrier annulus}}};
\draw[->, thin, rounded corners=4pt] (0,-7,2) to[out=-90, in=90] (0,-6,0);
\node[above] at (0,-8,2) {\scalebox{0.7}{Gluing region}};
\draw[dashed] (0,0,-0.25) ellipse (2*11/12 and 2*11/12); 
\foreach \angle in {-40,110}
{\draw[dashed] (0,0,-3) -- ({2*cos(\angle)},{2*sin(\angle)},0);}
\end{scope}
}
\newcommand{\AnnulusUnitplain}[4]{%
    \begin{scope}[shift={(#1,#2,0)}, scale=#3]
\begin{scope}
    \fill[blue!40, opacity=0.8] (0,0,0) ellipse (5 and 5);
    \clip (0,0,0) ellipse (2 and 2);
    \fill[white] (0,0,0) ellipse (2 and 2); 
\end{scope}
\begin{scope}
    \fill[red!90, opacity=0.8] (0,0,0) ellipse (2 and 2);
    \clip (0,0,0) ellipse (1.75 and 1.75);
    \fill[white] (0,0,0) ellipse (1.75 and 1.75); 
\end{scope}
\fill[black] (0,0,0) circle (2pt);
\node[left] at (0,0.8,0) {\scalebox{0.5}{$\cb_{#4}$}};
\draw (0,0,0) ellipse (2 and 2); 
\draw (0,0,0) ellipse (1.75 and 1.75); 
\draw (0,0,0) ellipse (5 and 5);
\draw (0,0,0) ellipse (8 and 8);
\draw[dashed] (0,0,-0.25) ellipse (2*11/12 and 2*11/12); 
\foreach \angle in {-40,110}
{\draw[dashed] (0,0,-3) -- ({2*cos(\angle)},{2*sin(\angle)},0);}
\end{scope}
}
\AnnulusUnit{5}{5.4}{0.73}{1}
\AnnulusUnitplain{-8}{6}{0.45}{2}
\AnnulusUnitplain{-7}{-8}{0.36}{3}
 \draw[->] (9,-8,0) -- (9,-8,9) node[anchor=south]{\scalebox{0.5}{$t$}};
\end{tikzpicture}
\caption{$(\Si,g,k)$ in Theorem \ref{maintheoremintro}.}
\label{IntroID}
\end{figure}
\subsection{Outline of the proof of the main theorem}
We now sketch the heuristics of the proof of Theorem \ref{maintheoremintro}. The proof is divided into the following $4$ main intermediary steps. The first two steps are similar to \cite{LY}, in which we construct suitable annuli. The annuli will play the role of interior solutions. The third step is to deduce the main properties of the Brill-Lindquist metric \eqref{gBLdfintro}, which plays the role of the exterior solution. The last step is to use the spacelike obstruction-free gluing established in \cite{MOT} to construct the desired initial data.
\subsubsection{Trapped surface formation in a short-pulse cone}
The first step is to prove an analog result of \cite{AnLuk,Chr,kr}, see Theorem \ref{shortpulseconeintro} below. We introduce a double null foliation $(u,\ub)$ where $u$ and $\ub$ are optical functions. We use $H_u$ (resp. $\Hb_\ub$) to denote the outgoing (resp. incoming) null hypersurfaces that are the level sets of $u$ (resp. $\ub$). We also denote $S_{u,\ub}:=H_u\cap\Hb_\ub$.\footnote{See Section \ref{doublenullfoliation} for more details of the double null foliation and the geometry quantities.}
\begin{thm}\label{shortpulseconeintro}
    There exists a sufficiently large $a_0>0$. Let $a>a_0$, $0<\de\leq a^{-1}$ and $u_0\in[-2,-1]$. Assume that there is an initial data on $H_{u_0}\cup\Hb_{0}$ that satisfies:
    \begin{itemize}
    \item The null shear satisfies $\hch\simeq\af$ along $H_{u_0}^{(0,1)}$.
    \item Minkowskian initial data along $\Hb_0$.
    \end{itemize}
Then \eqref{EVE} admits a unique solution in the short-pulse region $\{(u,\ub)\in[u_0,-\frac{\de a}{4}]\times[0,\de]\}$. Moreover, we have:
\begin{itemize}
    \item The $2$--sphere $S_{-\frac{\de a}{4},\de}$ is a trapped surface.
    \item The following estimates hold for $u\in[u_0,-\frac{1}{2}]$:
\begin{align*}
    |\a|&\les\de^{-1}\afd,\qquad |\b|\les\af, \qquad |\rho,\si|\les\de a,\qquad |\bb|\les\de^2a^\frac{3}{2},\qquad |\aa|\les\de^3a^2,\\
    |\hch|&\les\af,\qquad\qquad|\om|\les 1,\qquad\qquad\quad\;\;\;|(\eta,\etab,\ze,\hchb)|\les\de\af,\qquad\;\,|\omb|\les\de^2a,\\
    |\log\Om|&\les\de,\qquad\qquad\; |\bbb^A|\les\de^2\af,\qquad\quad\;\;\; |\slg_{AB}-\slgo_{AB}|\les\de\af,\\
    |\trchc|&\les\de a,\qquad\quad\, |\trchbc|\les\de.
\end{align*}
Moreover, the same estimates also hold for their higher order derivatives. See Figure \ref{3Dshortpulseconeintro} for a geometric illustration.
\end{itemize}
\end{thm}
    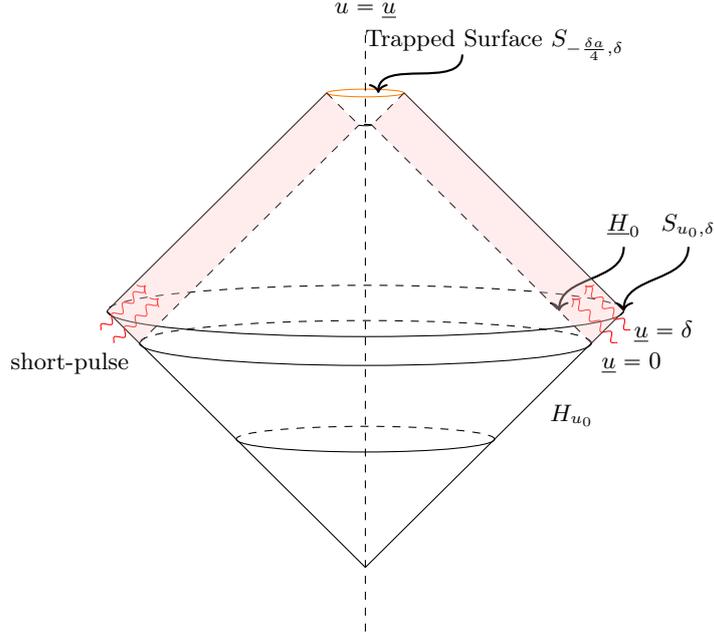
\begin{figure}[H]
    \centering
\begin{tikzpicture}[scale=0.85, decorate]
    \draw[orange] (0,3.4) ellipse (0.6 and 0.06);
    \draw[->, thick, rounded corners=8pt] (1.5,4) to[out=-90, in=90] (0.2,3.4);
    \node[above] at (2,3.75) {\footnotesize Trapped Surface $S_{-\frac{\de a}{4},\de}$};
    \draw[dashed] (0.1,2.9) arc[start angle=0,end angle=180,x radius=0.1,y radius=0.01];
    \draw (-0.1,2.9) arc[start angle=180,end angle=360,x radius=0.1,y radius=0.01];
    \draw[dashed] (4,0) arc[start angle=0,end angle=180,x radius=4,y radius=0.4];
    \draw (-4,0) arc[start angle=180,end angle=360,x radius=4,y radius=0.4];
\begin{scope}
    \fill[white] 
        (3,0) arc[start angle=0,end angle=180,x radius=3,y radius=0.3] --
        (-3,0) arc[start angle=180,end angle=360,x radius=3,y radius=0.3] -- cycle;
\end{scope}
    \draw[dashed] (3.5,-0.5) arc[start angle=0,end angle=180,x radius=3.5,y radius=0.35];
    \draw (-3.5,-0.5) arc[start angle=180,end angle=360,x radius=3.5,y radius=0.35];
    \draw[dashed] (2,-2) arc[start angle=0,end angle=180,x radius=2,y radius=0.2];
    \draw (-2,-2) arc[start angle=180,end angle=360,x radius=2,y radius=0.2];
    \draw[dashed] (0,-5) -- (0,4.4) node[above]{\footnotesize $u=\ub$};
    \draw (0,-4) -- (4,0);
    \draw (-4,0) -- (0,-4);
    \draw (4,0) -- (0.6,3.4);
    \draw (-4,0) -- (-0.6,3.4);
    \draw[dashed] (3.5,-0.5) -- (0.1,2.9);
    \draw[dashed] (0.6,3.4) -- (0.1,2.9);
    \draw[dashed] (-0.6,3.4) -- (-0.1,2.9);
    \draw[dashed] (-3.5,-0.5) -- (-0.1,2.9);
    \draw[->, thick, rounded corners=8pt] (4,1) to[out=-90, in=90] (3,0);
    \node[above] at (4,1) {\footnotesize $\Hb_0$};
    \draw[->, thick, rounded corners=8pt] (5,1) to[out=-90, in=90] (4,0);
    \node[above] at (5,1) {\footnotesize $S_{u_0,\de}$};
    \node[below right] at (4,0) {\footnotesize $\ub=\de$};
    \node[below right] at (3.5,-0.5) {\footnotesize $\ub=0$};
    \node[below right] at (2.7,-1.3) {\footnotesize $H_{u_0}$};
\draw[->, decorate, decoration={snake, amplitude=0.5mm, segment length=2mm}, thin, red]  (4.1,-0.3) -- (3.4,0.4);
\draw[->, decorate, decoration={snake, amplitude=0.5mm, segment length=2mm}, thin, red]  (3.9,-0.5) -- (3.2,0.2);
\draw[->, decorate, decoration={snake, amplitude=0.5mm, segment length=2mm}, thin, red]  (-4.1,-0.3) -- (-3.4,0.4);
\draw[->, decorate, decoration={snake, amplitude=0.5mm, segment length=2mm}, thin, red]  (-3.9,-0.5) -- (-3.2,0.2);
\node[below left] at (-3.5,-0.5)  {\footnotesize short-pulse};
    \fill[red!20,opacity=0.35](3.5,-0.5) -- (0.1,2.9) -- (0.6,3.4) -- (4,0) -- cycle;
    \fill[red!20,opacity=0.35](-3.5,-0.5) -- (-0.1,2.9) -- (-0.6,3.4) -- (-4,0) -- cycle;
\end{tikzpicture}
\caption{Short-pulse cone. The red region is called the short-pulse region where the existence result and estimates are established; the orange circle denotes the trapped surface $S_{-\frac{\de a}{4},\de}$.}
\label{3Dshortpulseconeintro}
\end{figure}
Theorem \ref{shortpulseconeintro} is restated as Theorem \ref{shortpulsecone} in Section \ref{sectrapped}. We first prove a scale-invariant trapped surface formation result, similar to \cite{An}. This is done by applying the $|u|^p$--weighted estimates introduced in \cite{ShenWan} to the Bianchi equations and integrating the null structure equations along the outgoing and incoming null cones. We then apply a standard rescaling argument to prove Theorem \ref{shortpulseconeintro}.
\begin{rk}
In \cite{LY}, the vectorfield method has been used to prove the semi-global existence of the short-pulse region. Here, we provide a simpler self-contained proof of trapped surface formation by the signature $s_2$ introduced in \cite{An} and the $|u|^p$--weighted estimates developed in \cite{ShenWan}. 
\end{rk}
\subsubsection{Stability result in the transition region}
In order to construct a spacelike initial data which evolves to trapped surfaces, we expect to take a constant-time slice in the short-pulse cone constructed in Theorem \ref{shortpulseconeintro}. However, the abnormal behavior in the short-pulse region makes it difficult to apply the gluing theorem in \cite{MOT}. To this end, we proceed as in Section 4 of \cite{LY} to construct the transition region. This is done by proving the following stability result.
\begin{thm}\label{mainstabilityintro}
    There exists a sufficiently large $a_0>0$. Let $a>a_0$, $0<\de\leq a^{-2}$. Assume that there are initial data on $H_{-\frac{3}{2}+\de}\cup\Hb_{0}$ that satisfy: \begin{itemize}
    \item $H_{-\frac{3}{2}+\de}^{(0,1)}$ is endowed with an outgoing geodesic foliation.
    \item The null shear satisfies $\hch\simeq\af$ along $H_{-\frac{3}{2}+\de}^{(0,\de)}$ and vanishes along $H_{-\frac{3}{2}+\de}^{(\de,1)}$.
    \item Minkowskian initial data along $\Hb_0$.
    \end{itemize}
    Then \eqref{EVE} admits a unique solution in $\{(u,\ub)\in[-\frac{3}{2}+\de,-\frac{1}{2}]\times[0,1]\}$. Moreover
    \begin{itemize}
        \item All quantities are controlled by $a^{-1}$ in the transition region $\{(u,\ub)\in[u_0,-\frac{1}{2}]\times[\de,1]\}$.
        \item The constant-time slice $\Si:=\{u+\ub=-1+2\de\}$ satisfies
        \begin{equation}\label{g-eest}
            \|(g-e,k)\|_{C^2(\Si,e)}\les a^{-1}.
        \end{equation}
        \item Denoting
        \begin{equation*}
            B_r:=\{p\in\Si/\, |r(p)|<r\},\qquad\quad r:=\ub-u,
        \end{equation*}
        we call $B_1\setminus\ov{B_{1-2\de}}$ the \emph{short-pulse annulus} and $B_2\setminus\ov{B_1}$ the \emph{barrier annulus}.
    \end{itemize}
    See Figure \ref{fig:shortpulse+stabintro} for a geometric illustration. 
\end{thm}
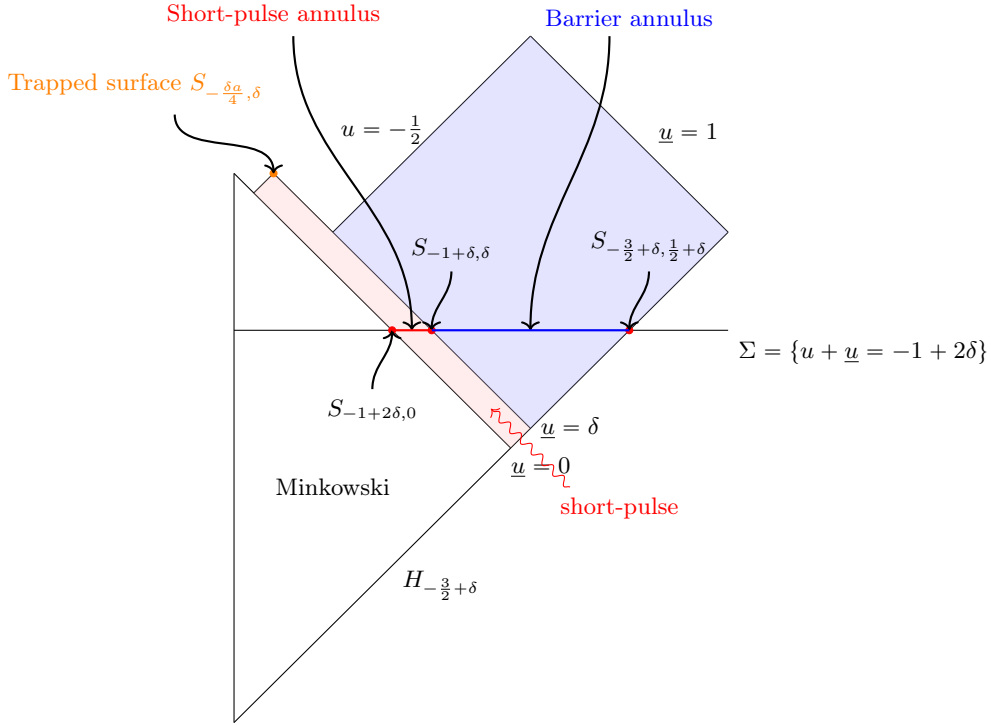
\begin{figure}[H]
        \centering
        \begin{tikzpicture}[scale=2.6, decorate]
  \coordinate (A) at (0,-2);
  \coordinate (B) at (0,0);
  \coordinate (C) at (0,0.8);
  \coordinate (D) at (0.1,0.7);
  \coordinate (E) at (0.2,0.8);
  \coordinate (F) at (1.4,-0.6);
  \coordinate (G) at (1.5,-0.5);
  \coordinate (H) at (2,0);
  \coordinate (I) at (2.5,0.5);
  \coordinate (J) at (1,0);
  \coordinate (K) at (0.5,0.5);
  \coordinate (L) at (1.5,1.5);
  \coordinate (M) at (2.5,0);
  \coordinate (N) at (0.8,0);
  \draw (A) -- (C);
  \draw (A) -- (I);
  \draw (F) -- (C);
  \draw (D) -- (E);
  \draw (G) -- (E);
  \draw (K) -- (L);
  \draw (I) -- (L);
  \draw (B) -- (M);
\fill[red!20,opacity=0.35](G) -- (E) -- (D) -- (F) -- cycle;
\fill[blue!30,opacity=0.35](G) -- (I) -- (L) -- (K) -- cycle;
  \node[right] at (1.35,-0.7) {\footnotesize $\ub = 0$};
  \node[right] at (G) {\footnotesize $\ub = \de$};
  \node[above] at (2.3,0.9) {\footnotesize $\ub = 1$};
  \node[above] at (0.75,0.9) {\footnotesize $u = -\frac{1}{2}$};
  \node[below right] at (M) {\footnotesize $\Si=\{u+\ub=-1+2\de\}$};
  \node at (0.5, -0.8) {\footnotesize Minkowski};
  \node[right] at (0.8, -1.3) {\footnotesize $H_{-\frac{3}{2}+\de}$};
  \node[below right] at (1.6,-0.8) {\footnotesize \red{short-pulse}};
\filldraw[red] (J) circle (0.5pt);
\filldraw[red] (H) circle (0.5pt);
\filldraw[red] (N) circle (0.5pt);
\filldraw[orange] (E) circle (0.5pt);
    \draw[->, thick, rounded corners=8pt] (1.1,0.3) to[out=-90, in=90] (J);
    \node[above] at (1.1,0.3) {\footnotesize $S_{-1+\de,\de}$};
    \draw[->, thick, rounded corners=8pt] (0.7,-0.3) to[out=90, in=-90] (N);
    \node[below] at (0.7,-0.3) {\footnotesize $S_{-1+2\de,0}$};
   \draw[->, thick, rounded corners=8pt] (2.1,0.3) to[out=-90, in=90] (H);
   \node[above] at (2.1,0.3) {\footnotesize $S_{-\frac{3}{2}+\de,\frac{1}{2}+\de}$};
    \draw[->, thick, rounded corners=8pt] (-0.3,1.1) to[out=-90, in=90] (E);
    \node[above] at (-0.5,1.1) {\footnotesize {\color{orange}Trapped surface $S_{-\frac{\de a}{4},\de}$}};
    \draw[->, thick, rounded corners=8pt] (0.3,1.5) to[out=-90, in=90] (0.9,0);
    \node[above] at (0.2,1.5) {\footnotesize \red{Short-pulse annulus}};
    \draw[->, thick, rounded corners=8pt] (1.9,1.5) to[out=-90, in=90] (1.5,0);
    \node[above] at (2,1.5) {\footnotesize \blue{Barrier annulus}};
    \draw[->, decorate, decoration={snake, amplitude=0.5mm, segment length=2mm}, thin, red]  (1.7,-0.8) -- (1.3,-0.4);
    \draw[draw=blue, thick] (J) -- (H);
    \draw[draw=red, thick] (N) -- (J);
\end{tikzpicture}
        \caption{The red region marks the short-pulse domain, and the blue region indicates the transition zone. The hypersurface $\Si$ represents the target constant-time slice. The red line-segment denotes the short-pulse annulus, while the blue line-segment corresponds to the barrier annulus.}
        \label{fig:shortpulse+stabintro}
    \end{figure}
Theorem \ref{mainstabilityintro} follows from Theorems \ref{mainstability} and \ref{interiorsolution}, which are proved in Section \ref{secstability}. The proof can be considered as a finite-region version of the stability of Minkowski. Based on the $r^p$--weighted estimates introduced by Dafermos-Rodnianski \cite{Da-Ro} and applied to Minkowski stability in \cite{Shen22,Shen23,Shen24}, we provide a simple proof of Theorem \ref{mainstabilityintro} in this finite transition region applying the particular case $p=0$.
\begin{rk}
    The initial data on $\Hb_\de$, obtained as a consequence of Theorem \ref{shortpulseconeintro}, is in the small data regime. In fact, even if the null shear satisfies $\hch\simeq\af$ in the short-pulse region $\{(u,\ub)\in[-\frac{3}{2}+\de,-\frac{1}{2}]\times [0,\de]\}$. We have from Theorem \ref{shortpulseconeintro}
    \begin{equation}\label{nab3hchsmall}
        \nabs_3\hch=O(\de\af)=O(a^{-1}).
    \end{equation}
    Recalling that $\hch$ vanishes along $H_{u_0}^{(\de,1)}$, we deduce by integrating \eqref{nab3hchsmall} along $\Hb_\de$.
    \begin{equation*}
        |\hch|\les a^{-1},\qquad\mbox{ on }\;\Hb_\de^{(-\frac{3}{2}+\de,-\frac{1}{2})}.
    \end{equation*}
    The smallness of other abnormal quantities on $\Hb_\de$ can be deduced in a similar way.
\end{rk}
\begin{rk}
    In \cite{LY}, an additional condition on $\hch$ is introduced to prove that the metric $\g$ in the transition region is close to a Schwarzschild metric $\g_{m_0}$. This is necessary to apply the gluing theorem of Corvino-Schoen \cite{CSglue}, see Section 5 in \cite{LY}. Here, in Theorem \ref{mainstabilityintro} we prove that in the transition region, the metric $\g$ is close to Minkowski with an error $a^{-1}$. This is sufficient to apply the obstruction-free gluing theorem of Mao-Oh-Tao \cite{MOT}, see Section \ref{secgluing} for more details.
\end{rk}
\subsubsection{Vacuum geometrostatic manifolds}
The constant-time slice $\Si$ obtained in Theorem \ref{mainstabilityintro} will evolve to trapped surfaces in finite time. In order to construct Cauchy data evolving to multiple trapped surfaces, we should find multi-centered solutions of the constraint equations \eqref{constrainteq}.
\begin{thm}\label{BrillLindquistintro}
    We introduce the following Brill-Lindquist metric:
    \begin{align}\label{gBLdfintro}
        \gBL:=\left(1+\sum_{I=1}^N\frac{m_I}{2|\x-\cb_I|}\right)^4e,
    \end{align}
which satisfies the small-mass and large-separation condition:
\begin{equation}\label{smallMlarged}
M:=\sum_{I=1}^Nm_I\ll 1,\qquad\quad d_I:=\min_{J\ne I}|\cb_I-\cb_J|\gg m_I^{-1},\quad\forall\; I\in\{1,2,\dots,N\}.
\end{equation}
Then, for any fixed parameters $N\in\mathbb{N}$, $m_I>0$ and $\cb_I\in\RRR^3$, $(\gBL,0)$ is a solution of the constraint equations \eqref{constrainteq} on $\RRR^3\setminus\{\cb_I\}_{I=1}^N$. The Riemannian manifold $(\RRR^3\setminus \{\cb_I\}_{I=1}^N,\gBL)$ is called a \emph{vacuum geometrostatic manifold}. Moreover, we have for any $I\in\{1,2,\dots,N\}$:
\begin{itemize}
\item The averaged charges satisfy:\footnote{Remark that the energy $\E$, the momentum $\P$, the center of mass $\C$, and the angular momentum $\J$ are defined in Section \ref{sseccharges}.}
\begin{align*}
\E\left[(\gBL,0);A_{32}^{(I)}\right]&=8\pi m_I+O(m_IM+Md_I^{-1}),\\
\P_l\left[(\gBL,0);A_{32}^{(I)}\right]&=0,\\
\C_l\left[((\gBL,0);A_{32}^{(I)}\right]&=(\cb_I)_{l}\big(8\pi m_I+O(m_IM+Md_I^{-1})\big)+O(m_IM+Md_I^{-1}),\\
\J_l\left[(\gBL,0);A_{32}^{(I)}\right]&=0,
\end{align*}
with $l=1,2,3$ and
\begin{equation*}
   A_{32}^{(I)}:=A_{32}\left(\cb_I\right)=B_{64}\left(\cb_I\right)\setminus\overline{B_{32}\left(\cb_I\right)}.
\end{equation*}
\item The following estimate holds for any $s\in\mathbb{N}$:
\begin{equation}\label{eq:Sobolev-boundintro}
\|\gBL-e\|_{H^{s}(A_{32}^{(I)})}\les m_I.
\end{equation}
\end{itemize}
\end{thm}
Theorem \ref{BrillLindquistintro} follows from Corollary \ref{cor:BL-annulus} and Proposition \ref{prop:BL-Sobolev}, which are proved in Section \ref{secBL}. The proof is done by expanding the explicit form \eqref{gBLdfintro} of the Brill-Lindquist metric around the poles $\{\cb_I\}_{I=1}^N$, see Section \ref{secBL} for detailed computations.
\subsubsection{Gluing construction of the Cauchy data}
Using the vacuum geometrostatic manifold $(\RRR^3\setminus\{\cb_I\}_{I=1}^N, \gBL)$ introduced in Theorem \ref{BrillLindquistintro} and the constant-time slice $\Si$ constructed in Theorem \ref{mainstabilityintro}. We now provide a sketch of the proof of Theorem \ref{maintheoremintro}.
\begin{enumerate}
    \item Let $(\Si,g,k):=(\Si_{\de,a},g_{\de,a},k_{\de,a})$ be the constant-time slice constructed by Theorem \ref{mainstabilityintro} with parameters $(\de,a)$. Then, we have from \eqref{g-eest}
    \begin{equation}\label{interiorQintro}
        |\Q[(g,k);A_1]|\les a^{-1},\qquad A_{1}:=B_2\setminus\ov{B_1},
    \end{equation}
    where $\Q[(g,k);A_1]$ is defined in \eqref{eq:charge-avg}. We also have from \eqref{g-eest}
    \begin{equation}\label{interiorsobolevintro}
        \|(g-e,k)\|^2_{H^2\times H^1(A_1)}\les a^{-2}.
    \end{equation}
    \item Let $\gBL$ be the Brill-Lindquist metric defined in \eqref{gBLintro}. For any $I\in\{1,2,\dots,N\}$, let $(x^1,x^2,x^3)$ be a coordinate system centered at $\cb_I$.\footnote{In other words, we assume that $\cb_I=\bm{0}$ in this coordinate system.} Then, we have from Theorem \ref{BrillLindquistintro} and \eqref{Mdintro}
    \begin{align}
    \begin{split}\label{exteriorintro}
        \E\left[(\gBL,0);A_{32}^{(I)}\right]&=8\pi m_I+ O(m_IM+Md_I^{-1})\simeq m_I,\\
        \left|(\P,\C,\J)\left[(\gBL,0);A_{32}^{(I)}\right]\right|&\les m_IM+Md_I^{-1}\ll m_I,\\
        \|\gBL-e\|_{H^s(A_{32}^{(I)})}&\les m_I.
    \end{split}
    \end{align}
    \item For $I=1$, we define
    \begin{equation*}
        (g_{in},k_{in}):=(g_{\de_1,a_1},k_{\de_1,a_1}),\qquad\quad (g_{out},k_{out}):=(\gBL,0),
    \end{equation*}
    where $(\Si_{\de_1,a_1},g_{\de_1,a_1},k_{\de_1,a_1})$ is the constant-time slice obtained by applying Theorem \ref{shortpulseconeintro} centered at $\cb_1$ and the parameters $(\de_1,a_1)$ satisfy $a_1\gg m_1^{-1}\gg 1$ and $0<\de_1\leq a_1^{-2}$. Denoting
    \begin{equation*}
        \De\Q^{(1)}:=\Q\left[(g_{out},k_{out});A_{32}^{(1)}\right]-\Q\left[(g_{in},k_{in});A_1^{(1)}\right],
    \end{equation*}
     we obtain from \eqref{interiorQintro} and \eqref{exteriorintro}
     \begin{align*}
         |\De\P,\De\C,\De\J|\ll \De \E,\qquad\quad \De\E\ll 1.
     \end{align*}
     We also have from \eqref{interiorsobolevintro} and \eqref{exteriorintro}
     \begin{align*}
          \|(g-e,k)\|^2_{H^2\times H^1(A_1^{(1)})}+\|(\gBL-e,0)\|^2_{H^2\times H^1(A_{32}^{(1)})}\les m_I^2+a_I^{-2}\ll m_I.
     \end{align*}
     Hence, all the hypotheses in Theorem 1.7 of \cite{MOT} are valid\footnote{See Theorem \ref{thm:MOT1.7} for the restatement of Theorem 1.7 in \cite{MOT}.}, we deduce that there exists $(g^1,k^1)$, which solves \eqref{constrainteq} and satisfies
     \begin{align*}
         (g^1,k^1)&=(g_{\de_1,a_1},k_{\de_1,a_1}),\qquad \mbox{ in } B_2(\cb_1),\\
         (g^1,k^1)&=(\gBL,0),\qquad \qquad\;\,\mbox{ in } B_{32}^c(\cb_1).
     \end{align*}
     \item Proceeding as in the previous steps, we deduce by induction that, for any $J\in\{1,2,\dots,N\}$, there exists $(g^J,k^J)$ and $2J$ parameters $\{(\de_I,a_I)\}_{I=1}^J$, such that
     \begin{itemize}
         \item $(g^J,k^J)$ solves \eqref{constrainteq} in $\RRR^3\setminus \{\cb_I\}_{I=J+1}^N$.
         \item For any $I\in\{1,2,\dots,J\}$:
     \begin{align*}
         (g^J,k^J)&=(g_{\de_I,a_I},k_{\de_I,a_I}),\qquad \mbox{ in } B_2(\cb_I),\\
         (g^J,k^J)&=(\gBL,0),\qquad \qquad\;\,\mbox{ in } B_{32}^c(\cb_I).
     \end{align*}
     \end{itemize}
Taking $(g,k)=(g^N,k^N)$, we obtain the desired Cauchy data.
\item We have by construction that $N$ trapped surfaces will form in $\bigcup\limits_{I=1}^ND^+(B_1(\cb_I))$. Moreover, $(\Si,g,k)$ can be shown to be free of trapped surfaces by a standard mean curvature comparison lemma, see Proposition \ref{notrapped}. This concludes the proof of Theorem \ref{maintheoremintro}.
\end{enumerate}
\subsection{Notations and conventions}
Throughout the paper, we use the following notation and conventions.
\begin{itemize}
\item We have the following types of manifolds in this paper: $4$--dimensional spacetime $\M$, $3$--dimensional spacelike hypersurface $\Si$, and $2$--spheres $S$. Every manifold has its metric, Levi-Civita connection and curvature tensor:
\begin{align*}
 (\M,\g,\D,\R),\qquad\quad(\Si,g,\nab,R),\qquad\quad (S,\slg,\nabs,\K),
\end{align*}
where $\K$ is the Gauss curvature of $S$.
\item We use capital Latin letters $A,B,C,...$ to denote an index from $1$ to $2$, lowercase Latin letters $i,j,k...$ to denote an index from $1$ to $3$, Greek letters $\a,\b,\ga,\la,\mu,\nu...$ to denote an index from $1$ to $4$, and capital Latin letters $I,J,...$ to denote an index from $1$ to $N$.
\end{itemize}
\subsubsection{Key parameters}
The following parameters will be frequently used throughout this paper:
\begin{itemize}
    \item $s\in\NNN$ denotes the regularity of the Cauchy data, which is fixed at the beginning.
    \item $N\in\NNN$ denotes the number of poles of the Brill-Lindquist metric \eqref{gBLintro}.
    \item $\{m_I\}_{I=1}^N\in\RRR^N$ denote the mass parameters of \eqref{gBLintro}. The total mass $M$ is defined as follows:
    \begin{equation*}
        M:=\sum_{I=1}^N m_I.
    \end{equation*}
    \item $\{\cb_I\}_{I=1}^N\in\RRR^{3N}$ denote the positions of the poles of \eqref{gBLintro}.
    \item $d>0$ denotes the minimal distance between the poles:
    \begin{align*}
        d:=\min_{1\leq I<J\leq N}|\cb_I-\cb_J|.
    \end{align*}
    \item $\{a_I\}_{I=1}^N\in\RRR^N$ denotes the size of short-pulse  near $\cb_I$.
    \item $\{\de_I\}_{I=1}^N\in\RRR^N$ denotes the length of the short-pulse posed on the characteristic initial data near $\cb_I$.
\end{itemize}
\subsubsection{Smallness constants}
For any quantities $A$ and $B$:
\begin{itemize}
    \item We write $A\les B$ to indicate that there exists a universal constant $C(s)>0$ that depends only on the regularity parameter $s$ such that $A\leq C(s) B$.
    \item We write $A=O(B)$ to indicate that $|A|\les |B|$.
    \item We write $A\simeq B$ to indicate that $|A|\les |B|\les |A|$.
    \item We write $A\ll B$ to indicate that $CA<B$ where $C$ is the largest universal constant among all the constants involved in the paper by $\les$.
\end{itemize}
Throughout this paper, the constants are chosen such that for all $I\in\{1,2,\dots,N\}$:
\begin{align}\label{smallconstant}
0<\de_I\leq a_I^{-2},\qquad\quad a_I^{-1},d_I^{-1}\ll m_I\leq M\ll 1.
\end{align}
\subsection{Structure of the paper}
The remainder of this paper is organized as follows.
\begin{itemize}
    \item In Section \ref{secpre}, we first recall the fundamental notions of double null foliation and state the null structure equations and the Bianchi equations. We then state the obstruction-free gluing theorem for almost flat annuli, which is established in \cite{MOT}. Finally, we recall the definitions of trapped surfaces and marginally outer trapped surfaces (MOTS).
    \item In Section \ref{sectrapped}, we first provide a self-contained proof of a scale-critical result of trapped surface formation. This has been done by applying the $|u|^p$--weighted estimates introduced in \cite{ShenWan} to the Bianchi equations and integrating the null structure equations along the outgoing and incoming null cones. We then apply a standard rescaling argument to deduce a trapped surface formation result in a finite region, which will be used in Section \ref{secstability}.
    \item In Section \ref{secstability}, we prove a stability result in the finite future of the short-pulse region constructed in Section \ref{sectrapped}. The new region obtained is called the transition region. Taking a constant-time slice in the transition region, we construct a barrier annulus, which will serve as the interior annuli in the gluing construction of the Cauchy data.
    \item In Section \ref{secBL}, we first recall the definitions of the vacuum geometrostatic manifold and the Brill--Lindquist metric introduced in \cite{BL}. We then compute the mean curvature of the coordinate spheres, evaluate the local ADM charges, and estimate the Sobolev norms on annular regions centered at the poles. These annuli will serve as the outer annuli regions in the gluing construction of the Cauchy data.
    \item In Section \ref{secgluing}, we apply the celebrated obstruction-free gluing theorem in \cite{MOT} to construct the desired Cauchy data. We show that the Cauchy data will evolve to multiple trapped surfaces. We also show that the Cauchy data are free of trapped surfaces.
\end{itemize}
\subsection{Acknowledgments}
The authors would like to thank Sergiu Klainerman for introducing the problem of the formation of multiple trapped surfaces and for many valuable discussions. The authors also thank Xinliang An for insightful discussions on the short-pulse method.
\section{Preliminaries}\label{secpre}
\subsection{Geometric setup}\label{geometricsetup}
We first introduce the geometric setup of double null foliation, which will be used in Sections \ref{sectrapped} and \ref{secstability}.
\subsubsection{Double null foliation}\label{doublenullfoliation}
We denote $(\M,\g)$ a spacetime $\M$ with the Lorentzian metric $\g$ and $\D$ its Levi-Civita connection. Let $u$ and $\ub$ be two optical functions on $\M$, that is
\begin{equation*}
    \g(\grad u,\grad u)=\g(\grad\ub,\grad\ub)=0.
\end{equation*}
The spacetime $\M$ is foliated by the level sets of $u$ and $\ub$ respectively, and the functions $u,\ub$ are required to increase towards the future. We use $H_u$ to denote the outgoing null hypersurfaces which are the level sets of $u$ and use $\Hb_\ub$ to denote the incoming null hypersurfaces which are the level sets of $\ub$. We denote
\begin{equation}
    S(u,\ub):=H_u\cap\Hb_\ub,
\end{equation}
which are spacelike $2$--spheres. We introduce the vectorfields $L$ and $\Lb$ by
\begin{equation*}
    L:=-\grad u,\qquad\quad \Lb:= -\grad\ub.
\end{equation*}
We define a positive function $\Om$ by the formula
\begin{equation}\label{deflapse}
    \g(L,\Lb)=-2\Om^{-2},
\end{equation}
where $\Om$ is called the lapse function. We then define the normalized null pair $(e_3,e_4)$ by
\begin{equation*}
   e_3=\Om\Lb,\qquad\quad e_4=\Om L, 
\end{equation*}
and define another null pair by
\begin{equation*}
    \Nb=\Om e_3, \qquad\quad N=\Om e_4.
\end{equation*}
On a given $2$--sphere $S(u,\ub)$, we choose a local orthonormal frame $(e_1,e_2)$, we call $(e_1,e_2,e_3,e_4)$ a null frame.
\subsubsection{Principal quantities}
The spacetime metric $\g$ induces a Riemannian metric $\slg$ on $S(u,\ub)$. We use $\nabs$ to denote the Levi-Civita connection of $\slg$ on $S(u,\ub)$. Using $(u,\ub)$, we introduce a local coordinate system $(u,\ub,\phi^A)$ on $\M$ with $e_4(\phi^A)=0$. In that coordinates system, the metric $\g$ takes the form:
\begin{equation}\label{metricg}
\g=-2\Om^2 (d\ub\otimes du+du\otimes d\ub)+\slg_{AB}(d\phi^A-\bbb^Adu)\otimes (d\phi^B-\bbb^B du),
\end{equation}
and we have 
\begin{equation*}
   \underline{N}=\pr_u+\bbb,\qquad N=\pr_{\ub},\qquad \bbb:=\bbb^A \pr_{\phi^A}.
\end{equation*}
We recall the null decomposition of the Ricci coefficients and curvature components of the null frame $(e_1,e_2,e_3,e_4)$ as follows:
\begin{align}
\begin{split}\label{defga}
\chib_{AB}&=\g(\D_A e_3, e_B),\qquad\quad \chi_{AB}=\g(\D_A e_4, e_B),\\
\xib_A&=\frac 1 2 \g(\D_3 e_3,e_A),\qquad\quad\,\,  \xi_A=\frac 1 2 \g(\D_4 e_4, e_A),\\
\omb&=\frac 1 4 \g(\D_3e_3 ,e_4),\qquad\quad \,\,\,\,\;\om=\frac 1 4 \g(\D_4 e_4, e_3), \\
\etab_A&=\frac 1 2 \g(\D_4 e_3, e_A),\qquad\quad\;\eta_A=\frac 1 2 \g(\D_3 e_4, e_A),\\
\ze_A&=\frac 1 2 \g(\D_{e_A}e_4, e_3),
\end{split}
\end{align}
and
\begin{align}
\begin{split}\label{defr}
\a_{AB} &=\R(e_A, e_4, e_B, e_4) , \qquad \,\;\aa_{AB} =\R(e_A, e_3, e_B, e_3), \\
\b_{A} &=\frac 1 2 \R(e_A, e_4, e_3, e_4), \qquad\,\, \;\bb_{A}=\frac 1 2 \R(e_A, e_3, e_3, e_4),\\
\rho&= \frac 1 4 \R(e_3, e_4, e_3, e_4), \qquad\,\;\;\,\;\; \si =\frac{1}{4}{^*\R}( e_3, e_4, e_3, e_4),
\end{split}
\end{align}
where $^*\R$ denotes the Hodge dual of $\R$. The null second fundamental forms $\chi, \chib$ are further decomposed in their traces $\trch$ and $\trchb$, and traceless parts $\hch$ and $\hchb$:
\begin{align*}
\trch&:=\de^{AB}\chi_{AB},\qquad\quad \,\hch_{AB}:=\chi_{AB}-\frac{1}{2}\de_{AB}\trch,\\
\trchb&:=\de^{AB}\chib_{AB},\qquad\quad \, \hchb_{AB}:=\chib_{AB}-\frac{1}{2}\de_{AB}\trchb.
\end{align*}
We define the horizontal covariant operator $\nabs$ as follows:
\begin{equation*}
\nabs_X Y:=\D_X Y-\frac{1}{2}\chib(X,Y)e_4-\frac{1}{2}\chi(X,Y)e_3.
\end{equation*}
We also define $\nabs_4 X$ and $\nabs_3 X$ to be the horizontal projections:
\begin{align*}
\nabs_4 X&:=\D_4 X-\frac{1}{2} \g(X,\D_4e_3)e_4-\frac{1}{2} \g(X,\D_4e_4)e_3,\\
\nabs_3 X&:=\D_3 X-\frac{1}{2} \g(X,\D_3e_3)e_3-\frac{1}{2} \g(X,\D_3e_4)e_4.
\end{align*}
A tensor field $\psi$ defined in $\MM$ is called tangent to $S$ if it is defined a priori in spacetime $\M$ and all possible contractions of $\psi$ with $e_3$ or $e_4$ are zero. We use $\nabs_3 \psi$ and $\nabs_4 \psi$ to denote the projection to $S(u,\ub)$ of the usual derivatives $\D_3\psi$ and $\D_4\psi$.
\begin{prop}
The following identities hold for the double null foliation introduced in Section \ref{doublenullfoliation}:
\begin{align}
\begin{split}\label{nullidentities}
    \nabs\log\Om&=\frac{1}{2}(\eta+\etab),\qquad\;\;\;\,\om=-\frac{1}{2}\D_4(\log\Om), \qquad\;\;\;\,\omb=-\frac{1}{2}\D_3(\log\Om),\\ 
    \eta&=\ze+\nabs\log\Om,\qquad \etab=-\zeta+\nabs\log\Om,\qquad\quad\, \xi=\xib=0.
\end{split}
\end{align}
\end{prop}
\begin{proof}
    See (6) in \cite{kr}.
\end{proof}
\subsection{Hodge systems}\label{ssec7.2}
\begin{df}\label{tensorfields}
For tensor fields defined on a $2$--sphere $S$, we denote by $\sk_0:=\sk_0(S)$ the set of pairs of scalar functions, $\sk_1:=\sk_1(S)$ the set of $1$--forms and $\sk_2:=\sk_2(S)$ the set of symmetric traceless $2$--tensors.
\end{df}
\begin{df}\label{def7.2}
Given $\xi\in\sk_1$, we define its Hodge dual
\begin{equation*}
    {^*\xi}_A := \ins_{AB}\xi^B.
\end{equation*}
Given $U \in \sk_2$, we define its Hodge dual
\begin{equation*}
{^*U}_{AB}:=\ins_{AC} {U^C}_B.
\end{equation*}
\end{df}
\begin{df}
    Given $\xi,\eta\in\sk_1$ and $U,V\in\sk_2$, we denote
\begin{align*}
    \xi\c \eta&:= \slg^{AB}\xi_A \eta_B,\qquad \xi\wedge \eta:= \ins^{AB} \xi_A \eta_B,\qquad (\xi\hot \eta)_{AB}:=\xi_A \eta_B +\xi_B \eta_A -\slg_{AB}\xi\c\eta,\\
    (\xi\c U)_A&:= \slg^{BC} \xi_B U_{AC},\qquad\qquad\qquad\qquad\qquad\,\,(U\wedge V)_{AB}:=\ins^{AB}U_{AC}V_{CB}.
\end{align*}
\end{df}
\begin{df}
    For a given $\xi\in\sk_1$, we define the following differential operators:
    \begin{align*}
        \sdivs \xi&:= \slg^{AB} \nabs_A\xi_B,\\
        \curls \xi&:= \ins^{AB} \nabs_A \xi_B,\\
        (\nabs\hot\xi)_{AB}&:=\nabs_A \xi_B+\nabs_B \xi_A-\slg_{AB}(\sdivs\xi).
    \end{align*}
\end{df}
\begin{df}\label{hodgeop}
    We define the following Hodge type operators:
    \begin{itemize}
        \item $\sld_1$ takes $\sk_1$ into $\sk_0$ and is given by:
        \begin{equation*}
            \sld_1 \xi :=(\sdivs\xi,\curls \xi),
        \end{equation*}
        \item $\sld_2$ takes $\sk_2$ into $\sk_1$ and is given by:
        \begin{equation*}
            (\sld_2 U)_A := \nabs^B U_{AB}, 
        \end{equation*}
        \item $\sld_1^*$ takes $\sk_0$ into $\sk_1$ and is given by:
        \begin{align*}
            \sld_1^*(f,f_*)_{A}:=-\nabs_A f +{\ins_A}^B \nabs_B f_*,
        \end{align*}
        \item $\sld_2^*$ takes $\sk_1$ into $\sk_2$ and is given by:
        \begin{align*}
            \sld_2^*\xi:=-\frac{1}{2}\nabs\hot\xi.
        \end{align*}
    \end{itemize}
\end{df}
We have the following identities:
\begin{align}
    \begin{split}\label{dddd}
        \sld_1^*\sld_1&=-\Des_1+\mathbf{K},\qquad\qquad \sld_1\sld_1^*=-\Des_0,\\
        \sld_2^*\sld_2&=-\frac{1}{2}\Des_2+\mathbf{K},\qquad\quad\; \sld_2 \sld_2^*=-\frac{1}{2}(\Des_1+\mathbf{K}).
    \end{split}
\end{align}
where $\mathbf{K}$ denotes the Gauss curvature on $S$, see (2.2.2) in \cite{Ch-Kl}. 
\begin{df}\label{Lpnorms}
For a tensor field $f$ on a $2$--sphere $S$, we denote its $L^p(S)$--norm:
\begin{equation}
    \|f\|_{L^p(S)}:= \left(\int_S |f|^p \right)^\frac{1}{p}.
\end{equation}
\end{df}
\begin{prop}\label{ellipticLp}
Assume that $S$ is an arbitrary compact 2-surface with positive bounded Gauss curvature and area radius $r$. Then the following statements hold:
\begin{enumerate}
\item Let $\phi\in\sk_0$ be a solution of $\Des \phi=f$. Then we have
\begin{align*}
    \|\nabs^2 \phi\|_{L^2(S)}+r^{-1}\|\nabs\phi\|_{L^2(S)}+r^{-2}\|\phi-\overline{\phi}\|_{L^2(S)}\les \|f\|_{L^2(S)}.
\end{align*}
\item Let $\xi\in\sk_1$ be a solution of $\sld_1\xi=(f,f_*)$. Then we have
\begin{align*}
    \|\nabs\xi\|_{L^2(S)}+r^{-1}\|\xi\|_{L^2(S)}\les \|(f,f_*)\|_{L^2(S)}.
\end{align*}
\item Let $U\in\sk_2$ be a solution of $\sld_2 U=f$. Then we have
\begin{align*}
\|\nabs U\|_{L^2(S)}+r^{-1}\|U\|_{L^2(S)}\les\|f\|_{L^2(S)}.
\end{align*}
\end{enumerate}
\end{prop}
\begin{proof}
See Corollary 2.3.1.1 in \cite{Ch-Kl}.
\end{proof}
\subsection{Main equations}\label{secmaineq}
We summarize the main equations for a double null foliation.
\subsubsection{Null structure equations}\label{sec-nullstr}
\begin{prop}\label{nulles}
We have the following null structure equations:
\begin{align}
\begin{split}\label{basicnull}
\nabs_4\eta&=-\chi\c(\eta-\etab)-\b,\\
\nabs_3\etab&=-\chib\c(\etab-\eta)+\bb,\\
\nabs_4\hch+(\trch)\hch&=-2\om\hch-\a,\\
\nabs_4\trch+\frac{1}{2}(\trch)^2&=-|\hch|^2-2\om\trch,\\
\nabs_3\hchb+(\trchb)\hchb&=-2\omb\hchb-\aa,\\
\nabs_3\trchb+\frac{1}{2}(\trchb)^2&=-|\hchb|^2-2\omb\trchb,\\
\nabs_4\hchb+\frac{1}{2}(\trch)\hchb&=\nabs\hot\etab+2\om\hchb-\frac{1}{2}\trchb\,\hch+\etab\hot\etab,\\
\nabs_3\hch+\frac{1}{2}(\trchb)\hch&=\nabs\hot\eta+2\omb\hch-\frac{1}{2}\trch\,\hchb+\eta\hot\eta,\\
\nabs_4\trchb+\frac{1}{2}(\trch)\trchb&=2\om\trchb+2\rho-\hch\c\hchb+2\sdivs\etab+2|\etab|^2,\\
\nabs_3\trch+\frac{1}{2}(\trchb)\trch&=2\omb\trch+2\rho-\hch\c\hchb+2\sdivs\eta+2|\eta|^2.
\end{split}
\end{align}
We also have the Codazzi equations:
\begin{align}
\begin{split}\label{codazzi}
\sdivs\hch&=\frac{1}{2}\nabs\trch-\ze\c\left(\hch-\frac{1}{2}\trch\right)-\b,\\ 
\sdivs\hchb&=\frac{1}{2}\nabs\trchb+\ze\c\left(\hchb-\frac{1}{2}\trchb\right)+\bb,
\end{split}
\end{align}
the torsion equation:
\begin{equation}\label{torsion}
\curls\eta=-\curls\etab=\si-\frac{1}{2}\hch\wedge\hchb,
\end{equation}
and the Gauss equation:
\begin{equation}\label{gauss}
    \K=-\frac{1}{4}\trch\trchb+\frac{1}{2}\hch\c\hchb-\rho.
\end{equation}
Moreover, we have
\begin{align*}
\nabs_4\omb&=2\om\omb+\frac{3}{4}|\eta-\etab|^2-\frac{1}{4}(\eta-\etab)\c(\eta+\etab)-\frac{1}{8}|\eta+\etab|^2+\frac{1}{2}\rho,\\
\nabs_3\om&=2\om\omb+\frac{3}{4}|\eta-\etab|^2+\frac{1}{4}(\eta-\etab)\c(\eta+\etab)-\frac{1}{8}|\eta+\etab|^2+\frac{1}{2}\rho.
\end{align*}
\end{prop}
\begin{proof}
    See (3.1)--(3.5) of \cite{kr}.
\end{proof}
\subsubsection{Bianchi equations}\label{sec-Bianchi}
\begin{prop}\label{bianchiequations}
We have the following Bianchi equations:
\begin{align*}
\nabs_3\a+\frac{1}{2}\trchb\,\a&=\nabs\hot\b+4\omb\a-3(\hch\rho+{^*\hch}\si)+(\ze+4\eta)\hot\b,\\
\nabs_4\b+2\trch\,\b&=\sdivs\a-2\om\b+(2\ze+\etab)\c\a,\\
\nabs_3\b+\trchb\,\b&=\nabs\rho+{^*\nabs}\si+2\omb\b+2\hch\c\bb+3(\eta\rho+{^*\eta}\si),\\
\nabs_4\rho+\frac{3}{2}\trch\,\rho&=\sdivs\b-\frac{1}{2}\hchb\c\a+\ze\c\b+2\etab\c\b,\\
\nabs_4\si+\frac{3}{2}\trch\,\si&=-\curls\b+\frac{1}{2}\hchb\c{^*\a}-\ze\c{^*\b}-2\etab\c{^*\b},\\
\nabs_3\rho+\frac{3}{2}\trchb\,\rho&=-\sdivs\bb-\frac{1}{2}\hch\c\aa+\ze\c\bb-2\eta\c\bb,\\
\nabs_3\si+\frac{3}{2}\trchb\,\si&=-\curls\bb+\frac{1}{2}\hch\c{^*\aa}-\ze\c{^*\bb}-2\eta\c{^*\bb},\\
\nabs_4\bb+\trch\,\bb&=-\nabs\rho+{^*\nabs\si}+2\om\bb+2\hchb\c\b-3(\etab\rho-{^*\etab}\si),\\
\nabs_3\bb+2\trchb\,\bb&=-\sdivs\aa-2\omb\bb+(2\ze-\eta)\c\aa,\\
\nabs_4\aa+\frac{1}{2}\trch\,\aa&=-\nabs\hot\bb+4\om\aa-3(\hchb\rho-{^*\hchb}\si)+(\ze-4\etab)\hot\bb.
\end{align*}
\end{prop}
\begin{proof}
    See Proposition 3.2.4 in \cite{kn}.
\end{proof}
\subsection{Commutation identities}
We recall the following commutation formulae.
\begin{lem}\label{comm}
Let $U_{A_1...A_k}$ be an $S$-tangent $k$-covariant tensor on $(\M,\g)$. Then
\begin{align*}
    [\Om\nabs_4,\nabs_B]U_{A_1...A_k}&=-\Om\chi_{BC}\nabs_CU_{A_1...A_k}+\sum_{i=1}^k \Om(\chi_{A_iB}\,\etab_C-\chi_{BC}\,\etab_{A_i}+\ins_{A_iC}{^*\b}_B)U_{A_1...C...A_k},\\
    [\Om\nabs_3,\nabs_B]U_{A_1...A_k}&=-\Om\chib_{BC}\nabs_C U_{A_1...A_k}+\sum_{i=1}^k\Om(\chib_{A_iB}\,\eta_C-\chib_{BC}\,\eta_{A_i}+\ins_{A_iC}{^*\bb}_B)U_{A_1...C...A_k},\\
    [\Om\nabs_3,\Om\nabs_4]U_{A_1...A_k}&=4\Om^2\ze_B\nabs_B U_{A_1...A_k}+2\Om^2\sum_{i=1}^k(\etab_{A_i}\,\eta_C-\etab_{A_i}\,\eta_C+\ins_{A_iC}\si)U_{A_1...C...A_k}.
\end{align*}
\end{lem}
\begin{proof}
It is a direct consequence of Lemma 7.3.3 in \cite{Ch-Kl} and \eqref{nullidentities}.
\end{proof}
\subsection{Charges of Cauchy initial data}\label{sseccharges}
\begin{df}\label{chargesMOT}
Let $\Si\subseteq\RRR^3$ and let $(\Si,g,k)$ be a solution of \eqref{constrainteq} equipped with a coordinate system $(x^1,x^2,x^3)$. We denote
\begin{align*}
    B_r:=\left\{(x^1,x^2,x^3)/\, |x|<r\right\}, \qquad\mbox{ where }\quad   |x|:=\sqrt{(x^1)^2+(x^2)^2+(x^3)^2}.
\end{align*}
We introduce the following charges of $(g,k)$ measured on the sphere $\pr B_r$:
    \begin{align*}
        \E[(g,k);\pr B_r]&:=\frac{1}{2}\int_{\pr B_r}(\pr_ig_{ij}-\pr_jg_{ii})\nu^jdS,\\
        \P_i[(g,k);\pr B_r]&:=\int_{\pr B_r}(k_{ij}-\de_{ij}\tr_ek)\nu^jdS,\\
        \C_l[(g,k);\pr B_r]&:=\frac{1}{2}\int_{\pr B_r}\big(x_l\pr_ig_{ij}-x_l\pr_jg_{ii}-\de_{il}(g-e)_{ij}+\de_{jl}(g-e)_{ii}\big)\nu^jdS,\\
        \J_l[(g,k);\pr B_r]&:=\int_{\pr B_r}(k_{ij}-\de_{ij}\tr_e k)Y_l^i\nu^jdS,
    \end{align*}
where $e$ is the Euclidean metric and we denote
    \begin{equation*}
        \pr_i:=\pr_{x^i},\qquad\quad \nu:=\frac{x^j}{|x|}\pr_j,\qquad\quad Y_i:={\in_{ij}}^lx^j\pr_l.
    \end{equation*}
    We put these together to form a $10$--vector:
\begin{align}\label{dfQ}
    \Q[(g,k);\pr B_r]=(\E,\P_1,\P_2,\P_3,\C_1,\C_2,\C_3,\J_1,\J_2,\J_3)[(g,k);\pr B_r].
\end{align}
\end{df}
Since we will be working with annuli, it is also convenient to introduce the following averaged charges. 
\begin{df}\label{ADMannulus}
    Let $\eta \in C^{\infty}_{c}(0, \infty)$ be a cutoff function such that 
    \begin{equation}\label{cutoffeta}
    \supp \eta \subseteq (1, 2),\qquad \int_1^2\eta(r) d r = 1,
    \end{equation}
    and we denote for any $r>0$
    \begin{align}\label{scaleeta}
        \eta_{r}(r') := r^{-1} \eta(r^{-1} r').
    \end{align}
    We then define
\begin{align}\label{eq:charge-avg}
\Q[(g, k);A_r]=\int_r^{2r}\eta_{r}(r')\Q[(g, k);\pr B_{r'}]\, dr',
\end{align}
where $A_r:=B_{2r}\setminus \ov{B_r}$ denotes the annulus between the spheres $\pr B_{2r}$ and $\pr B_r$.
\end{df}
\subsection{Obstruction-free gluing}
We state in the following the main result of \cite{MOT}, which plays an essential role in the construction of the Cauchy data in our main theorem.
\begin{thm}[Obstruction-free gluing for annuli \cite{MOT}] \label{thm:MOT1.7}
Given $s > \frac{3}{2}$ and $\Ga>1$, there exist $\eps_{o} =\eps_{o}(s,\Ga) > 0$, $\mu_o=\mu_o(s,\Ga)>0$ and $C_{o} = C_{o}(s,\Ga)>0$ such that the following holds. Let $(g_{in}, k_{in}) \in H^s\times H^{s-1} (A_{1})$ and $(g_{out}, k_{out}) \in H^{s} \times H^{s-1} (A_{32})$ be solutions of the constraint equations \eqref{constrainteq}. We define $\De\Q = (\De\E,\De\P,\De\C,\De\J)\in\RRR^{10}$ by
\begin{equation} \label{eq:Delta-Q}
	\De\Q = \Q[(g_{out}, k_{out}); A_{32}]-\Q[(g_{in}, k_{in});A_1],
\end{equation}
and assume that
\begin{align}
\De\E &> |\De \P|, \label{eq:obs-free-unit:EP} \\
\frac{\De\E}{\sqrt{(\De\E)^{2}-|\De\P|^{2}}} &< \Ga, \label{eq:obs-free-unit:Gamma} \\
\De\E &< \ep_{o}^{2}, \label{eq:obs-free-unit:ep} \\
|\De\C|+|\De\J|&<\mu_{o}\De\E, \label{eq:obs-free-unit:CJ}
\end{align}
and
\begin{equation}
\|(g_{in}-e, k_{in})\|_{H^s\times H^{s-1}(A_1)}^{2}+\|(g_{out}-e,k_{out})\|_{H^{s} \times H^{s-1}(A_{32})}^{2}<\mu_{o}\De\E\label{eq:obs-free-unit:data}.
\end{equation}
Then there exists $(g, k) \in H^{s} \times H^{s-1}(B_{64} \setminus \overline{B_{1}})$ solving \eqref{constrainteq} such that
\begin{equation}
(g, k) = (g_{in}, k_{in}) \quad \mbox{ on } A_1, \qquad (g, k)=(g_{out}, k_{out}) \quad \mbox{ on } A_{32}, \label{eq:obs-free-gluing}
\end{equation}
and
\begin{equation}
    \|(g-e, k)\|_{H^{s} \times H^{s-1}(B_{64} \setminus\overline{B_{1}})}^{2} < C_{o} \De\E. \label{eq:obs-free-conc}
\end{equation}
\end{thm}
\subsection{Trapped surfaces and MOTS}
We introduce the following definitions of trapped surfaces.
\begin{df}\label{dftrapped2+2}
    Let $(\M,\g)$ be a spacetime endowed with a double null $(u,\ub)$--foliation defined in Section \ref{doublenullfoliation}. Then, with respect to the given double null $(u,\ub)$--foliation, a leaf $S_{u,\ub}\subseteq\M$ is called 
    \begin{itemize}
    \item a trapped surface if the following hold on $S_{u,\ub}$:
    \begin{equation}\label{trappeduub}
        \trch<0,\qquad\quad \trchb<0;
    \end{equation}
    \item a marginally outer trapped surface (MOTS) if the following hold on $S_{u,\ub}$:
    \begin{equation}\label{motsuub}
        \trch=0,\qquad\quad \trchb<0.
    \end{equation}
    \end{itemize}
\end{df}
\begin{df}\label{dftrappedsurface}
Let $g$ be a Riemannian metric on $\Si$ and $k$ be a symmetric covariant $2$--tensor on $\Si$. Let $S\subset\Si$ be a compact, embedded smooth $2$--surface. Let $\th$ be the second fundamental form of $S$. We will say that $S$ is
\begin{itemize}
\item a trapped surface if the following hold on $S$:
\begin{equation}\label{dftrapped}
    \tr_\slg(-\th-k)<0,\qquad\quad\tr_\slg(\th-k)<0;
\end{equation}
\item a marginally outer trapped surface (MOTS) if the following hold on $S$:
    \begin{equation}\label{dfMOTS}
    \tr_\slg(-\th-k)<0,\qquad\quad\tr_\slg(\th-k)<0.
\end{equation}
\end{itemize}
\end{df}
\begin{rk}
Definition \ref{dftrappedsurface} is consistent with \cite{Ch-Kl}. More precisely, we have the following identities:
\begin{align}\label{chichib}
    \chi=\th-k,\qquad\quad \chib=-\th-k,
\end{align}
as in (7.5.2b) in \cite{Ch-Kl}. Notice that the convention of the sign of $k$ is different from Definition 3.3 in \cite{LM} and Section 5.3 in \cite{LY}. 
\end{rk}
\section{Short-pulse cone and formation of trapped surfaces}\label{sectrapped}
In this section, we first prove Theorem \ref{maintrapped}, which provides a scale-critical result for trapped surface formation. We then apply a standard rescaling argument to deduce Theorem \ref{shortpulsecone}, which constructs a family of finite short-pulse cones endowed with regular characteristic initial data, which will evolve to trapped surfaces in its future.
\subsection{Short-pulse initial data}
Throughout this section, we always denote
\begin{align*}
    V:=V(u,\ub):=\big\{(u',\ub')\in\left[u_\infty,u\right]\times[0,\ub]\big\},\qquad V_*:=V\left(-\frac{a}{4},1\right),
\end{align*}
where $a>0$ is a sufficiently large constant and $u_\infty\leq -a$. As taken in \cite{An,Chr}, we prescribe $\hch$ so that
\begin{align*}
    \sum_{i+j\leq s+10}\left|\nabs_4^j(|u_\infty|\nabs)^i\hch\right|\simeq \frac{a^\frac{1}{2}}{|u_\infty|}\qquad \mbox{ along }\; H_{u_\infty}^{(0,1)},
\end{align*}
where $s\in\mathbb{N}$ is a parameter that describes the regularity of the initial data. Following the same procedures as in Chapter 2 of \cite{Chr}, we obtain the following estimates on $H_{u_\infty}^{(0,1)}$:
\begin{align*}
    |\a|&\les\frac{a^\frac{1}{2}}{|u_\infty|},\qquad |\b|\les\frac{a^\frac{1}{2}}{|u_\infty|^2}, \qquad |(\rho,\si)|\les\frac{a}{|u_\infty|^3},\qquad |\bb|\les\frac{a}{|u_\infty|^4},\qquad\quad |\aa|\les\frac{a^\frac{3}{2}}{|u_\infty|^5},\\
    |\om|&\les\frac{1}{|u_\infty|},\quad|\trchc|\les\frac{a}{|u|^2},\qquad\quad\;\, |\trchbc|\les\frac{1}{|u_\infty|^2},\quad\;\;\;\,|\omb|\les\frac{a}{|u_\infty|^3},\quad|\eta,\etab,\hchb|\les\frac{a^\frac{1}{2}}{|u_\infty|^2},
\end{align*}
where we denote
\begin{align}\label{dftrchctrchbc}
    \trchc:=\trch-\frac{2}{|u|},\qquad\quad \trchbc:=\trchb+\frac{2}{|u|}.
\end{align}
Moreover, the analogy of the above estimates also holds for derivatives of $\nabs_4^j(|u_\infty|\nabs)^i$ with $i+j\leq s+5$.
\subsection{Signatures and scale invariant norms}\label{secsignatures}
It will be difficult to treat the above weights $|u|$ and $a$ term by term. We hope to design a \emph{scale-invariant} $L^\infty_{sc}(S_{u,\ub})$ with built-in weights $|u|$ and $a$ such that for most geometric quantities $\phi$, we have
\begin{align*}
    \|\phi\|_{L^\infty_{sc}(S_{u,\ub})}\les 1.
\end{align*}
To achieve this, we proceed as in \cite{An}. We first relax the above estimates for $\bb$ and $\aa$
\begin{align*}
    |\bb|\les\frac{a}{|u_\infty|^4}\les\frac{a^\frac{3}{2}}{|u_\infty|^4},\qquad\quad |\aa|\les\frac{a^\frac{3}{2}}{|u_\infty|^5}\les\frac{a^2}{|u_\infty|^5}.
\end{align*}
Keeping the other estimates for now, we find a systematic way to define $L^\infty_{sc}(S_{u,\ub})$.
\begin{df}\label{signature}
    We first introduce the signature for the decay rates to $\phi$, we assign the signature $s_2(\phi)$ according to the rule:
    \begin{align*}
        s_2(\phi)=\frac{1}{2} N_A(\phi)+N_3(\phi)-1,
    \end{align*}
    with $N_3(\phi)$ the number of times $e_3$ appears in the definition of $\phi$ and $N_3(\phi)$ the number of times $e_A$ appears in the definition of $\phi$.
\end{df}
\begin{rk}
Following Definition \ref{signature}, we have the following signature table:
\begin{center}
\begin{tabular}{|c|c|c|c|c|c|c|c|c|c|c|c|c|c|c|c|}
\hline
{} & $\a$ & $\b$ & $\rho$ & $\si$ & $\bb$ & $\aa$ & $\chi$ & $\om$& $\Om$ & $\ze$ & $\eta$ & $\etab$ & $\chib$ & $\omb$ & $\slg$ \\
$s_2$ & $0$ & $0.5$ & $1$ & $1$ & $1.5$ & $2$ & $0$ & $0$ & $0$ & $0.5$ & $0.5$ & $0.5$ & $1$ & $1$ & $0$ \\
\hline
\end{tabular}
\end{center}
Moreover, we have
\begin{equation}\label{s1s2}
s_2(\nabs_4\phi)=s_2(\phi),\qquad s_2(\nabs\phi)=s_2(\phi)+\frac{1}{2},\qquad s_2(\nabs_3\phi)=s_2(\phi)+1.
\end{equation}
\end{rk}
For any horizontal tensorfield $\phi$ with signature $s_2(\phi)$, we define the scale invariant norms:
\begin{align}
    \begin{split}\label{dfsc}
        \|\phi\|_{L^\infty_{sc}(S_{u,\ub})}&:=a^{-s_2(\phi)}|u|^{2s_2(\phi)+1}\|\phi\|_{L^\infty(S_{u,\ub})},\\
        \|\phi\|_{L^2_{sc}(S_{u,\ub})}&:=a^{-s_2(\phi)}|u|^{2s_2(\phi)}\|\phi\|_{L^2(S_{u,\ub})},\\
        \|\phi\|_{L^1_{sc}(S_{u,\ub})}&:=a^{-s_2(\phi)}|u|^{2s_2(\phi)-1}\|\phi\|_{L^1(S_{u,\ub})}.
    \end{split}
\end{align}
For convenience, we also define the scale-invariant norms along null hypersurfaces
\begin{align}
\begin{split}\label{fluxsc}
    \|\phi\|_{L^2_{sc}(\cuv)}^2&:=\int_0^\ub\|\phi\|_{L^2_{sc}(S_{u,\ub'})}^2d\ub',\\
    \|\phi\|_{L^2_{sc}(\ucuv)}^2&:=\int_{u_\infty}^u\frac{a}{|u'|^2}\|\phi\|_{L^2_{sc}(S_{u',\ub})}^2du'.
\end{split}
\end{align}
The following proposition is an immediate consequence of \eqref{dfsc} and the H\"older inequality.
\begin{prop}\label{Holder}
We have the following inequalities:
    \begin{align*}
    \|\phi_1\c\phi_2\|_{L^2_{sc}(S_{u,\ub})}&\leq\frac{1}{|u|}\|\phi_1\|_{L^\infty_{sc}(S_{u,\ub})}\|\phi_2\|_{L^2_{sc}(S_{u,\ub})},\\
    \|\phi_1\c\phi_2\|_{L^1_{sc}(S_{u,\ub})}&\leq\frac{1}{|u|}\|\phi_1\|_{L^\infty_{sc}(S_{u,\ub})}\|\phi_2\|_{L^1_{sc}(S_{u,\ub})},\\
    \|\phi_1\c\phi_2\|_{L^1_{sc}(S_{u,\ub})}&\leq\frac{1}{|u|}\|\phi_1\|_{L^2_{sc}(S_{u,\ub})}\|\phi_2\|_{L^2_{sc}(S_{u,\ub})}.
    \end{align*}
\end{prop}
\begin{rk}
Note that in the region $V_*$, we have $\frac{1}{|u|}\leq\frac{4}{a}\ll 1$. This means if all terms are normal, the nonlinear terms could be treated as lower order terms.
\end{rk}
\subsection{Fundamental norms}
In this section, we define the fundamental norms in $V_*$.
\subsubsection{Schematic notation \texorpdfstring{$\Gag$}{} and \texorpdfstring{$\Gab$}{}}
We introduce the following schematic notation.
\begin{df}\label{gammag}
We divide the Ricci coefficients into two parts:
\begin{align*}
    \Gag:=\left\{\frac{|u|}{a}\trchc,\,\frac{a}{|u|}\trchbc,\,\eta,\,\etab,\,\ze,\,\om,\,\omb\right\},\qquad\qquad\Gab:=\left\{\hch,\,\frac{a}{|u|}\hchb\right\},
\end{align*}
where $\trchc$ and $\trchbc$ are defined in \eqref{dftrchctrchbc}. We also denote:
\begin{align*}
\Gag^{(1)}:=(\af\nabs)^{\leq 1}\Gag\cup\{\b,\,\rho,\,\si,\,\bb,\,\aa\},\qquad\qquad\Gab^{(1)}:=(\af\nabs)^{\leq 1}\Gab\cup\{\a\}.
\end{align*}
Finally, we define for $i\geq 1$
\begin{align*}
    \Gag^{(i+1)}:=(\af\nabs)^{\leq 1}\Gag^{(i)},\qquad\qquad \Gab^{(i+1)}:=(\af\nabs)^{\leq 1}\Gab^{(i)}.
\end{align*}
\end{df}
\subsubsection{\texorpdfstring{$\mr$}{} norms (\texorpdfstring{$L^2$}{}--flux of curvature)}\label{secRnorms}
We define
\begin{align*}
    \mr_i(u,\ub)&:=\afd\left\|\a^{(i)}\right\|_{L^2_{sc}(\cuv)}+\left\|(\b,\rho,\si,\bb)^{(i)}\right\|_{L^2_{sc}(\cuv)},\\
    \ur_i(u,\ub)&:=\afd\left\|\b^{(i)}\right\|_{L^2_{sc}(\ucuv)}+\left\|(\rho,\si,\bb,\aa)^{(i)}\right\|_{L^2_{sc}(\ucuv)}.
\end{align*}
Then, we denote
\begin{align*}
    \mr:=\sup_{V_*}\sum_{i=0}^{s+5}\big(\mr_i(u,\ub)+\ur_i(u,\ub)\big).
\end{align*}
\subsubsection{\texorpdfstring{$\mo$}{} norms (\texorpdfstring{$L^2(S_{u,\ub})$}{}--norms of geometric quantities)}\label{secOnorms}
We define
\begin{align*}
    \mo_i(u,\ub):=\left\|\Gag^{(i)}\right\|_{L^2_{sc}(S_{u,\ub})}+\afd\left\|\Gab^{(i)}\right\|_{L^2_{sc}(S_{u,\ub})}.
\end{align*}
Finally, we denote
\begin{align*}
    \mo:=\sup_{V_*}\sum_{i=0}^{s+5}\mo_i(u,\ub).
\end{align*}
\subsubsection{\texorpdfstring{$\mo_{(0)}$}{} and \texorpdfstring{$\mr_{(0)}$}{} norms (Initial data)}\label{initialO0}
We introduce the following norms on $H_{u_\infty}$:
\begin{align*}
\mo_{(0)}:=\sup_{\ub\in[0,1]}\sum_{i=0}^{s+5}\mo_i(u_\infty,\ub),\qquad\quad\mr_{(0)}:=\sup_{\ub\in[0,1]}\sum_{i=0}^{s+5}\mr_i(u_\infty,\ub).
\end{align*}
\subsection{Main intermediate results of trapped surface formation}
The main goal of this section is to prove the following theorem.
\begin{thm}\label{maintrapped}
There exists a sufficiently large constant $a_0>0$ such that the following holds. For $a>a_0$ and $u_\infty \leq -a$, an initial data on $H_{u_\infty}$ that satisfies:
\begin{itemize}
    \item the following estimate hold along $H_{u_\infty}$:
    $$
    \sum_{i+j\leq s+10}\afd|u_\infty|\left\|\nabs_4^j(|u_\infty|\nabs)^i\hch\right\|_{L^\infty(S_{u_\infty,\ub})}\leq 1,\qquad \forall\; \ub\in [0,1];
    $$
    \item Minkowskian initial data along $\ub=0$.
\end{itemize}
Einstein vacuum equations \eqref{EVE} admit a unique solution in $V_*$ which satisfies:
\begin{equation}\label{finalestimates}
    \mr\les 1,\qquad\quad \mo\les 1.
\end{equation}
The estimate \eqref{finalestimates} implies, in particular, on every sphere $S_{u,\ub}\subseteq V_*$:
\begin{align*}
    |\hch,\a|&\les\frac{a^\frac{1}{2}}{|u|},\qquad |\b|\les\frac{a^\frac{1}{2}}{|u|^2}, \qquad |\rho,\si|\les\frac{a}{|u|^3},\qquad |\bb|\les\frac{a^\frac{3}{2}}{|u|^4},\qquad |\aa|\les\frac{a^2}{|u|^5},\\
    |\om,\trch,\log\Om|&\les\frac{1}{|u|},\quad|\trchbc|\les\frac{1}{|u|^2},\qquad\;\;\;|\omb|\les\frac{a}{|u|^3},\qquad\qquad\quad\;\;\,\,|\eta,\etab,\ze,\hchb|\les\frac{a^\frac{1}{2}}{|u|^2}.
\end{align*}
Moreover, analog estimates also hold for their $H^{s+4}(S_{u,\ub})$--norms.
\end{thm}
The proof of Theorem \ref{maintrapped} is given in Section \ref{proofmain}. It is based on two intermediate results stated as follows, concerning the estimates for the norms $\mo$ and $\mr$ under bootstrap assumptions. 
\begin{thm}\label{M1}
Under the assumptions of Theorem \ref{maintrapped}. Assume in addition that
\begin{equation}
\mo_{(0)}\les 1,\qquad \mr_{(0)}\les 1,\qquad\mo\leq a^\frac{1}{6},\qquad \RR\leq a^\frac{1}{6}.
\end{equation}
Then, we have
\begin{equation}
    \mr\les 1.
\end{equation}
\end{thm}
Theorem \ref{M1} is proved in Section \ref{secR}. The proof is based on the $|u|^p$--weighted estimates for Bianchi equations established in \cite{ShenWan}. Similar arguments are also used in \cite{An12,An,AnLuk,holzegel,Shen22,Shen24}.
\begin{thm}\label{M2}
Under the assumptions of Theorem \ref{maintrapped}. Assume in addition that
\begin{align}
    \mo_{(0)}\les 1,\qquad \mr_{(0)}\les 1,\qquad \mo\leq a^\frac{1}{6},\qquad \RR\les 1.
\end{align}
Then, we have
\begin{equation}
    \mo\les 1.
\end{equation}
\end{thm}
Theorem \ref{M2} is proved in Section \ref{secO}. The proof is done by integrating the null structure equations and the Bianchi equations along the outgoing and incoming null cones.
\subsection{Bootstrap assumptions and first consequences}\label{secboot}
In the rest of Section \ref{sectrapped}, we always make the following bootstrap assumptions:
\begin{align}
    \mo\leq a^\frac{1}{6},\quad\qquad\mr\leq a^\frac{1}{6}.\label{B1}
\end{align}
The following consequences of \eqref{B1} will be used frequently throughout this paper.
\begin{lem}\label{evolution}
    For any $S$--tangent tensor field $\phi$, we have
    \begin{align*}
        \|\phi\|_{L^2_{sc}(S_{u,\ub})}\les\|\phi\|_{L^2_{sc}(S_{u,0})}+\int_0^\ub \|\nabs_4\phi\|_{L^2_{sc}(S_{u,\ub'})}d\ub'.
    \end{align*}
\end{lem}
\begin{proof}
    See Proposition 3.6 in \cite{An}.
\end{proof}
\begin{lem}\label{3evolution}
    Let $\phi$ and $F$ be $S$--tangent tensor fields satisfying the following transport equation:
    \begin{align*}
        \nabs_3\phi+\la_0\trchb\phi=F.
    \end{align*}
    Denoting $\la_1=2\la_0-1$, we have
    \begin{align*}
        |u|^{\la_1-2s_2(\phi)}\|\phi\|_{L^2_{sc}(S_{u,\ub})}\les |u_\infty|^{\la_1-2s_2(\phi)}\|\phi\|_{L^2_{sc}(S_{u_\infty,\ub})}+a\int_{u_\infty}^u|u'|^{\la_1-2s_2(F)}\|F\|_{L^2_{sc}(S_{u',\ub})}du'.
    \end{align*}
\end{lem}
\begin{proof}
    See Proposition 3.7 in \cite{An}.
\end{proof}
\begin{prop}\label{sobolev}
    We have the following Sobolev inequality:
    \begin{align}
        \|\phi\|_{L^\infty_{sc}(S_{u,\ub})}\les\|\phi^{(2)}\|_{L^2_{sc}(S_{u,\ub})}.
    \end{align}
\end{prop}
\begin{proof}
    See Proposition 3.11 in \cite{An}.
\end{proof}
\begin{prop}\label{commutation}
We have the following schematic commutation formulae:
\begin{align*}
    [\Om\nabs_4,\nabs]&=\Gab\c\nabs+\Gag^{(1)},\\
    [\Om\nabs_3,\nabs]&=-\frac{1}{2}\Om\trchb\nabs+\frac{|u|}{a}\Gab\c\nabs+\Gag^{(1)}.
\end{align*}
\end{prop}
\begin{proof}
    It follows directly from Lemma \ref{comm} and Definition \ref{gammag}.
\end{proof}
\subsection{Energy estimates for curvature}\label{secR}
In this section, we prove Theorem \ref{M1} by the $|u|^p$--weighted estimate established in \cite{ShenWan}.
\subsubsection{Estimates for general Bianchi pairs}
The following lemma provides the general structure of Bianchi pairs.
\begin{lem}\label{keypoint}
Let $k=1,2$ and $a_{(1)}$, $a_{(2)}$ be real numbers. Then, we have the following properties.
\begin{enumerate}
    \item If $\psi_{(1)},h_{(1)}\in\sk_k$ and $\psi_{(2)},h_{(2)}\in \sk_{k-1}$ satisfying
    \begin{align}
    \begin{split}\label{bianchi1}
    \nabs_3(\psi_{(1)})+\left(\frac{1}{2}+s_2(\psi_{(1)})\right)\trchb\,\psi_{(1)}&=-k\sld_k^*(\psi_{(2)})+h_{(1)},\\
    \nabs_4(\psi_{(2)})&=\sld_k(\psi_{(1)})+h_{(2)}.
    \end{split}
    \end{align}
Then, the pair $(\psi_{(1)},\psi_{(2)})$ satisfies for any real number $p$
\begin{align}
\begin{split}\label{div}
       &\bdiv(|u|^p|\psi_{(1)}|^2e_3)+k\bdiv(|u|^p|\psi_{(2)}|^2e_4)+\left(p-4s_2(\psi_{(1)})\right)|u|^{p-1}|\psi_{(1)}|^2\\
       =&2k|u|^p\sdivs(\psi_{(1)}\c\psi_{(2)})
       +2|u|^p\psi_{(1)}\cdot h_{(1)}+2k|u|^p\psi_{(2)}\c h_{(2)}\\
       +&a^{-1}|u|^{p+1}\Gag|\psi_{(1)}|^2+|u|^p\Gag|\psi_{(2)}|^2.
\end{split}
\end{align}
    \item If $\psi_{(1)},h_{(1)}\in\sk_{k-1}$ and $\psi_{(2)},h_{(2)}\in\sk_k$ satisfying
\begin{align}
\begin{split}\label{bianchi2}
\nabs_3(\psi_{(1)})+\left(\frac{1}{2}+s_2(\psi_1)\right)\trchb\,\psi_{(1)}&=\sld_k(\psi_{(2)})+h_{(1)},\\
\nabs_4(\psi_{(2)})&=-k\sld_k^*(\psi_{(1)})+h_{(2)}.
\end{split}
\end{align}
Then, the pair $(\psi_{(1)},\psi_{(2)})$ satisfies for any real number $p$
\begin{align}
\begin{split}\label{div2}
    &k\bdiv(|u|^p|\psi_{(1)}|^2e_3)+\bdiv(|u|^p|\psi_{(2)}|^2e_4)+k(p-4s_2(\psi_{(1)}))|u|^{p-1}|\psi_{(1)}|^2\\
    =&2|u|^p\sdivs(\psi_{(1)}\c\psi_{(2)})
    +2k|u|^p\psi_{(1)}\c h_{(1)}+2|u|^p\psi_{(2)}\c h_{(2)}\\
    +&a^{-1}|u|^{p+1}\Gag|\psi_{(1)}|^2+|u|^p\Gag|\psi_{(2)}|^2.
\end{split}
\end{align}
\end{enumerate}
\end{lem}
\begin{proof}
    See Lemma 5.1 in \cite{ShenWan}.
\end{proof}
\begin{rk}
    Note that the Bianchi equations can be written as systems of equations of type \eqref{bianchi1} and \eqref{bianchi2}. In particular
    \begin{itemize}
        \item the Bianchi pair $(\a,\b)$ satisfies \eqref{bianchi1} with $k=2$,
        \item the Bianchi pair $(\b,(\rho,-\si))$ satisfies \eqref{bianchi1} with $k=1$,
        \item the Bianchi pair $((\rho,\si),\bb)$ satisfies \eqref{bianchi2} with $k=1$,
        \item the Bianchi pair $(\bb,\aa)$ satisfies \eqref{bianchi2} with $k=2$.
    \end{itemize}
\end{rk}
\begin{prop}\label{keyintegral}
Let $(\psi_{(1)},\psi_{(2)})$ be a Bianchi pair that satisfies \eqref{bianchi1} or \eqref{bianchi2}. Then, we have
\begin{align*}
&\;\|\psi_{(1)}\|_{L^2_{sc}(\cuv)}^2+\|\psi_{(2)}\|^2_{L^2_{sc}(\ucuv)}\\
\les&\;\|\psi_{(1)}\|_{L^2_{sc}(H_{u_\infty}^{(0,\ub)})}^2+\|\psi_{(2)}\|^2_{L^2_{sc}(\Hb_{0}^{(u_\infty,u)})}\\
+&\int_0^\ub\int_{u_\infty}^u\frac{a}{|u'|}\left\|\psi_{(1)}\c\left(h_{(1)}+\frac{|u'|}{a}\Gag\c\psi_{(1)}\right)\right\|_{L^1_{sc}(S_{u',\ub'})}du'd\ub'\\
+&\int_0^\ub\int_{u_\infty}^u\frac{a}{|u'|}\left\|\psi_{(2)}\c (h_{(2)}+\Gag\c\psi_{(2)})\right\|_{L^1_{sc}(S_{u',\ub'})}du'd\ub'.
\end{align*}
\end{prop}
\begin{proof}
See Proposition 5.3 in \cite{ShenWan}.
\end{proof}
\begin{df}\label{dfdkb}
We define the weighted angular derivatives $\dkb$ as follows:
\begin{align*}
    \dkb U &:= \af\sld_2 U,\qquad \forall U\in \sk_2,\\
    \dkb \xi&:=\af\sld_1 \xi,\qquad\,\,\, \forall \xi\in \sk_1,\\
    \dkb f&:= \af\sld_1^* f,\qquad \,\,\forall f\in \sk_0.
\end{align*}
We denote for any tensor $h\in\sk_k$, $k=0,1,2$,
\begin{equation*}
    h^{(0)}:=h,\qquad\quad h^{(i)}:=(h,\dkb h,...,\dkb^i h).
\end{equation*}
\end{df}
\subsubsection{Estimates for the Bianchi pair \texorpdfstring{$(\a,\b)$}{}}
\begin{prop}\label{estab}
We have the following estimate:
    \begin{align*}
        \|\a^{(s+5)}\|_{L^2_{sc}(\cuv)}+\|\b^{(s+5)}\|_{L^2_{sc}(\ucuv)}\les\af.
    \end{align*}
\end{prop}
\begin{proof}
We have from Proposition \ref{bianchiequations}\footnote{We ignore $\b$ in the nonlinear terms since it decays better than $\a$.}
\begin{align*}
\nabs_3\a+\frac{1}{2}\trchb\,\a&=-2\sld_2^*\b+\Gab\c\a,\\
\nabs_4\b&=\sld_2\a+\Gag\c\a.
\end{align*}
Commuting it with $\dkb^i$ and applying Proposition \ref{commutation}, we infer\footnote{The lower order linear term of $\b$ on the R.H.S. comes from \eqref{dddd}.}
\begin{align*}
\nabs_3(\dkb^i\a)+\frac{i+1}{2}\trchb(\dkb^i\a)&=-\slD^*(\dkb^i\b)+|u|^{-2}\b^{(i-1)}+\frac{|u|}{a}(\Gab\c\a)^{(i)},\\
\nabs_4(\dkb^i\b)&=\slD(\dkb^i\a)+(\Gab\c\a)^{(i)},
\end{align*}
where $\slD$ denotes an elliptic Hodge operator which depends on $i$. Applying Proposition \ref{keyintegral} with $\psi_{(1)}=\dkb^i\a$ and $\psi_{(2)}=\dkb^i\b$, we deduce by induction
\begin{align*}
\|\dkb^i\a\|_{L^2_{sc}(\cuv)}^2+\|\dkb^i\b\|^2_{L^2_{sc}(\ucuv)}&\les\int_0^\ub\|\dkb^i\a\|_{L^2_{sc}(S_{u_\infty,\ub'})}^2d\ub'\\
&+\int_0^\ub\int_{u_\infty}^u\|\a^{(i)}\c(\Gab\c\a)^{(i)}\|_{L^1_{sc}(S_{u',\ub'})}du'd\ub'.
\end{align*}
Notice that we have from Propositions \ref{Holder}, \ref{sobolev} and \eqref{fluxsc}
\begin{align*}
    &\int_0^\ub\int_{u_\infty}^u\|\a^{(i)}\c(\Gab\c\a)^{(i)}\|_{L^1_{sc}(S_{u',\ub'})}du'd\ub'\\
    \les&\int_0^\ub\int_{u_\infty}^u\frac{1}{|u'|}\|\a^{(i)}\|_{L^2_{sc}(S_{u',\ub'})}\|\Gab^{(i)}\c\a\|_{L^2_{sc}(S_{u',\ub'})}+\frac{1}{|u'|}\|\a^{(i)}\|_{L^2_{sc}(S_{u',\ub'})}\|\Gab\c\a^{(i)}\|_{L^2_{sc}(S_{u',\ub'})}du'd\ub'\\
    \les&\int_0^\ub\int_{u_\infty}^u\frac{a^\frac{1}{6}\af}{|u'|^2}\|\a^{(i)}\|_{L^2_{sc}(S_{u',\ub'})}\|\a\|_{L^\infty_{sc}(S_{u',\ub'})}+\frac{a^\frac{1}{6}\af}{|u'|^2}\|\a^{(i)}\|_{L^2_{sc}(S_{u',\ub'})}^2du'd\ub'\\
    \les&\int_{u_\infty}^u\frac{a^\frac{2}{3}}{|u'|^2}\|\a^{(i)}\|_{L^2_{sc}(\cuv)}^2du'\\
    \les&\int_{u_\infty}^u\frac{a^2}{|u'|^2}du'\les a.
\end{align*}
Combining with the initial condition $\mr_{(0)}\les 1$, we obtain for $i\leq s+5$
\begin{align*}
\|\dkb^i\a\|^2_{L^2_{sc}(\cuv)}+\|\dkb^i\b\|^2_{L^2_{sc}(\ucuv)}\les a.
\end{align*}
This concludes the proof of Proposition \ref{estab}.
\end{proof}
\subsubsection{Estimate for \texorpdfstring{$\Gab$}{}}
\begin{prop}\label{esthchhchb}
We have the following estimates:
\begin{align}
    \|\hch^{(s+5)}\|_{L^2_{sc}(S_{u,\ub})}&\les\af,\label{esthch}\\
    \|\hchb^{(s+5)}\|_{L^2_{sc}(S_{u,\ub})}&\les\frac{|u|}{\af},\label{esthchb}\\
    \|\a^{(s+4)}\|_{L^2_{sc}(S_{u,\ub})}&\les\af.\label{esta}
\end{align}
\end{prop}
\begin{rk}
    Proposition \ref{esthchhchb} implies that:
    \begin{align}\label{estGab}
        \|\Gab^{(s+5)}\|_{L^2_{sc}(S_{u,\ub})}\les\af.
    \end{align}
    In the remainder of this paper, we will use \eqref{estGab}, which improves the $\Gab$ bound in \eqref{B1}.
\end{rk}
\begin{proof}[Proof of Proposition \ref{esthchhchb}]
    We have from Proposition \ref{nulles}
    \begin{align*}
        \nabs_4\hch=-\a+\Gag\c\Gab.
    \end{align*}
    Differentiating it by $\dkb^i$ and applying Proposition \ref{commutation}, we infer
    \begin{align*}
        \nabs_4(\dkb^i\hch)=\a^{(i)}+(\Gag\c\Gab)^{(i)}.
    \end{align*}
    Applying Lemma \ref{evolution}, we deduce for $i\leq s+5$
    \begin{align*}
        &\quad\;\|\dkb^i\hch\|_{L^2_{sc}(S_{u,\ub})}\\
        &\les\int_0^\ub\|\a^{(i)}\|_{L^2_{sc}(S_{u,\ub'})}+\|(\Gag\c\Gab)^{(i)}\|_{L^2_{sc}(S_{u,\ub'})}d\ub'\\
        &\les\|\a^{(i)}\|_{L^2_{sc}(\cuv)}+\frac{1}{|u|}\int_0^\ub\|\Gag^{(i)}\|_{L^2_{sc}(S_{u,\ub'})}\|\Gab\|_{L^\infty_{sc}(S_{u,\ub'})}+\|\Gab^{(i)}\|_{L^2_{sc}(S_{u,\ub'})}\|\Gag\|_{L^\infty_{sc}(S_{u,\ub'})}d\ub'\\
        &\les\af+\frac{a^\frac{1}{3}\af}{|u|}\les\af,
    \end{align*}
    which implies \eqref{esthch}. Next, we have from Proposition \ref{nulles}
    \begin{align*}
        \nabs_3\hchb+\trchb\,\hchb=-\aa+\frac{|u|}{a}\Gag\c\Gab.
    \end{align*}
    Differentiating it by $\dkb^i$ and applying Proposition \ref{commutation}, we infer
    \begin{align*}
        \nabs_3(\dkb^i\hchb)+\frac{i+2}{2}\trchb(\dkb^i\hchb)=\aa^{(i)}+\frac{|u|^2}{a^2}(\Gab\c\Gab)^{(i)}.
    \end{align*}
    Applying Lemma \ref{3evolution}, we infer for $i\leq s+5$
    \begin{align*}
        |u|^{-1}\|\dkb^i\hchb\|_{L^2_{sc}(S_{u,\ub})}&\les|u_\infty|^{-1}\|\dkb^i\hchb\|_{L^2_{sc}(S_{u_\infty,\ub})}+a\int_{u_\infty}^u|u'|^{-3}\|\dkb^i\aa\|_{L^2_{sc}(S_{u',\ub})}du'\\
        &+\int_{u_\infty}^u\frac{1}{a|u'|^2}\|\Gab^{(i)}\|_{L^2_{sc}(S_{u',\ub})}\|\Gab\|_{L^\infty_{sc}(S_{u',\ub})}du'\\
        &\les\afd+\frac{a^\frac{1}{6}}{|u|}+\int_{u_\infty}^u\frac{a^\frac{1}{3}a}{a|u'|^2}du'\\
        &\les\afd,
    \end{align*}
    which implies \eqref{esthchb}. Finally, we have from Proposition \ref{bianchiequations}
    \begin{align*}
        \nabs_3\a+\frac{1}{2}\trchb\,\a=\b^{(1)}+(\Gag\c\Gab)^{(1)}.
    \end{align*}
    Differentiating it by $\dkb^{i-1}$ and applying Proposition \ref{commutation}, we infer
    \begin{align*}
        \nabs_3(\dkb^{i-1}\a)+\frac{i}{2}\trchb(\dkb^{i-1}\a)=\b^{(i)}+\frac{|u|}{a}(\Gab\c\Gab)^{(i)}.
    \end{align*}
    Applying Lemma \ref{3evolution}, we deduce for $i\leq s+5$
    \begin{align*}
        \|\dkb^{i-1}\a\|_{L^2_{sc}(S_{u,\ub})}&\les\|\dkb^{i-1}\a\|_{L^2_{sc}(S_{u_\infty,\ub})}+\int_{u_\infty}^u\frac{a}{|u'|^2}\|\dkb^i\b\|_{L^2_{sc}(S_{u',\ub})}du'\\
        &+a\int_{u_\infty}^u\frac{1}{a|u'|^2}\|\Gab^{(i)}\|_{L^2_{sc}(S_{u',\ub})}\|\Gab\|_{L^\infty_{sc}(S_{u',\ub})}du'\\
        &\les\af+a^\frac{1}{6}+\int_{u_\infty}^u\frac{a^\frac{1}{3}a^2}{a|u'|^2}du'\\
        &\les\af,
    \end{align*}
    which implies \eqref{esta}. This concludes the proof of Proposition \ref{esthchhchb}.
\end{proof}
\subsubsection{Estimates for the Bianchi pairs \texorpdfstring{$(\b,(-\rho,\si))$}{}, \texorpdfstring{$((\rho,\si),\bb)$}{} and \texorpdfstring{$(\bb,\aa)$}{}}
\begin{prop}\label{estbr}
Let $(\psi_1,\psi_2)\in\{(\b,(-\rho,\si)),((\rho,\si),\bb), (\bb,\aa)\}$. Then, we have
\begin{align*}
    \|\psi_1^{(s+5)}\|_{L^2_{sc}(\cuv)}+\|\psi_2^{(s+5)}\|_{L^2_{sc}(\ucuv)}\les 1.
\end{align*}
\end{prop}
\begin{proof}
We have from Proposition \ref{bianchiequations}
    \begin{align*}
        \nabs_3\psi_1+\left(\frac{1}{2}+s_2(\psi_1)\right)\trchb\,\psi_1&=-\slD^*\psi_2+\frac{|u|}{a}\Gab\c\psi,\\
        \nabs_4\psi_2&=\slD\psi_1+\frac{|u|}{a}\Gab\c\a,
    \end{align*}
where $\slD$ denotes an elliptic Hodge operator which depends on the type of $\psi_1$ and $\psi_2$. Differentiating it by $\dkb^i$ and applying Proposition \ref{commutation}, we infer
    \begin{align*}
        \nabs_3(\dkb^i\psi_1)+\left(\frac{i+1}{2}+s_2(\psi_1)\right)\trchb(\dkb^i\psi_1)&=-\slD^*(\dkb^i\psi_2)+\frac{|u|}{a}(\Gab\c\psi)^{(i)},\\
        \nabs_4(\dkb^i\psi_2)&=\slD(\dkb^i\psi_1)+\frac{|u|}{a}(\Gab\c\a)^{(i)}.
    \end{align*}
Applying Proposition \ref{keyintegral}, we deduce
\begin{align*}
\|\dkb^i\psi_1\|_{L^2_{sc}(\cuv)}^2+\|\dkb^i\psi_2\|^2_{L^2_{sc}(\ucuv)}&\les\int_0^\ub\|\dkb^i\psi_1\|_{L^2_{sc}(S_{u_\infty,\ub'})}^2d\ub'\\
&+\int_0^\ub\int_{u_\infty}^u\|\psi^{(i)}\c(\Gab\c\a)^{(i)}\|_{L^1_{sc}(S_{u',\ub'})}du'd\ub',
\end{align*}
where $\psi:=(\psi_1,\psi_2)$. Notice that we have from \eqref{estGab} and \eqref{fluxsc}
\begin{align*}
    &\int_0^\ub\int_{u_\infty}^u\|\psi^{(i)}\c(\Gab\c\a)^{(i)}\|_{L^1_{sc}(S_{u',\ub'})}du'd\ub'\\
    \les&\int_0^\ub\int_{u_\infty}^u\frac{1}{|u'|}\|\psi^{(i)}\|_{L^2_{sc}(S_{u',\ub'})}\|(\Gab\c\a)^{(i)}\|_{L^2_{sc}(S_{u',\ub'})}du'd\ub'\\
    \les&\int_0^\ub\int_{u_\infty}^u\frac{\af}{|u'|^2}\|\psi^{(i)}\|_{L^2_{sc}(S_{u',\ub'})}\|\a^{(i)}\|_{L^2_{sc}(S_{u',\ub'})}du'd\ub'\\
    \les&\int_0^\ub\left(\int_{u_\infty}^u\frac{a}{|u'|^2}\|\psi^{(i)}\|^2_{L^2_{sc}(S_{u',\ub'})}du'\right)^\frac{1}{2}\left(\int_{u_\infty}^u\frac{1}{|u'|^2}\|\a^{(i)}\|_{L^2_{sc}(S_{u',\ub'})}^2du'\right)^\frac{1}{2}d\ub'\\
    \les&\sup_{0\leq\ub\leq 1}\|\psi^{(i)}\|_{L^2_{sc}(\ucuv)}\left(\int_{u_\infty}^u\frac{1}{|u'|^2}\|\a^{(i)}\|_{L^2_{sc}(H_{u'}^{(0,\ub)})}^2du'\right)^\frac{1}{2}.
\end{align*}
Moreover, we have from $\mr_{(0)}\les 1$ and $i\leq s+5$
\begin{align*}
    \int_0^\ub\|\dkb^i\psi_1\|_{L^2_{sc}(S_{u_\infty,\ub'})}^2d\ub'\les 1.
\end{align*}
Taking the sum for $0\leq i\leq s+5$ and $(\psi_1,\psi_2)\in\{(\b,(-\rho,\si)),((\rho,\si),\bb), (\bb,\aa)\}$ and applying Cauchy-Schwarz inequality, we deduce from Proposition \ref{estab}
\begin{align*}
\|\psi_1^{(s+5)}\|_{L^2_{sc}(\cuv)}^2+\|\psi_2^{(s+5)}\|^2_{L^2_{sc}(\ucuv)}\les 1+\int_{u_\infty}^u\frac{1}{|u'|^2}\|\a^{(s+5)}\|_{L^2_{sc}(H_{u'}^{(0,\ub)})}^2du'\les 1.
\end{align*}
This concludes the proof of Proposition \ref{estbr}.
\end{proof}
Combining Propositions \ref{estab} and \ref{estbr}, this concludes the proof of Theorem \ref{M1}.
\subsection{\texorpdfstring{$L^2(S_{u,\ub})$}{}--estimates for Ricci coefficients and curvature}\label{secO}
The goal of this section is to prove Theorem \ref{M2}.
\subsubsection{Estimate for \texorpdfstring{$\om$}{}}
\begin{prop}\label{estom}
    We have the following estimate:
    \begin{align*}
        \|\om^{(s+5)}\|_{L^2_{sc}(S_{u,\ub})}\les 1.
    \end{align*}
\end{prop}
\begin{proof}
    We have from Proposition \ref{nulles}
    \begin{align*}
        \nabs_3\om=\frac{1}{2}\rho+\Gag\c\Gag.
    \end{align*}
    Differentiating it by $\dkb^i$ and applying Proposition \ref{commutation}, we infer
    \begin{align*}
        \nabs_3(\dkb^i\om)+\frac{i}{2}\trchb(\dkb^i\om)=\rho^{(i)}+\frac{|u|}{a}(\Gag\c\Gab)^{(i)}.
    \end{align*}
    Applying Lemma \ref{3evolution}, we deduce for $i\leq s+5$
    \begin{align*}
        |u|^{-1}\|\dkb^i\om\|_{L^2_{sc}(S_{u,\ub})}&\les|u_\infty|^{-1}\|\dkb^i\om\|_{L^2_{sc}(S_{u_\infty,\ub})}+a\int_{u_\infty}^u|u'|^{-3}\|\rho^{(i)}\|_{L^2_{sc}(S_{u',\ub})}du'\\
        &+\int_{u_\infty}^u|u'|^{-2}\|(\Gag\c\Gab)^{(i)}\|_{L^2_{sc}(S_{u',\ub})}du'\\
        &\les\frac{1}{|u_\infty|}+\left(\int_{u_\infty}^u\frac{a}{|u'|^2}\|\rho^{(i)}\|^2_{L^2_{sc}(S_{u',\ub})}du'\right)^\frac{1}{2}\left(\int_{u_\infty}^u\frac{a}{|u'|^4}du'\right)^\frac{1}{2}+\int_{u_\infty}^u\frac{a^\frac{1}{3}\af}{|u'|^3}du'\\
        &\les\frac{1}{|u|}+\frac{\af}{|u|^\frac{3}{2}}+\frac{a^\frac{5}{6}}{|u|^2}\\
        &\les\frac{1}{|u|}.
    \end{align*}
    This concludes the proof of Proposition \ref{estom}.
\end{proof}
\subsubsection{Estimate for \texorpdfstring{$\log\Om$}{}}
\begin{prop}\label{estOm}
    We have the following estimate:
    \begin{align*}
        \|\log\Om\|_{L^2_{sc}(S_{u,\ub})}\les 1.
    \end{align*}
\end{prop}
\begin{proof}
    We have from \eqref{nullidentities}
    \begin{align*}
        \nabs_4\log\Om=-2\om.
    \end{align*}
    Applying Lemma \ref{evolution} and Proposition \ref{estom}, we infer
    \begin{align*}
        \|\log\Om\|_{L^2_{sc}(S_{u,\ub})}\les \int_{0}^\ub \|\om\|_{L^2_{sc}(S_{u,\ub'})}d\ub'\les1.
    \end{align*}
    This concludes the proof of Proposition \ref{estOm}.
\end{proof}
\subsubsection{Estimates for \texorpdfstring{$\omb$}{} and \texorpdfstring{$\eta$}{}}
\begin{prop}\label{estomb}
    We have the following estimate:
    \begin{align*}
        \|(\omb,\eta)^{(s+5)}\|_{L^2_{sc}(S_{u,\ub})}\les 1.
    \end{align*}
\end{prop}
\begin{proof}
    We have from Proposition \ref{nulles}
    \begin{align*}
        \nabs_4(\omb,\eta)=(\rho,\b)^{(0)}+\Gab\c\Gab.
    \end{align*}
    Differentiating it by $\dkb^i$ and applying Proposition \ref{commutation}, we infer
    \begin{align*}
        \nabs_4(\dkb^i\omb,\dkb^i\eta)=(\rho,\b)^{(i)}+(\Gab\c\Gab)^{(i)}.
    \end{align*}
    Applying Lemma \ref{evolution}, we deduce for $i\leq s+5$
    \begin{align*}
        \|\dkb^i(\omb,\eta)\|_{L^2_{sc}(S_{u,\ub})}&\les \int_{0}^\ub\|(\rho,\b)^{(i)}\|_{L^2_{sc}(S_{u,\ub'})}+\|(\Gab\c\Gab)^{(i)}\|_{L^2_{sc}(S_{u,\ub'})}d\ub'\\
        &\les\left(\int_{0}^\ub\|(\rho,\b)^{(i)}\|^2_{L^2_{sc}(S_{u,\ub'})}d\ub'\right)^\frac{1}{2}+\int_0^\ub \frac{1}{|u|}\|\Gab^{(i)}\|_{L^2_{sc}(S_{u,\ub'})}\|\Gab\|_{L^\infty_{sc}(S_{u,\ub'})} d\ub'\\
        &\les 1+\frac{\af\af}{|u|}\\
        &\les 1,
    \end{align*}
    where we used \eqref{estGab} at the third step. This concludes the proof of Proposition \ref{estomb}.
\end{proof}
\subsubsection{Estimate for \texorpdfstring{$\trchc$}{}}
\begin{prop}\label{esttrchc}
    We have the following estimate:
    \begin{align*}
        \left\|\trchc^{(s+5)}\right\|_{L^2_{sc}(S_{u,\ub})}\les\frac{a}{|u|}.
    \end{align*}
\end{prop}
\begin{proof}
    We have from Proposition \ref{nulles}
    \begin{align*}
        \nabs_4\trchc=\Gab\c\Gab.
    \end{align*}
    Differentiating it by $\dkb^i$ and applying Proposition \ref{commutation}, we infer
    \begin{align*}
        \nabs_4(\dkb^i\trchc)=(\Gab\c\Gab)^{(i)}.
    \end{align*}
    Applying Lemma \ref{evolution}, we deduce for $i\leq s+5$
    \begin{align*}
        \|\dkb^i\trchc\|_{L^2_{sc}(S_{u,\ub})}&\les\int_0^\ub\|(\Gab\c\Gab)^{(i)}\|_{L^2_{sc}(S_{u,\ub'})}d\ub'\\
        &\les\int_0^\ub \frac{1}{|u|}\|\Gab^{(i)}\|_{L^2_{sc}(S_{u,\ub'})}\|\Gab\|_{L^\infty_{sc}(S_{u,\ub'})} d\ub'\\
        &\les\frac{a}{|u|},
    \end{align*}
    where we used \eqref{estGab} at the last step. This concludes the proof of Proposition \ref{esttrchc}.
\end{proof}
\subsubsection{Estimate for \texorpdfstring{$\trchbc$}{}}
\begin{prop}\label{esttrchbc}
    We have the following estimate:
    \begin{align*}
        \left\|\trchbc^{(s+5)}\right\|_{L^2_{sc}(S_{u,\ub})}\les\frac{|u|}{a}.
    \end{align*}
\end{prop}
\begin{proof}
    We have from Proposition \ref{nulles}
    \begin{align*}
        \nabs_3\trchb+\frac{1}{2}(\trchb)^2=|u|^{-1}\Gag+\frac{|u|^2}{a^2}\Gab\c\Gab.
    \end{align*}
    We also have
    \begin{align*}
        \nabs_3\left(\frac{2}{|u|}\right)-\frac{1}{2}\frac{4}{|u|^2}=|u|^{-2}\Gag.
    \end{align*}
    Taking the sum, we infer
    \begin{align*}
        \nabs_3\trchbc+\trchb\,\trchbc=|u|^{-1}\Gag+\frac{|u|^2}{a^2}\Gab\c\Gab.
    \end{align*}
    Differentiating it by $\dkb^i$ and applying Proposition \ref{commutation}, we infer
    \begin{align*}
        \nabs_3(\dkb^i\trchbc)+\frac{i+2}{2}\trchb(\dkb^i\trchbc)=|u|^{-1}\Gag^{(i)}+\frac{|u|^2}{a^2}(\Gab\c\Gab)^{(i)}.
    \end{align*}
    Applying Lemma \ref{3evolution} and \eqref{estGab}, we deduce for $i\leq s+5$
    \begin{align*}
        &|u|^{-1}\|\dkb^i\trchbc\|_{L^2_{sc}(S_{u,\ub})}\\
        \les\;&|u_\infty|^{-1}\|\dkb^i\trchbc\|_{L^2_{sc}(S_{u_\infty,\ub})}+\int_{u_\infty}^u\frac{a}{|u'|^4}\|\Gag^{(i)}\|_{L^2_{sc}(S_{u',\ub})}+\frac{1}{a|u'|}\|(\Gab\c\Gab)^{(i)}\|_{L^2_{sc}(S_{u',\ub})}du'\\
        \les\;&|u_\infty|^{-1}\frac{|u_\infty|}{a}+\frac{a^\frac{1}{6}a}{|u|^3}+\frac{\af\af}{a|u|}\\
        \les\;&a^{-1}.
    \end{align*}
    This concludes the proof of Proposition \ref{esttrchbc}.
\end{proof}
\subsubsection{Estimate for \texorpdfstring{$\etab$}{}}
\begin{prop}\label{estetab}
    We have the following estimate:
    \begin{align*}
        \|\etab^{(s+5)}\|_{L^2_{sc}(S_{u,\ub})}\les 1.
    \end{align*}
\end{prop}
\begin{proof}
    We have from Proposition \ref{nulles}
    \begin{align*}
        \nabs_3\etab+\frac{1}{2}\trchb\,\etab=\bb+|u|^{-1}\Gag+\frac{|u|}{a}\Gag\c\Gab.
    \end{align*}
    Differentiating it by $\dkb^i$ and applying Proposition \ref{commutation}, we infer
    \begin{align*}
        \nabs_3(\dkb^i\etab)+\frac{i+1}{2}\trchb(\dkb^i\etab)=\bb^{(i)}+|u|^{-1}\Gag^{(i)}+\frac{|u|}{a}(\Gag\c\Gab)^{(i)}.
    \end{align*}
    Applying Lemma \ref{3evolution}, Proposition \ref{estbr} and \eqref{estGab}, we deduce for $i\leq s+5$
    \begin{align*}
         |u|^{-1}\|\dkb^i\etab\|_{L^2_{sc}(S_{u,\ub})}&\les|u_\infty|^{-1}\|\dkb^i\etab\|_{L^2_{sc}(S_{u_\infty,\ub})}+\int_{u_\infty}^u\frac{a}{|u'|^3}\|\bb^{(i)}\|_{L^2_{sc}(S_{u',\ub})}du'\\
         &+\int_{u_\infty}^u\frac{a}{|u'|^4}\|\Gag^{(i)}\|_{L^2_{sc}(S_{u',\ub})}du'+\frac{1}{|u'|^2}\|(\Gag\c\Gab)^{(i)}\|_{L^2_{sc}(S_{u',\ub})}du'\\
         &\les\frac{1}{|u_\infty|}+\left(\int_{u_\infty}^u\frac{a}{|u'|^2}\|\bb^{(i)}\|_{L^2_{sc}(S_{u',\ub})}^2du'\right)^\frac{1}{2}\left(\int_{u_\infty}^u\frac{a}{|u'|^4}du'\right)^\frac{1}{2}+\frac{a^\frac{7}{6}}{|u|^3}+\frac{a^\frac{5}{6}}{|u|^2}\\
         &\les\frac{1}{|u|}+\frac{\af}{|u|^\frac{3}{2}}\\
         &\les\frac{1}{|u|}.
    \end{align*}
    This concludes the proof of Proposition \ref{estetab}.
\end{proof}
\subsection{Estimates for \texorpdfstring{$\b$}{}, \texorpdfstring{$\rho$}{}, \texorpdfstring{$\si$}{}, \texorpdfstring{$\bb$}{} and \texorpdfstring{$\aa$}{}}
\begin{prop}\label{estR}
We have the following estimate:
\begin{align*}
    \|(\b,\rho,\si,\bb,\aa)^{(s+4)}\|_{L^2_{sc}(S_{u,\ub})}\les 1.
\end{align*}
\end{prop}
\begin{proof}
    Given $R\in\{\b,\rho,\si,\bb,\aa\}$, we have from Proposition \ref{bianchiequations}\footnote{Here, we ignore all the terms which decay better.}
    \begin{align*}
        \nabs_4R=\afd\a^{(1)}+\frac{|u|}{a}\Gab\c\Gab^{(1)}.
    \end{align*}
    Differentiating it by $\dkb^{i-1}$ and applying Proposition \ref{commutation}, we infer
    \begin{align*}
        \nabs_4(\dkb^{i-1}R)=\afd\a^{(i)}+\frac{|u|}{a}(\Gab\c\Gab)^{(i)}.
    \end{align*}
    Applying Lemma \ref{evolution}, Proposition \ref{esthchhchb} and \eqref{estGab}, we deduce for $i\leq s+5$
    \begin{align*}
        \|\dkb^{i-1}R\|_{L^2_{sc}(S_{u,\ub})}&\les\afd\int_0^\ub\|\a^{(i)}\|_{L^2_{sc}(S_{u,\ub'})} d\ub'+a^{-1}\int_0^\ub\|\Gab^{(i)}\|_{L^2_{sc}(S_{u,\ub'})}\|\Gab\|_{L^\infty_{sc}(S_{u,\ub'})} d\ub'\\
        &\les\afd\left(\int_0^\ub\|\a^{(i)}\|^2_{L^2_{sc}(S_{u,\ub'})} d\ub'\right)^\frac{1}{2}+a^{-1}\af\af\\
        &\les 1.
    \end{align*}
This concludes the proof of Proposition \ref{estR}.
\end{proof}
Combining Propositions \ref{esthchhchb}, \ref{estom}--\ref{estR}, this concludes the proof of Theorem \ref{M2}.
\subsection{Proof of the Theorem \ref{maintrapped}}\label{proofmain}
We now use Theorems \ref{M1} and \ref{M2} to prove Theorem \ref{maintrapped}.
\begin{df}\label{bootstrap}
For any $u_\infty\leq u_*\leq -\frac{a}{4}$, let $\aleph(u_*)$ the set of spacetimes $V(u_*,1)$ associated with a double null foliation $(u,\ub)$ in which we have the following bounds:
\begin{align}
    \mo\leq a^\frac{1}{6},\qquad\quad\mr\leq a^\frac{1}{6}.\label{B2}
\end{align}
\end{df}
\begin{df}\label{defboot}
We denote by $\UU$ the set of values $u_*$ such that $\aleph(u_*)\ne\emptyset$.
\end{df}
The initial assumption and the results in Section \ref{secsignatures} imply that
\begin{equation*}
    \mo_{(0)}+\mr_{(0)}\les 1.
\end{equation*}
Combining with the local existence theorem, we deduce that \eqref{B2} holds if $u_*$ is sufficiently close to $u_\infty$. So, we have $\UU\ne\emptyset$.\\ \\
Define $u_*$ as the supremum of the set $\UU$. We assume by contradiction that $u_*<-\frac{a}{4}$. In particular, we may assume $u_*\in\mathcal{U}$. We consider the region $V(u_*,1)$. Applying Theorem \ref{M1}, we obtain
\begin{equation*}
    \mr\les 1.
\end{equation*}
Then, we apply Theorem \ref{M2} to obtain
\begin{equation*}
    \mo\les 1.
\end{equation*}
Applying local existence results,\footnote{See Theorem 1 in \cite{luk}.} we can extend $V(u_*,1)$ to $V(u_*+\nu,1)$ for a $\nu$ sufficiently small. We denote $\widetilde{\mo}$ and $\widetilde{\mr}$ the norms in the extended region. We have
\begin{equation*}
    \widetilde{\mo}\les1,\qquad\qquad \widetilde{\mr}\les1,
\end{equation*}
as a consequence of continuity. We deduce that $V(u_*+\nu,1)$ satisfies all the properties in Definition \ref{bootstrap}, and so $\aleph(u_*+\nu)\ne\emptyset$, which is a contradiction. Thus, we have $u_*=-\frac{a}{4}$. Moreover, we have
\begin{equation*}
    \mo\les 1,\qquad\quad \mr\les1\qquad\mbox{ in }\;\; V\left(-\frac{a}{4},1\right).
\end{equation*}
Combining with Proposition \ref{sobolev}, we obtain the detailed estimates stated in Theorem \ref{maintrapped}. This concludes the proof of Theorem \ref{maintrapped}.
\subsection{Formation of trapped surface}
\begin{thm}\label{thmtrappedsurface}
    Under the hypothesis of Theorem \ref{maintrapped}, we assume in addition that
    \begin{align}\label{geqathm}
        \int_0^1|u_\infty|^2|\hch|^2(u_\infty,\ub',x^1,x^2)d\ub'\geq a,\qquad \forall \;(x^1,x^2).
    \end{align}
    Then $S_{-\frac{a}{4},1}$ is a trapped surface.
\end{thm}
\begin{proof}
    We have from \eqref{basicnull}
    \begin{align}\label{nab3hch}
        \nabs_3\hch+\frac{1}{2}\trchb\,\hch=\nabs\hot\eta+2\omb\hch-\frac{1}{2}\trch\,\hchb+\eta\hot\eta.
    \end{align}
    Contracting it with $\hch$, we infer
    \begin{align*}
    e_3(|\hch|^2)+\trchb|\hch|^2=2\hch\c\nabs\hot\eta+4\omb|\hch|^2-\trch(\hch\c\hchb)+2\hch\c\eta\hot\eta.
    \end{align*}
Thus, we obtain
\begin{align}
\begin{split}\label{e3u2hch2}
e_3(|u|^2|\hch|^2)&=|u|^2\left(2\hch\c\nabs\hot\eta+4\omb|\hch|^2-\trch(\hch\c\hchb)+2\hch\c\eta\hot\eta\right)\\
&-|u|^2\left(\trchb+\frac{2}{\Om|u|}\right)|\hch|^2.
\end{split}
\end{align}
Applying Theorem \ref{maintrapped}, we have
    \begin{align*}
        \left|e_3(|u|^2|\hch|^2)\right|\les\frac{a}{|u|^2}.
    \end{align*}
    Recall that
    \begin{align}\label{Ome3bbb}
        \Om e_3=\pr_u+\bbb^A\pr_A,
    \end{align}
    where $\bbb^A$ satisfies for $A=1,2$
    \begin{align*}
        e_4(\bbb^A)=-4\Om\ze^A.
    \end{align*}
    Thus, we have
    \begin{align*}
        |\bbb^A|\les\int_0^\ub |4\Om\ze^A|d\ub' \les \frac{\af}{|u|^3}.
    \end{align*}
Then, we deduce from Theorem \ref{maintrapped}
\begin{align*}
\left|\pr_u(|u|^2|\hch|^2)\right|\les\left|\bbb^A\pr_A(|u|^2|\hch|^2)\right|+\left|e_3(|u|^2|\hch|^2)\right|\les\frac{a}{|u|^2}.
\end{align*}
Integrating it from $u_\infty$ to $u$, we deduce for $a\gg 1$
\begin{align*}
|u|^2|\hch|^2\geq|u_\infty|^2|\hch|^2\big|_{H_{u_\infty}}-\frac{a^\frac{3}{2}}{|u|}.
\end{align*}
Combining with \eqref{geqathm}, we obtain for $|u|\geq \frac{a}{4}$ and $a\gg 1$
\begin{align*}
    \int_0^1 |u|^2|\hch|^2(u,\ub',x^1,x^2)d\ub'&\geq \int_0^1|u_\infty|^2|\hch|^2(u_\infty,\ub',x^1,x^2)d\ub'-\frac{a^\frac{3}{2}}{|u|}\\
    &\geq a-\frac{a^\frac{3}{2}}{|u|}\\
    &\geq\frac{3a}{4}.
\end{align*}
Taking $u=-\frac{a}{4}$, we have
    \begin{align*}
        \int_0^1|\hch|^2\left(-\frac{a}{4},\ub',x^1,x^2\right)d\ub'\geq\frac{12}{a}.
    \end{align*}
    Finally, we have from Proposition \ref{nulles} and \eqref{nullidentities}
    \begin{align*}
        \pr_\ub(\Om^{-1}\trch)=-\frac{1}{2}(\trch)^2-|\hch|^2\leq -|\hch|^2.
    \end{align*}
    Recalling that we have on $\Hb_0$
    \begin{align*}
        \Om^{-1}\trch\left(-\frac{a}{4},0,x^1,x^2\right)=\frac{8}{a},
    \end{align*}
    we infer
    \begin{align}\label{trchne}
        \Om^{-1}\trch\left(-\frac{a}{4},1,x^1,x^2\right)\leq\frac{8}{a}-\int_0^1|\hch|^2\left(-\frac{a}{4},\ub',x^1,x^2\right)d\ub'\leq-\frac{4}{a}<0.
    \end{align}
    Next, we have from Theorem \ref{maintrapped} that
    \begin{align*}
        \trchb+\frac{2}{|u|}\les \frac{a}{|u|^3}\les a^{-2},
    \end{align*}
    which implies for $a\gg 1$
    \begin{align*}
        \Om\trchb\leq -\frac{2}{|u|}+\frac{4}{a}.
    \end{align*}
    Taking $u=-\frac{a}{4}$, we obtain
    \begin{align}\label{trchbne}
        \Om\trchb\left(-\frac{a}{4},1\right)\leq -\frac{8}{a}+\frac{4}{a}=-\frac{4}{a}<0.
    \end{align}
    Combining \eqref{trchne} and \eqref{trchbne}, we conclude that $S_{-\frac{a}{4},1}$ is a trapped surface. This concludes the proof of Theorem \ref{thmtrappedsurface}.
\end{proof}
\subsection{Short-pulse cone}
We have the following scaling version of Theorems \ref{maintrapped} and \ref{thmtrappedsurface}.
\begin{lem}\label{thmscaling}
     There exists a sufficiently large $a_0>0$. Let $a>a_0$, $0<\de\leq a^{-1}$ and $u_\infty\leq -\de a$. Assume that there are an initial data that satisfy: \begin{itemize}
    \item The following assumption holds along $u=u_\infty$:
    $$
    \sum_{i\leq 10,k\leq 4}\afd|u_\infty|\left\|(\de\nabs_{4})^k(|u_\infty|\nabs)^i\hch\right\|_{L^\infty(S_{u_\infty,\ub})}\leq 1,\qquad \forall\; \ub\in [0,\de].
    $$
    \item Minkowskian initial data along $\ub=0$.
    \item The following lower bound condition holds along $u=u_\infty$:
    \begin{align*}
        \int_0^\de |u_\infty|^2|\hch|^2(u_\infty,\ub,x^{1},x^{2})d\ub\geq \de a,\qquad\forall\; (x^1,x^2).
    \end{align*}
    \end{itemize}
    \begin{figure}[H]
    \centering
    \begin{tikzpicture}[scale=0.9]
\fill[yellow] (1,1)--(2,2)--(4,0)--(3,-1)--cycle;

\draw[very thick] (1,1)--(2,2)--(4,0)--(3,-1)--cycle;

\node[scale=1, black] at (4.1,-1) {$H_{u_\infty}$};
\node[scale=1, black] at (1.1,1.9) {$H_{-\frac{\de a}{4}}$};
\node[scale=1, black] at (3.2,1.4) {$\Hb_\de$};
\node[scale=1, black] at (1.7,-0.3) {$\Hb_0$};

\draw[dashed] (0,-4)--(0,4);
\draw[dashed] (0,-4)--(4,0)--(0,4);
\draw[dashed] (0,0)--(2,2);
\draw[dashed] (0,2)--(3,-1);

\draw[->] (3.3,-0.6)--(3.6,-0.3) node[midway, sloped, above, scale=1] {$e_4$};
\draw[->] (2.4,-0.3)--(1.7,0.4) node[midway, sloped, above, scale=1] {$e_3$};

\filldraw[red] (2,2) circle (2pt);
\filldraw[red] (4,0) circle (2pt);

\draw[thin] (2.1,2.5) node[above, scale=1] {$S_{-\frac{\de a}{4},\de}$} -- (2,2);
\draw[thin] (4.1,0.5) node[above, scale=1] {$S_{u_\infty,\de}$} -- (4,0);
\end{tikzpicture}
\caption{{\footnotesize The initial conditions in Lemma \ref{thmscaling} lead to trapped surface ($S_{-\frac{\de a}{4},\de}$) formation in the future of the $H_{u_\infty}$ and $\Hb_0$.}}
\label{scalingfigure}
\end{figure}
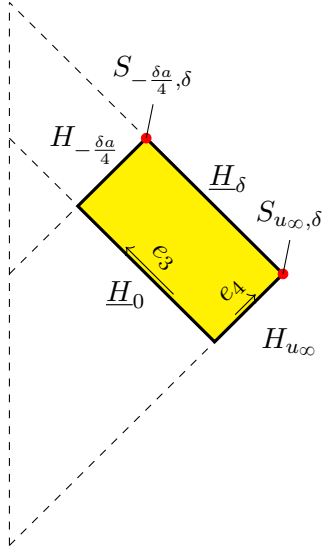
Then, \eqref{EVE} admits a unique solution in the colored region of Figure \ref{scalingfigure}. Moreover, we have:
\begin{itemize}
    \item The $2$--sphere $S_{-\frac{\de a}{4},\de}$ is a trapped surface.
    \item The following estimates hold in the colored region of Figure \ref{scalingfigure}:
\begin{align*}
    |\a|&\les\frac{\af}{\de|u|},\qquad |\b|\les\frac{a^\frac{1}{2}}{|u|^2}, \qquad |\rho,\si|\les\frac{\de a}{|u|^3},\qquad |\bb|\les\frac{\de^2 a^\frac{3}{2}}{|u|^4},\qquad |\aa|\les\frac{\de^3 a^2}{|u|^5},\\
    |\hch|&\les\frac{\af}{|u|},\qquad\;\,|\om|\les\frac{1}{|u|},\qquad\qquad\qquad\quad\,|\eta,\etab,\ze,\hchb|\les\frac{\de\af}{|u|^2},\qquad\;\,|\omb|\les\frac{\de^2a}{|u|^3},\\
    |\trchc|&\les\frac{\de a}{|u|^2},\quad\, |\trchbc|\les\frac{\de}{|u|^2}.
\end{align*}
\end{itemize}
Moreover, analog estimates also hold for their $H^{s+4}(S_{u,\ub})$--norms.
\end{lem}
\begin{proof}
It follows from Theorem \ref{maintrapped} and the following standard spacetime rescaling:
\begin{align*}
    u':=\de u,\qquad\quad\ub':=\de\ub,\qquad\quad{\phi'}^A:=\de\phi^A.
\end{align*}
See Section 8 in \cite{An} or Section 10 in \cite{ShenWan}.
\end{proof}
We are now ready to prove the following theorem, which constructs the short-pulse cone described in Figure \ref{3Dshortpulsecone}.
\begin{thm}\label{shortpulsecone}
    There exists a sufficiently large $a_0>0$. Let $a>a_0$, $0<\de\leq a^{-1}$ and $u_0\in[-2,-1]$. Assume that there is an initial data on $H_{u_0}\cup\Hb_{0}$ that satisfies:
    \begin{itemize}
    \item The following assumption holds along $H_{u_0}$:
    $$
    \sum_{i+j\leq s+10}\afd\left\|(\de\nabs_4)^j(|u_0|\nabs)^i\hch\right\|_{L^\infty(S_{u_0,\ub})}\leq 1.
    $$
    \item Minkowskian initial data along $\Hb_0$.
    \item The following lower bound condition holds along $H_{u_0}$:
    \begin{align}\label{geqa}
    \int_0^\de |u_0|^2|\hch|^2(u_0,\ub,x^{1},x^{2})d\ub\geq \de a,\qquad\forall\; (x^1,x^2).
    \end{align}
    \end{itemize}
    \begin{figure}[H]
    \centering
\begin{tikzpicture}[scale=1, decorate]
    \draw[orange] (0,3.4) ellipse (0.6 and 0.06);
    \draw[->, thick, rounded corners=8pt] (1.5,4) to[out=-90, in=90] (0.2,3.4);
    \node[above] at (2,3.75) {\footnotesize Trapped Surface $S_{-\frac{\de a}{4},\de}$};
    \draw[dashed] (0.1,2.9) arc[start angle=0,end angle=180,x radius=0.1,y radius=0.01];
    \draw (-0.1,2.9) arc[start angle=180,end angle=360,x radius=0.1,y radius=0.01];
    \draw[dashed] (4,0) arc[start angle=0,end angle=180,x radius=4,y radius=0.4];
    \draw (-4,0) arc[start angle=180,end angle=360,x radius=4,y radius=0.4];
\begin{scope}
    \fill[white] 
        (3,0) arc[start angle=0,end angle=180,x radius=3,y radius=0.3] --
        (-3,0) arc[start angle=180,end angle=360,x radius=3,y radius=0.3] -- cycle;
\end{scope}
    \draw[dashed] (3.5,-0.5) arc[start angle=0,end angle=180,x radius=3.5,y radius=0.35];
    \draw (-3.5,-0.5) arc[start angle=180,end angle=360,x radius=3.5,y radius=0.35];
    \draw[dashed] (2,-2) arc[start angle=0,end angle=180,x radius=2,y radius=0.2];
    \draw (-2,-2) arc[start angle=180,end angle=360,x radius=2,y radius=0.2];
    \draw[dashed] (0,-5) -- (0,4.4) node[above]{\footnotesize $u=\ub$};
    \draw (0,-4) -- (4,0);
    \draw (-4,0) -- (0,-4);
    \draw (4,0) -- (0.6,3.4);
    \draw (-4,0) -- (-0.6,3.4);
    \draw[dashed] (3.5,-0.5) -- (0.1,2.9);
    \draw[dashed] (0.6,3.4) -- (0.1,2.9);
    \draw[dashed] (-0.6,3.4) -- (-0.1,2.9);
    \draw[dashed] (-3.5,-0.5) -- (-0.1,2.9);
    \draw[->, thick, rounded corners=8pt] (4,1) to[out=-90, in=90] (3,0);
    \node[above] at (4,1) {\footnotesize $\Hb_0$};
    \draw[->, thick, rounded corners=8pt] (5,1) to[out=-90, in=90] (4,0);
    \node[above] at (5,1) {\footnotesize $S_{u_0,\de}$};
    \node[below right] at (4,0) {\footnotesize $\ub=\de$};
    \node[below right] at (3.5,-0.5) {\footnotesize $\ub=0$};
    \node[below right] at (2.7,-1.3) {\footnotesize $H_{u_0}$};
\draw[->, decorate, decoration={snake, amplitude=0.5mm, segment length=2mm}, thin, red]  (4.1,-0.3) -- (3.4,0.4);
\draw[->, decorate, decoration={snake, amplitude=0.5mm, segment length=2mm}, thin, red]  (3.9,-0.5) -- (3.2,0.2);
\draw[->, decorate, decoration={snake, amplitude=0.5mm, segment length=2mm}, thin, red]  (-4.1,-0.3) -- (-3.4,0.4);
\draw[->, decorate, decoration={snake, amplitude=0.5mm, segment length=2mm}, thin, red]  (-3.9,-0.5) -- (-3.2,0.2);
\node[below left] at (-3.5,-0.5)  {\footnotesize short-pulse};
    \fill[red!20,opacity=0.35](3.5,-0.5) -- (0.1,2.9) -- (0.6,3.4) -- (4,0) -- cycle;
    \fill[red!20,opacity=0.35](-3.5,-0.5) -- (-0.1,2.9) -- (-0.6,3.4) -- (-4,0) -- cycle;
\end{tikzpicture}
\caption{Short-pulse cone. The red wiggle lines are the short-pulse input; the red region is the region where semi-global existence result is established; the orange circle denotes the trapped surface $S_{-\frac{\de a}{4},\de}$.}
\label{3Dshortpulsecone}
\end{figure}
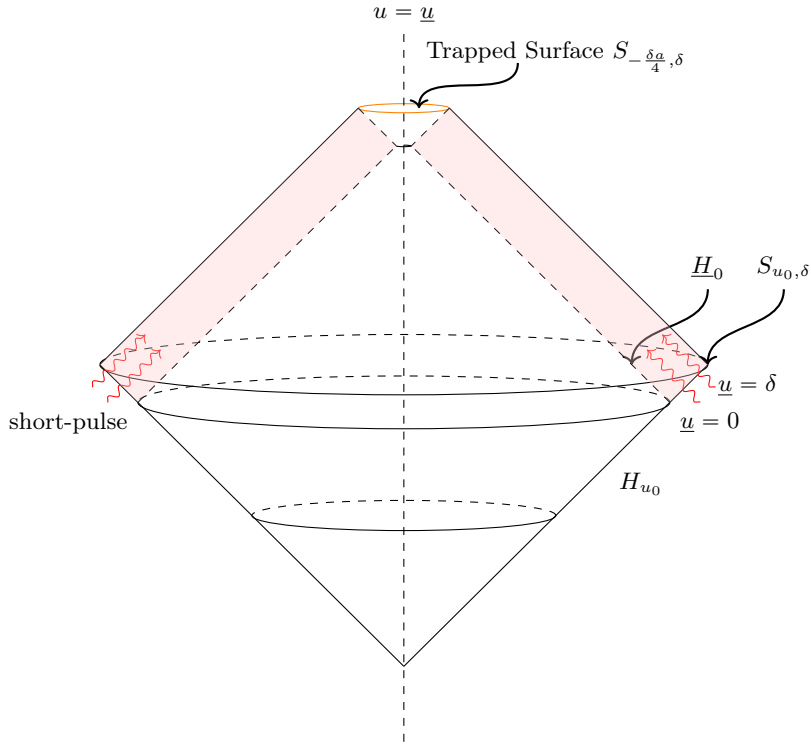
    Then \eqref{EVE} admits a unique solution in the region 
    $$
    \left\{(u,\ub)\in\left[u_0,-\frac{\de a}{4}\right]\times[0,\de]\right\}.
    $$
    Moreover, we have:
\begin{itemize}
    \item The $2$--sphere $S_{-\frac{\de a}{4},\de}$ is a trapped surface.
    \item The following estimates hold for $u\in[u_0,-\frac{1}{2}]$:
\begin{align*}
    |\a|&\les\de^{-1}\afd,\qquad |\b|\les\af, \qquad |\rho,\si|\les\de a,\qquad |\bb|\les\de^2a^\frac{3}{2},\qquad |\aa|\les\de^3a^2,\\
    |\hch|&\les\af,\qquad\qquad|\om|\les 1,\qquad\qquad\quad\;\;\;|(\eta,\etab,\ze,\hchb)|\les\de\af,\qquad\;\,|\omb|\les\de^2a,\\
    |\log\Om|&\les\de,\qquad\qquad\; |\bbb^A|\les\de^2\af,\qquad\quad\;\;\; |\slg_{AB}-\slgo_{AB}|\les\de\af,\\
    |\trchc|&\les\de a,\qquad\quad\, |\trchbc|\les\de.
\end{align*}
Moreover, the same estimates also hold for their $H^{s+4}(S_{u,\ub})$--norms.
\end{itemize}
\end{thm}
\begin{proof}
The fact that $S_{-\frac{\de a}{4},\de}$ is a trapped surface and the estimates for the curvature components and the Ricci coefficients follow by taking $u_\infty=u_0$ in Lemma \ref{thmscaling}. Next, we have
\begin{align*}
    \pr_\ub(\log\Om)=-2\om,\qquad \pr_\ub(\bbb^A)=-4\Om^2\ze^A,\qquad \pr_\ub(\slg_{AB})=2\Om\chi_{AB}.
\end{align*}
Integrating them from along $H_u^{(0,\ub)}$ and combining with the estimates for $\om$, $\trch$, $\hch$ and $\ze$, we infer
\begin{align*}
    \left\|\nabs^{\leq s+4}\log\Om\right\|_{L^2(S_{u,\ub})}&\les\de,\qquad\left\|\nabs^{\leq s+4}\bbb^A\right\|_{L^2(S_{u,\ub})}\les \de^2\af,\\
    \left\|\nabs^{\leq s+4}\left(\slg_{AB}-\slgo_{AB}\right)\right\|_{L^2(S_{u,\ub})}&\les\de\af.
\end{align*}
This concludes the proof of Theorem \ref{shortpulsecone}.
\end{proof}
\section{Transition region and barrier annulus}\label{secstability}
In this section, we first construct in Theorem \ref{mainstability} a transition region based on the short-pulse region obtained in Section \ref{sectrapped}. Then, in Theorem \ref{interiorsolution}, we construct a spacelike initial data set by taking a constant-time slice in the region obtained in Theorem \ref{mainstability}.
\subsection{Construction of transition region}
Throughout this section, we always assume that $u\in[u_0,-\frac{1}{2}]$\footnote{Recall that $u_0\in[-2,-1]$ is introduced in Theorem \ref{shortpulsecone}.} and $\ub\in[\de,1]$. We define
\begin{align*}
    V:=V(u,\ub):=\big\{(u',\ub')\in\left[u_0,u\right]\times[\de,\ub]\big\},\qquad V_*:=V\left(-\frac{1}{2},1\right).
\end{align*}
The main goal of this section is to prove the following theorem.
\begin{thm}\label{mainstability}
    There exists a sufficiently large $a_0>0$. Let $a>a_0$, $0<\de\leq a^{-2}$, $u_0\in[-2,-1]$ and we denote
    \begin{equation}\label{dfep0}
        \ep_0:=a^{-1}\geq \de a.
    \end{equation}
    Assume that there are initial data on $H_{u_0}\cup\Hb_{0}$ that satisfy: \begin{itemize}
    \item $H_{u_0}^{(0,1)}$ is endowed with an outgoing geodesic foliation.
    \item The following assumption holds along $H_{u_0}^{(0,\de)}$:
    \begin{equation}\label{hchassumption}
    \sum_{i+j\leq s+10}\afd\left\|(\de\nabs_4)^j(|u_0|\nabs)^i\hch\right\|_{L^\infty(S_{u_0,\ub})}\leq 1.
    \end{equation}
    \item The following conditions holds along $H_{u_0}^{(\de,1)}$:
    \begin{align}\label{vanishingconditions}
        \hch=0,\qquad\quad \a=0,\qquad\quad \om=0.
    \end{align}
    \item Minkowskian initial data along $\Hb_0$.
    \end{itemize}
    Then \eqref{EVE} admits a unique solution in $\left\{(u,\ub)\in\left[u_0,-\frac{1}{2}\right]\times[0,1]\right\}$ and the following estimates hold in $V_*$:\footnote{Here, $\slg$ denotes the induced metric of $\g$ on $S_{u,\ub}$.}
\begin{align*}
\left\|\nabs^{\leq s+2}(\a,\b,\rho,\si,\bb,\aa)\right\|_{L^2(S_{u,\ub})}&\les\ep_0,\\
\left\|\nabs^{\leq s+3}(\trcht,\trchbt,\hch,\hchb,\om,\omb,\eta,\etab,\ze,\log\Om,\bbb,\slgc)\right\|_{L^2(S_{u,\ub})}&\les\ep_0,
\end{align*}
where we denote
\begin{align}\label{dftrchttrchbt}
\trcht:=\trch-\frac{2}{r},\qquad\quad \trchbt:=\trchbt+\frac{2}{r},\qquad\quad r:=\ub-u.
\end{align}
\end{thm}
\subsection{Fundamental norms}
In this section, we define the fundamental norms in $V_*$.
\subsubsection{Schematic notation \texorpdfstring{$\Ga$}{}, \texorpdfstring{$\Rl$}{} and \texorpdfstring{$\Rr$}{}}
The following schematic notations will be used frequently throughout Section \ref{secstability}.
\begin{df}\label{gamma}
We introduce the following schematic notations:
\begin{align*}
    \Ga&:=\left\{\trcht,\,\trchbt,\,\eta,\,\etab,\,\ze,\,\om,\,\omb,\,\log\Om\right\},\\
    \Rl&:=\left\{\a,\,\b,\,\rho,\,\si,\,\bb\right\},\\
    \Rr&:=\left\{\b,\,\rho,\,\si,\,\bb,\,\aa\right\}.
\end{align*}
We also denote:
\begin{align*}
\Ga^{(1)}:=\nabs^{\leq 1}\Ga\cup\{\a,\,\b,\,\rho,\,\si,\,\bb,\,\aa\}.
\end{align*}
Finally, we define for $i\geq 1$
\begin{align*}
    \Ga^{(i+1)}:=\nabs^{\leq 1}\Ga^{(i)},\qquad\quad \Rl^{(i+1)}:=\nabs^{\leq 1}\Rl^{(i)},\qquad\quad\Rr^{(i+1)}:=\nabs^{\leq 1}\Rr^{(i)}.
\end{align*}
\end{df}
\subsubsection{\texorpdfstring{$\Rk$}{} norms (\texorpdfstring{$L^2$}{}--flux of curvature)}\label{secRnor}
We define for $i\in\mathbb{N}$
\begin{align*}
\Rk_i(u,\ub):=\left\|\Rl^{(i)}\right\|_{L^2(\cuvs)},\qquad\quad \Rkb_i(u,\ub):=\left\|\Rr^{(i)}\right\|_{L^2(\ucuvs)}.
\end{align*}
Then, we denote
\begin{align*}
    \Rk:=\sup_{V_*}\sum_{i=0}^{s+3}\big(\Rk_i(u,\ub)+\Rkb_i(u,\ub)\big).
\end{align*}
\subsubsection{\texorpdfstring{$\Ok$}{} norms (\texorpdfstring{$L^2(S_{u,\ub})$}{}--norms of geometric quantities)}\label{secOnor}
We define for $i\in\mathbb{N}$
\begin{align*}
    \Ok_i(u,\ub):=\left\|\Ga^{(i)}\right\|_{L^2(S_{u,\ub})}.
\end{align*}
Then, we denote
\begin{align*}
    \Ok:=\sup_{V_*}\sum_{i=0}^{s+3}\Ok_i(u,\ub).
\end{align*}
\subsubsection{\texorpdfstring{$\Ok_{(0)}$}{} and \texorpdfstring{$\Rk_{(0)}$}{} norms (Initial data)}\label{initialO}
We introduce the following norms on $H_{u_0}^{(\de,1)}\cup\Hb_\de^{(u_0,-\frac{1}{2})}$:
\begin{align*}
\Ok_{(0)}&:=\sup_{\ub\in[\de,1]}\sum_{i=0}^{s+3}\Ok_i(u_0,\ub)+\sup_{u\in[u_0,-\frac{1}{2}]}\sum_{i=0}^{s+3}\Ok_i(u,\de),\\
\Rk_{(0)}&:=\sup_{\ub\in[\de,1]}\sum_{i=0}^{s+3}\Rk_i(u_0,\ub)+\sup_{u\in[u_0,-\frac{1}{2}]}\sum_{i=0}^{s+3}\Rkb_i(u,\de).
\end{align*}
\subsection{Main intermediate results}
Throughout Section \ref{secstability}, we always denote
\begin{align}\label{dfep}
    \ep:=\ep_0^\frac{2}{3},
\end{align}
which measure the size of bootstrap bounds.
\begin{thm}\label{MM0}
Under the assumptions of Theorem \ref{mainstability}. Assume in addition that
\begin{equation}
\Ok\leq \ep,\qquad\quad\Rk\leq\ep.
\end{equation}
Then, we have
\begin{equation}
    \Ok_{(0)}\les\ep_0,\qquad\quad \Rk_{(0)}\les\ep_0.
\end{equation}
\end{thm}
Theorem \ref{MM0} is proved in Section \ref{secII}. The estimates for normal quantities on $\Hb_\de$ follow directly from Theorem \ref{shortpulsecone}. The abnormal quantities can be estimated by their transport equations in the direction $e_3$. On the other hand, all quantities on $H_{u_0}$ can be estimated by integrating their equations along $H_{u_0}$ and combining with \eqref{vanishingconditions}.
\begin{thm}\label{MM1}
Under the assumptions of Theorem \ref{mainstability}. Assume in addition that
\begin{align}
    \Ok_{(0)}\les\ep_0,\qquad \Rk_{(0)}\les\ep_0,\qquad \Ok\leq\ep,\qquad \Rk\leq\ep.
\end{align}
Then, we have
\begin{equation}
    \Rk\les\ep_0.
\end{equation}
\end{thm}
Theorem \ref{MM1} is proved in Section \ref{secRR}. The proof is based on the energy estimates for the Bianchi equations. The method used here can also be considered as the $r^p$--weighted estimates used in \cite{holzegel,Shen22,Shen24} in the particular case $p=0$.
\begin{thm}\label{MM2}
Under the assumptions of Theorem \ref{mainstability}. Assume in addition that
\begin{align}
    \Ok_{(0)}\les\ep_0,\qquad \Rk_{(0)}\les\ep_0,\qquad \Ok\leq\ep,\qquad \Rk\les\ep_0.
\end{align}
Then, we have
\begin{equation}
    \Ok\les\ep_0.
\end{equation}
\end{thm}
Theorem \ref{MM2} is proved in Section \ref{secOO}. The proof is done by integrating the null structure equations and the Bianchi equations along the outgoing and incoming null cones.
\subsection{Bootstrap assumptions and first consequences}\label{secbootboot}
In the rest of Section \ref{secstability}, we always make the following bootstrap assumptions:
\begin{align}\label{BB}
    \Ok\leq\ep,\qquad\quad\Rk\leq\ep.
\end{align}
The following consequences of \eqref{BB} will be used frequently throughout this paper.
\begin{lem}\label{decayGa}
    We have the following estimates:
    \begin{align*}
        \|\Ga^{(s+3)}\|_{L^2(S_{u,\ub})}\les\ep,\qquad \|\Rl^{(s+2)}\|_{L^2(\cuvs)}\les\ep,\qquad \|\Rr^{(s+2)}\|_{L^2(\ucuvs)}\les\ep.
    \end{align*}
\end{lem}
\begin{proof}
    It follows directly from \eqref{BB} and Definition \ref{gamma}.
\end{proof}
\begin{lem}\label{evolutionlemma}
We have the following transport estimates:
\begin{enumerate}
\item Let $U,F\in\sk_k$ satisfying the outgoing transport equation
\begin{equation}
    \nabs_4 U +\la_0\trch\, U=F,
\end{equation}
where $\la_0\geq 0$. We have
\begin{equation}\label{transubU}
\|U\|_{L^2(S_{u,\ub})}\les\|U\|_{L^2(S_{u,\de})}+\int_\de^\ub\|F\|_{L^2(S_{u,\ub'})}d\ub'.
\end{equation}
\item Let $V,\Fb\in\sk_k$ satisfying the incoming transport equation
\begin{equation}\label{nabs3V}
\nabs_3 V+\la_0\trchb\, V=\Fb,
\end{equation}
where $\la_0\geq 0$. We have
\begin{equation}\label{transuV}
    \|V\|_{L^2(S_{u,\ub})}\les \|V\|_{L^2(S_{u_0,\ub})} + \int_{u_0}^{u}\|\Fb\|_{L^2(S_{u',\ub})}du'.
\end{equation}
\end{enumerate}
\end{lem}
\begin{proof}
See Lemma 4.1.5 in \cite{kn} and Lemma 6.1 in \cite{Shen24}.
\end{proof}
\begin{prop}\label{sobolevinequality}
We have the following inequality:
\begin{align*}
    \|\phi\|_{L^\infty(S_{u,\ub})}\les\|\phi^{(2)}\|_{L^2(S_{u,\ub})}.
\end{align*}
\end{prop}
\begin{proof}
See Lemma 4.1.3 of \cite{kn}.
\end{proof}
\begin{prop}\label{commu}
We have the following schematic commutator formulae:
\begin{align*}
    [\Om\nabs_4,r\nabs]&=\Ga\c\nabs+\Ga^{(1)},\\
    [\Om\nabs_3,r\nabs]&=\Ga\c\nabs+\Ga^{(1)}.
\end{align*}
\end{prop}
\begin{proof}
    It follows directly from Lemma \ref{comm} and Definition \ref{gamma}.
\end{proof}
\subsection{Estimates for characteristic initial data}\label{secII}
In this section, we prove Theorem \ref{MM0}. More precisely, we prove the following propositions, which estimate, respectively, the norms on $\Hb_\de^{(u_0,-\frac{1}{2})}$ and $H_{u_0}^{(\de,1)}$.
\begin{prop}\label{Hbde}
We have the following estimates on $\Hb_\de^{(u_0,-\frac{1}{2})}$:
    \begin{align*}
        \left\|\nabs^{\leq s+3}\left(\b,\rho,\si,\bb,\aa,\trcht,\trchbt,\hch,\hchb,\om,\omb,\eta,\etab,\ze,\log\Om,\bbb,\slgc\right)\right\|_{L^2(S_{u,\ub})}&\les\ep_0,\\
       \left\|\nabs^{\leq s+2}\a\right\|_{L^2(S_{u,\ub})}&\les\ep_0.
    \end{align*}
\end{prop}
\begin{proof}
As a consequence of Theorem \ref{shortpulsecone} and \eqref{dfep0}, we have\footnote{The difference between $(\trchc,\trchbc)$ and $(\trcht,\trchbt)$ is controlled by $||u|-r|\les\de\ll\ep_0$.}
    \begin{align}\label{consequenceoftrapped}
        \left\|\nabs^{\leq s+4}\left(\rho,\si,\bb,\aa,\trcht,\trchbt,\hchb,\omb,\eta,\etab,\ze,\log\Om,\bbb,\slgc\right)\right\|_{L^2(S_{u,\ub})}\les\ep_0.
    \end{align}
    It remains to estimate $\hch$, $\om$, $\b$ and $\a$ on $\Hb_\de^{(u_0,-\frac{1}{2})}$. We have from Propositions \ref{nulles}, \ref{commutation} and \eqref{consequenceoftrapped}
    \begin{align*}
        \left\|\nabs_3\nabs^{\leq s+3}(\hch,\om,\b)\right\|_{L^2(S_{u,\ub})}\les\ep_0.
    \end{align*}
    Recall that we have\footnote{The estimate for $\b$ follows from \eqref{codazzi} and the initial assumptions \eqref{hchassumption} and \eqref{vanishingconditions}.}
    \begin{align*}
        \hch=0,\qquad \om=0,\qquad |\nabs^{\leq s+3}\b|\les\ep_0\quad\mbox{ on }\; S_{u_0,\de}.
    \end{align*}
    Thus, we obtain by integrating them along $\Hb_\de$:
    \begin{align}\label{esthchomb}
        \left\|\nabs^{\leq s+3}(\hch,\om,\b)\right\|_{L^2(S_{u,\ub})}\les\ep_0\qquad \mbox{ on }\; \Hb_{\de}^{(u_0,-\frac{1}{2})}.
    \end{align}
    Finally, we have from Proposition \ref{bianchiequations} and \eqref{esthchomb}
    \begin{align*}
        \left\|\nabs_3\nabs^{\leq s+2}\a\right\|_{L^2(S_{u,\ub})}\les\ep_0.
    \end{align*}
    Integrating it along $\Hb_\de$ and applying \eqref{vanishingconditions}, we infer
    \begin{align*}
        \left\|\nabs^{\leq s+2}\a\right\|_{L^2(S_{u,\ub})}\les\ep_0.
    \end{align*}
    This concludes the proof of Proposition \ref{Hbde}.
\end{proof}
\begin{prop}\label{Hu0}
    We have the following estimates on $H_{u_0}^{(\de,1)}$:
    \begin{align}\label{Hu0eq}
        \left|\dk^{\leq s+3}\left(\b,\rho,\si,\bb,\trcht,\trchbt,\hchb,\omb,\ze\right)\right|\les\ep_0,
    \end{align}
    where we denote
    \begin{equation}\label{dfdk}
    \dk:=\{\nabs_4,\nabs_3,\nabs\}.
    \end{equation}
\end{prop}
\begin{proof}
Recall that $H_{u_0}$ is endowed with an outgoing geodesic foliation and we have
\begin{align}\label{hchaomvanishing}
    \hch=0,\qquad\quad \a=0,\qquad\quad\om=0,\qquad\quad\etab=-\ze\quad \mbox{ on }\; H_{u_0}^{(\de,1)}.
\end{align}
Proceeding as in Chapter 2 of \cite{Chr},\footnote{More precisely, all quantities can be estimated by their integrating transport equations of $\nabs_4$ along $H_{u_0}$ and applying \eqref{hchaomvanishing}.} we obtain \eqref{Hu0eq} as stated.
\end{proof}
Combining Propositions \ref{Hbde} and \ref{Hu0}, this concludes the proof of Theorem \ref{MM0}.
\subsection{Energy estimates for curvature components}\label{secRR}
In this section, we prove Theorem \ref{MM1} based on the energy estimates for the Bianchi equations. The strategy used here can be considered as the $r^p$--weighted estimates used in \cite{holzegel,Shen22,Shen24} in the particular case $p=0$.
\subsubsection{Estimates for general Bianchi pairs}
The following lemma provides the general structure of Bianchi pairs.
\begin{lem}\label{keypoint4}
Let $k=1,2$ and $a_{(1)}$, $a_{(2)}$ be real numbers. Then, we have the following properties:
\begin{enumerate}
    \item If $\psi_{(1)},h_{(1)}\in\sk_k$ and $\psi_{(2)},h_{(2)}\in \sk_{k-1}$ satisfying
    \begin{align}
    \begin{split}\label{bianchi14}
    \nabs_3(\psi_{(1)})+a_{(1)}\trchb\,\psi_{(1)}&=-k\sld_k^*(\psi_{(2)})+h_{(1)},\\
    \nabs_4(\psi_{(2)})+a_{(2)}\trch\,\psi_{(2)}&=\sld_k(\psi_{(1)})+h_{(2)}.
    \end{split}
    \end{align}
Then, the pair $(\psi_{(1)},\psi_{(2)})$ satisfies
\begin{align}
\begin{split}\label{div4}
    &\bdiv(|\psi_{(1)}|^2e_3)+k\bdiv(|\psi_{(2)}|^2e_4)\\
    +&(2a_{(1)}-1)\trchb|\psi_{(1)}|^2+k(2a_{(2)}-1)\trch|\psi_{(2)}|^2\\
    =&2k\sdivs(\psi_{(1)}\cdot\psi_{(2)})+2\psi_{(1)}\cdot h_{(1)}+2k\psi_{(2)}\cdot h_{(2)}+\Ga\left(|\psi_{(1)}|^2+|\psi_{(2)}|^2\right).
\end{split}
\end{align}
    \item If $\psi_{(1)},h_{(1)}\in\sk_{k-1}$ and $\psi_{(2)},h_{(2)}\in\sk_k$ satisfying
\begin{align}
\begin{split}\label{bianchi24}
\nabs_3(\psi_{(1)})+a_{(1)}\trchb\,\psi_{(1)}&=\sld_k(\psi_{(2)})+h_{(1)},\\
\nabs_4(\psi_{(2)})+a_{(2)}\trch\,\psi_{(2)}&=-k\sld_k^*(\psi_{(1)})+h_{(2)}.
\end{split}
\end{align}
Then, the pair $(\psi_{(1)},\psi_{(2)})$ satisfies
\begin{align}
\begin{split}\label{div24}
    &k\bdiv(|\psi_{(1)}|^2e_3)+\bdiv(|\psi_{(2)}|^2e_4)\\
    +&k(2a_{(1)}-1)\trchb|\psi_{(1)}|^2+(2a_{(2)}-1)\trch|\psi_{(2)}|^2\\
    =&2\sdivs(\psi_{(1)}\cdot\psi_{(2)})+2k\psi_{(1)}\cdot h_{(1)}+2\psi_{(2)}\cdot h_{(2)}+\Ga\left(|\psi_{(1)}|^2+|\psi_{(2)}|^2\right).
\end{split}
\end{align}
\end{enumerate}
\end{lem}
\begin{proof}
It follows by taking $p=0$ in Lemma 4.2 of \cite{Shen22}.
\end{proof}
\begin{prop}\label{keyintegral4}
Let $(\psi_{(1)},\psi_{(2)})$ be a Bianchi pair that satisfies \eqref{bianchi14} or \eqref{bianchi24}. Then, we have the following properties:
\begin{itemize}
\item In the case $1-2a_{(1)}\geq 0$ and $2a_{(2)}-1\geq 0$, we have
\begin{align}
\begin{split}\label{casepp}
&\quad\;\|\psi_{(1)}\|_{L^2(\cuvs)}^2+\|\psi_{(2)}\|^2_{L^2(\ucuvs)}\\
&\les\|\psi_{(1)}\|_{L^2(H_{u_0}^{(\de,\ub)})}^2+\|\psi_{(2)}\|^2_{L^2(\Hb_{\de}^{(u_0,u)})}\\
&+\int_V|\psi_{(1)}\c h_{(1)}|+|\psi_{(2)}\c h_{(2)}|+|\Ga|\left(|\psi_{(1)}|^2+|\psi_{(2)}|^2\right).
\end{split}
\end{align}
\item In the case $1-2a_{(1)}\leq 0$ and $2a_{(2)}-1>0$, we have
\begin{align}
\begin{split}\label{casenp}
&\quad\;\|\psi_{(1)}\|_{L^2(\cuvs)}^2+\|\psi_{(2)}\|^2_{L^2(\ucuvs)}\\
&\les\|\psi_{(1)}\|_{L^2(H_{u_0}^{(\de,\ub)})}^2+\|\psi_{(2)}\|^2_{L^2(\Hb_{\de}^{(u_0,u)})}\\
&+\int_V|\psi_{(1)}|^2+|\psi_{(1)}\c h_{(1)}|+|\psi_{(2)}\c h_{(2)}|+|\Ga|\left(|\psi_{(1)}|^2+|\psi_{(2)}|^2\right).
\end{split}
\end{align}
\end{itemize}
\end{prop}
\begin{proof}
It follows by integrating \eqref{bianchi14} or \eqref{bianchi24} in $V$ and applying Stokes' formula, see Proposition 5.7 of \cite{Shen24}.
\end{proof}
\begin{df}\label{dfdkbb}
We define the weighted angular derivatives $\dkbb$ as follows:
\begin{align*}
    \dkbb U &:= r\sld_2 U,\qquad \forall U\in \sk_2,\\
    \dkbb \xi&:=r\sld_1 \xi,\qquad\,\,\, \forall \xi\in \sk_1,\\
    \dkbb f&:=r\sld_1^* f,\qquad \,\,\forall f\in \sk_0.
\end{align*}
We denote for any tensor $h\in\sk_k$, $k=0,1,2$,
\begin{equation*}
    h^{(0)}:=h,\qquad\quad h^{(i)}:=(h,\dkbb h,...,\dkbb^i h).
\end{equation*}
\end{df}
\subsubsection{Estimates for the Bianchi pair \texorpdfstring{$(\a,\b)$}{}}
\begin{prop}\label{estabab}
We have the following estimate:
    \begin{align*}
        \|\a^{(s+3)}\|_{L^2(\cuvs)}+\|\b^{(s+3)}\|_{L^2(\ucuvs)}\les\ep_0.
    \end{align*}
\end{prop}
\begin{proof}
We have from Proposition \ref{bianchiequations} that
\begin{align*}
\nabs_3\a+\frac{1}{2}\trchb\,\a&=-2\sld_2^*\b+\Ga\c\Rl,\\
\nabs_4\b+2\trch\,\b&=\sld_2\a+\Ga\c\Rl.
\end{align*}
Differentiating it by $\dkbb^i$ and applying Proposition \ref{commu}, we obtain
\begin{align*}
\nabs_3(\dkbb^i\a)+\frac{1}{2}\trchb(\dkbb^i\a)&=-\slD^*(\dkbb^i\b)+\b^{(i-1)}+(\Ga\c\Rl)^{(i)},\\
\nabs_4(\dkbb^i\b)+2\trch(\dkbb^i\b)&=\slD(\dkbb^i\a)+(\Ga\c\Rl)^{(i)},
\end{align*}
where $\slD$ denotes an elliptic Hodge operator. Applying \eqref{casepp}, we deduce by induction
\begin{align*}
\int_\cuvs|\dkbb^i\a|^2+\int_\ucuvs|\dkbb^i\b|^2&\les \int_\cuvss|\dkbb^i\a|^2+\int_\ucuvss|\dkbb^i\b|^2+\int_V \left|\Rl^{(i)}\right|\left|(\Ga\c\Rl)^{(i)}\right|.
\end{align*}
Notice that we have for $i\leq s+3$
\begin{align}
\begin{split}\label{Rlnonlinear}
\int_V \left|\Rl^{(i)}\right|\left|(\Ga\c\Rl)^{(i)}\right|&\les\int_V|\Ga|\left|\Rl^{(i)}\right|\left|\Rl^{(i)}\right|+|\Rl|\left|\Rl^{(i)}\right|\left|\Ga^{(i)}\right|\\
&\les\ep\int_{u_0}^udu'\int_{H_{u'}^{(\de,\ub)}}\left|\Rl^{(i)}\right|^2\\
&+\ep\left(\int_{u_0}^udu'\int_{H_{u'}^{(\de,\ub)}}\left|\Rl^{(i)}\right|^2\right)^\frac{1}{2}\left(\int_V\left|\Ga^{(i)}\right|^2d\ub'du'\right)^\frac{1}{2}\\
&\les\ep^3+\ep^2\left(\int_{u_0}^u\int_{\de}^\ub\int_{S_{u',\ub'}}\left|\Ga^{(i)}\right|^2\right)^\frac{1}{2}\\
&\les\ep_0^2.
\end{split}
\end{align}
Combining with the initial assumption, we infer for $i\leq s+3$
\begin{align*}
    \int_\cuvs|\dkbb^i\a|^2+\int_\cuvs|\dkbb^i\b|^2\les\ep_0^2.
\end{align*}
This concludes the proof of Proposition \ref{estabab}.
\end{proof}
\subsubsection{Estimates for the Bianchi pair \texorpdfstring{$(\b,(\rho,-\si))$}{}}
\begin{prop}\label{estbrbr}
We have the following estimate:
\begin{align*}
    \|\b^{(s+3)}\|_{L^2(\cuvs)}+\|(\rho,\si)^{(s+3)}\|_{L^2(\ucuvs)}\les\ep_0.
\end{align*}
\end{prop}
\begin{proof}
We have from Proposition \ref{bianchiequations} that
\begin{align*}
\nabs_3\b+\trchb\,\b&=-\sld_1^*(\rho,-\si)+\Ga\c\Rl,\\
\nabs_4(\rho,-\si)+\frac{3}{2}\trch(\rho,-\si)&=\sld_1\b+\Ga\c\Rl.
\end{align*}
Differentiating it by $\dkbb^i$ and applying Proposition \ref{commu}, we obtain
\begin{align*}
\nabs_3(\dkbb^i\b)+\trchb(\dkbb^i\b)&=-\slD^*(\dkbb^i(\rho,-\si))+(\Ga\c\Rl)^{(i)},\\
\nabs_4(\dkbb^i(\rho,-\si))+\frac{3}{2}\trch(\dkbb^i(\rho,-\si))&=\slD(\dkbb^i\b)+(\Ga\c\Rl)^{(i)}.
\end{align*}
Applying \eqref{casenp}, \eqref{Rlnonlinear} and Proposition \ref{estabab}, we deduce for $i\leq s+3$
\begin{align*}
    \int_\cuvs|\dkbb^i\b|^2+\int_\ucuvs|\dkbb^i(\rho,\si)|^2&\les \int_\cuvss|\dkbb^i\b|^2+\int_\ucuvss|\dkbb^i(\rho,\si)|^2+\int_V|\dkbb^i\b|^2\\
    &+\int_V \left|\Rl^{(i)}\right|\left|(\Ga\c\Rl)^{(i)}\right|\\
    &\les\ep_0^2+\int_{\de}^\ub d\ub'\int_{\Hb_{\ub'}^{(u_0,u)}} |\dkbb^i\b|^2\\
    &\les\ep_0^2.
\end{align*}
This concludes the proof of Proposition \ref{estbrbr}.
\end{proof}
\subsubsection{Estimates for the Bianchi pair \texorpdfstring{$((\rho,\si),\bb)$}{}}
\begin{prop}\label{estrbrb}
We have the following estimate:
    \begin{align*}
        \|(\rho,\si)^{(s+3)}\|_{L^2(\cuvs)}+\|\bb^{(s+3)}\|_{L^2(\ucuvs)}\les\ep_0.
    \end{align*}
\end{prop}
\begin{proof}
We have from Proposition \ref{bianchiequations} that
\begin{align*}
\nabs_3(\rho,\si)+\frac{3}{2}\trchb(\rho,\si)&=-\sld_1\bb+\Ga\c\Rr,\\
\nabs_4\bb+\trch\,\bb&=\sld_1^*(\rho,\si)+\Ga\c\Rr.
\end{align*}
Differentiating it by $\dkbb^i$ and applying Proposition \ref{commu}, we obtain
\begin{align*}
\nabs_3(\dkbb^i(\rho,\si))+\frac{3}{2}\trchb(\dkbb^i(\rho,\si))&=-\sld_1(\dkbb^i\bb)+(\Ga\c\Rr)^{(i)},\\
\nabs_4(\dkbb^i\bb)+\trch(\dkbb^i\bb)&=\sld_1^*(\dkbb^i(\rho,\si))+(\Ga\c\Rr)^{(i)}.
\end{align*}
Proceeding as in \eqref{Rlnonlinear}, we have for $i\leq s+3$
\begin{align}
\begin{split}\label{Rrnonlinear}
\int_V \left|\Rr^{(i)}\right|\left|(\Ga\c\Rr)^{(i)}\right|&\les\int_V|\Ga|\left|\Rr^{(i)}\right|\left|\Rr^{(i)}\right|+|\Rr|\left|\Rr^{(i)}\right|\left|\Ga^{(i)}\right|\\
&\les\ep\int_{\de}^\ub d\ub'\int_{\Hb_{\ub'}^{(u_0,u)}}\left|\Rr^{(i)}\right|^2\\
&+\ep\left(\int_\de^\ub d\ub'\int_{\Hb_{\ub'}^{(u_0,u)}}\left|\Rr^{(i)}\right|^2\right)^\frac{1}{2}\left(\int_V\left|\Ga^{(i)}\right|^2d\ub'du'\right)^\frac{1}{2}\\
&\les\ep^3+\ep^2\left(\int_{u_0}^u\int_{\de}^\ub\int_{S_{u',\ub'}}\left|\Ga^{(i)}\right|^2\right)^\frac{1}{2}\\
&\les\ep_0^2.
\end{split}
\end{align}
Applying \eqref{casenp}, \eqref{Rrnonlinear} and Proposition \ref{estbrbr}, we deduce for $i\leq s+3$
\begin{align*}
    \int_\cuvs|\dkbb^i(\rho,\si)|^2+\int_\ucuvs|\dkbb^i\bb|^2&\les \int_\cuvss|\dkbb^i(\rho,\si)|^2+\int_\ucuvss|\dkbb^i\bb|^2+\int_V|\dkbb^i(\rho,\si)|^2\\
    &+\int_V \left|\Rl^{(i)}\right|\left|(\Ga\c\Rl)^{(i)}\right|\\
    &\les\ep_0^2+\int_{\de}^\ub d\ub'\int_{\Hb_{\ub'}^{(u_0,u)}}|\dkbb^i(\rho,\si)|^2\\
    &\les\ep_0^2.
\end{align*}
This concludes the proof of Proposition \ref{estrbrb}.
\end{proof}
\subsubsection{Estimates for the Bianchi pair \texorpdfstring{$(\bb,\aa)$}{}}
\begin{prop}\label{estbaba}
We have the following estimate:
    \begin{align*}
        \|\bb^{(s+3)}\|_{L^2(\cuvs)}+\|\aa^{(s+3)}\|_{L^2(\ucuvs)}\les\ep_0.
    \end{align*}
\end{prop}
\begin{proof}
We have from Proposition \ref{bianchiequations} that
\begin{align*}
\nabs_3\bb+2\trchb\,\bb&=-\sld_2\aa+\Ga\c\Rr,\\
\nabs_4\aa+\frac{1}{2}\trch\aa&=2\sld_2^*\bb+\Ga\c\Rr.
\end{align*}
Differentiating it by $\dkbb^i$ and applying Proposition \ref{commu}, we obtain
\begin{align*}
\nabs_3(\dkbb^i\bb)+2\trchb(\dkbb^i\bb)&=-\slD(\dkbb^i\aa)+(\Ga\c\Rr)^{(i)},\\
\nabs_4(\dkbb^i\aa)+\frac{1}{2}\trch(\dkbb^i\aa)&=\slD^*(\dkbb^i\bb)+\bb^{(i-1)}+(\Ga\c\Rr)^{(i)}.
\end{align*}
Applying \eqref{casenp}, \eqref{Rrnonlinear} and Proposition \ref{estrbrb}, we deduce for $i\leq s+3$
\begin{align*}
    \int_\cuvs|\dkbb^i\bb|^2+\int_\ucuvs|\dkbb^i\aa|^2&\les \int_\cuvss|\dkbb^i\bb|^2+\int_\ucuvss|\dkbb^i\aa|^2+\int_V|\dkbb^i\bb|^2\\
    &+\int_V \left|\Rr^{(i)}\right|\left|(\Ga\c\Rr)^{(i)}\right|\\
    &\les\ep_0^2+\int_{\de}^\ub d\ub'\int_{\Hb_{\ub'}^{(u_0,u)}}|\dkbb^i\bb|^2\\
    &\les\ep_0^2.
\end{align*}
This concludes the proof of Proposition \ref{estbaba}.
\end{proof}
Combining Propositions \ref{estabab}--\ref{estbaba}, this concludes the proof of Theorem \ref{MM1}.
\subsection{\texorpdfstring{$L^2(S_{u,\ub})$}{}--estimates for Ricci coefficients and curvature}\label{secOO}
In this section, we prove Theorem \ref{MM2} by integrating the null structure equations along the outgoing and incoming null cones.
\subsubsection{Estimates for curvature components}
\begin{prop}\label{estL2R}
We have the following estimates:
\begin{align*}
    \|(\a,\b,\rho,\si,\bb,\aa)^{(s+2)}\|_{L^2(S_{u,\ub})}\les\ep_0.
\end{align*}
\end{prop}
\begin{proof}
We have from Proposition \ref{bianchiequations}
\begin{align*}
    \nabs_3\a+\frac{1}{2}\trchb\,\a=\b^{(1)}+\Ga\c\Rl.
\end{align*}
Differentiating it by $\dkbb^{i-1}$ and applying Proposition \ref{commu}, we deduce
\begin{align*}
    \nabs_3(\dkbb^{i-1}\a)+\frac{1}{2}\trchb(\dkbb^{i-1}\a)=\b^{(i)}+(\Ga\c\Rl)^{(i-1)}.
\end{align*}
Applying Lemma \ref{evolutionlemma} and Proposition \ref{estabab}, we infer for $i\leq s+3$
\begin{align*}
\|\dkbb^{i-1}\a\|_{L^2(S_{u,\ub})}&\les\|\dkbb^{i-1}\a\|_{L^2(S_{u_0,\ub})}+\int_{u_0}^u\|\b^{(i)}\|_{L^2(S_{u',\ub})}+\|(\Ga\c\Rl)^{(i-1)}\|_{L^2(S_{u',\ub})}du'\\
&\les\left(\int_{u_0}^u\|\b^{(i)}\|_{L^2(S_{u',\ub})}^2du'\right)^\frac{1}{2}+\ep^2\\
&\les\ep_0.
\end{align*}
The following estimates can be deduced similarly via $\nabs_3$-transport:
\begin{align*}
     \|(\b,\rho,\si,\bb)^{(s+2)}\|_{L^2(S_{u,\ub})}\les\ep_0.
\end{align*}
Finally, we have from Proposition \ref{bianchiequations}
\begin{align*}
    \nabs_4\aa+\frac{1}{2}\trch\,\aa=\bb^{(1)}+\Ga\c\Rr.
\end{align*}
Differentiating it by $\dkbb^{i-1}$ and applying Proposition \ref{commu}, we deduce
\begin{align*}
    \nabs_4(\dkbb^{i-1}\aa)+\frac{1}{2}\trch(\dkbb^{i-1}\aa)=\bb^{(i)}+(\Ga\c\Rr)^{(i-1)}.
\end{align*}
Applying Lemma \ref{evolutionlemma} and Proposition \ref{estbaba}, we infer for $i\leq s+3$
\begin{align*}
    \|\dkbb^{i-1}\aa\|_{L^2(S_{u,\ub})}&\les\|\dkbb^{i-1}\aa\|_{L^2(S_{u,\de})}+\int_{\de}^\ub\|\b^{(i)}\|_{L^2(S_{u,\ub'})}+\|(\Ga\c\Rl)^{(i-1)}\|_{L^2(S_{u,\ub'})}d\ub'\\
    &\les\ep_0+\left(\int_\de^\ub\|\b^{(i)}\|_{L^2(S_{u,\ub'})}^2d\ub'\right)^\frac{1}{2}+\ep^2\\
    &\les\ep_0.
\end{align*}
This concludes the proof of Proposition \ref{estL2R}.
\end{proof}
\subsubsection{Estimates for \texorpdfstring{$\om$}{}, \texorpdfstring{$\omb$}{} and \texorpdfstring{$\log\Om$}{}}
\begin{prop}\label{estomombOm}
    We have the following estimates:
    \begin{align*}
        \left\|(\om,\omb,\log\Om)^{(s+3)}\right\|_{L^2(S_{u,\ub})}\les\ep_0.
    \end{align*}
\end{prop}
\begin{proof}
We have from Proposition \ref{nulles}
    \begin{align*}
        \nabs_3\om=\frac{1}{2}\rho+\Ga\c\Ga,\qquad\quad \nabs_4\omb=\frac{1}{2}\rho+\Ga\c\Ga.
    \end{align*}
Differentiating them by $\dkbb^i$ and applying Proposition \ref{commu}, we infer
    \begin{align*}
        \nabs_3(\dkbb^i\om)=\rho^{(i)}+\Ga\c\Ga^{(i)},\qquad\quad \nabs_4(\dkbb^i\omb)=\rho^{(i)}+\Ga\c\Ga^{(i)}.
    \end{align*}
Applying Lemma \ref{evolutionlemma}, we deduce for $i\leq s+3$
\begin{align*}
\|\dkbb^i\om\|_{L^2(S_{u,\ub})}&\les\|\dkbb^i\om\|_{L^2(S_{u_0,\ub})}+\int_{u_0}^u\|\rho^{(i)}\|_{L^2(S_{u',\ub})}+\|\Ga\|_{L^\infty(S_{u',\ub})}\|\Ga^{(i)}\|_{L^2(S_{u',\ub})}du'\\
&\les\left(\int_{u_0}^u\|\rho^{(i)}\|^2_{L^2(S_{u',\ub})}du'\right)^\frac{1}{2}+\ep^2\\
&\les\ep_0,
\end{align*}
where we used Proposition \ref{estbrbr} at the last step. Similarly, we have for $i\leq s+3$
\begin{align*}
\|\dkbb^i\omb\|_{L^2(S_{u,\ub})}&\les\|\dkbb^i\omb\|_{L^2(S_{u,\de})}+\int_{u_0}^u\|\rho^{(i)}\|_{L^2(S_{u,\ub'})}+\|\Ga\|_{L^\infty(S_{u,\ub'})}\|\Ga^{(i)}\|_{L^2(S_{u,\ub'})}d\ub'\\
&\les\left(\int_{\de}^\ub\|\rho^{(i)}\|^2_{L^2(S_{u,\ub'})}d\ub'\right)^\frac{1}{2}+\ep^2\\
&\les\ep_0,
\end{align*}
where we used Proposition \ref{estrbrb} at the last step. Thus, we obtain
\begin{align}\label{omombest}
    \|(\om,\omb)^{(s+3)}\|_{L^2(S_{u,\ub})}\les\ep_0.
\end{align}
Finally, we have from \eqref{nullidentities}
\begin{align*}
    \nabs_4\log\Om=-2\om.
\end{align*}
Differentiating them by $\dkbb^i$ and applying Proposition \ref{commu}, we infer
\begin{align*}
    \nabs_4(\dkbb^i\log\Om)=\om^{(i)}+\Ga\c\Ga^{(i)}.
\end{align*}
Applying Lemma \ref{evolutionlemma} and \eqref{omombest}, we infer for $i\leq s+3$
\begin{align*}
\|(\log\Om)^{(i)}\|_{L^2(S_{u,\ub})}&\les\|(\log\Om)^{(i)}\|_{L^2(S_{u,\de})}+\int_\de^\ub\|\om^{(i)}\|_{L^2(S_{u,\ub'})}+\|\Ga\|_{L^\infty(S_{u,\ub'})}\|\Ga^{(i)}\|_{L^2(S_{u,\ub'})}d\ub'\\
&\les\ep_0.
\end{align*}
This concludes the proof of Proposition \ref{estomombOm}.
\end{proof}
\subsubsection{Estimates for \texorpdfstring{$\trcht$}{} and \texorpdfstring{$\trchbt$}{}}
\begin{prop}\label{esttrchttrchbt}
    We have the following estimates:
    \begin{align*}
        \left\|(\trcht,\trchbt)^{(s+3)}\right\|_{L^2(S_{u,\ub})}\les\ep_0.
    \end{align*}
\end{prop}
\begin{proof}
We have from Proposition \ref{nulles}
    \begin{align*}
        \nabs_4\trch=-\frac{1}{2}(\trch)^2-2\om\trch+\Ga\c\Ga,\\
        \nabs_3\trchb=-\frac{1}{2}(\trchb)^2-2\omb\trchb+\Ga\c\Ga.
    \end{align*}
    Recalling that $r=\ub-u$, we have
    \begin{align*}
        e_4\left(\frac{2}{r}\right)&=-\frac{2}{r^2}e_4(r)=-\frac{2}{\Om r^2},\qquad\quad e_3\left(\frac{2}{r}\right)=-\frac{2}{r^2}e_3(r)=\frac{2}{\Om r^2}.
    \end{align*}
    Hence, we obtain
    \begin{align*}
        \nabs_4\trcht+\trch\,\trcht&=(\om,\log\Om)^{(0)}+\Ga\c\Ga,\\
        \nabs_3\trchbt+\trchb\,\trchbt&=(\omb,\log\Om)^{(0)}+\Ga\c\Ga.
    \end{align*}
    Differentiating them by $\dkbb^i$ and applying Proposition \ref{commu}, we infer
    \begin{align*}
    \nabs_4(\dkbb^i\trcht)+\trch(\dkbb^i\trcht)&=(\om,\log\Om)^{(i)}+\Ga\c\Ga^{(i)},\\
    \nabs_3(\dkbb^i\trchbt)+\trchb(\dkbb^i\trchbt)&=(\omb,\log\Om)^{(i)}+\Ga\c\Ga^{(i)}.
    \end{align*}
    Applying Lemma \ref{evolutionlemma} and Proposition \ref{estomombOm}, we infer for $i\leq s+3$
\begin{align*}
\|\dkbb^i\trcht\|_{L^2(S_{u,\ub})}&\les\|\dkbb^i\trcht\|_{L^2(S_{u,\de})}+\int_{\de}^u\|(\om,\log\Om)^{(i)}\|_{L^2(S_{u,\ub'})}d\ub'\\
&+\int_{\de}^u\|\Ga\|_{L^\infty(S_{u,\ub'})}\|\Ga^{(i)}\|_{L^2(S_{u,\ub'})}d\ub'\\
&\les\ep_0+\ep^2\les\ep_0.
\end{align*}
Similarly, we also have for $i\leq s+3$
\begin{align*}
    \|\dkbb^i\trchbt\|_{L^2(S_{u,\ub})}\les\ep_0.
\end{align*}
This concludes the proof of Proposition \ref{esttrchttrchbt}.
\end{proof}
\subsubsection{Estimates for \texorpdfstring{$\eta$}{}, \texorpdfstring{$\etab$}{} and \texorpdfstring{$\ze$}{}}
\begin{prop}\label{estetaetab}
We have the following estimates:
\begin{align*}
    \left\|(\eta,\etab,\ze)^{(s+3)}\right\|_{L^2(S_{u,\ub})}\les\ep_0.
\end{align*}
\end{prop}
\begin{proof}
We recall the following null structure equations from Proposition \ref{nulles}:
\begin{align*}
    \nabs_4\eta&=-\chi\c(\eta-\etab)-\b,\\
    \nabs_4\etab&=-\chib\c(\etab-\eta)+\bb.
\end{align*}
We define the mass aspect functions
\begin{align}
    \begin{split}
    \mu&:=-\sdivs\eta-\rho+\frac{1}{2}\hch\c\hchb,\\
    \mub&:= -\sdivs\etab-\rho+\frac{1}{2}\hch\c\hchb.\label{defmu}
    \end{split}
\end{align}
By a direct computation, we obtain\footnote{See for example Lemma 4.3.1 in \cite{kn}.}
\begin{align}
\begin{split}\label{computationlongmu}
\nabs_4\mu+\trch\,\mu&=\rho^{(0)}+\frac{1}{2}\tr\chi\,\mub+\Ga\c\Ga^{(1)}, \\
\nabs_3\mub+\trchb\,\mub&=\rho^{(0)}+\frac{1}{2}\trchb\,\mu+\Ga\c\Ga^{(1)}.
\end{split}
\end{align}
We also introduce the following modifications of $\mu$ and $\mub$:
\begin{equation}\label{mumc}
    [\mu]:=\mu+\frac{1}{4}\trch\trchb,\qquad [\mub]:=\mub+\frac{1}{4}\trch\trchb.
\end{equation}
Applying \eqref{defmu} and Proposition \ref{nulles}, we have\footnote{See for example (5.37) in \cite{Shen22} for the detailed computation.}
\begin{align*}
    \nabs_4[\mu]+\trch\,[\mu]&=\rho^{(0)}+\Ga\c\Ga^{(1)},\\
    \nabs_3[\mub]+\trchb\,[\mub]&=\rho^{(0)}+\Ga\c\Ga^{(1)}.
\end{align*}
Differentiating them by $\dkbb^i$ and applying Proposition \ref{commu}, we infer
\begin{align*}
    \nabs_4(\dkbb^{i-1}[\mu])+\trch(\dkbb^{i-1}[\mu])&=\rho^{(i-1)}+\Ga\c\Ga^{(i)},\\
    \nabs_3(\dkbb^{i-1}[\mub])+\trchb(\dkbb^{i-1}[\mub])&=\rho^{(i-1)}+\Ga\c\Ga^{(i)}.
\end{align*}
Applying Lemma \ref{evolutionlemma} and Proposition \ref{estrbrb}, we obtain for $2\leq i\leq s+3$
\begin{align}
\begin{split}\label{dkbmuest}
    \|\dkbb^{i-1}[\mu]\|_{L^2(S_{u,\ub})}&\les\|\dkbb^{i-1} [\mu]\|_{L^2(S_{u,\de})}+\int_\de^\ub \|\rho^{(i-1)}\|_{L^2(S_{u,\ub'})}d\ub'\\
    &+\int_\de^\ub\|\Ga\|_{L^\infty(S_{u,\ub'})}\|\Ga^{(i)}\|_{L^2(S_{u,\ub'})}d\ub'\\
    &\les\ep_0+\ep^2+\left(\int_\de^\ub\|\rho^{(i-1)}\|_{L^2(S_{u,\ub'})}^2 d\ub'\right)^\frac{1}{2}\\
    &\les\ep_0.
\end{split}
\end{align}
Recall that we have from \eqref{torsion}, \eqref{defmu} and \eqref{mumc}
\begin{align*}
    \sdivs\eta=-[\mu]+\frac{1}{4}\trch\trchb-\rho+\frac{1}{2}\hch\c\hchb,\qquad\curls\eta=\si-\frac{1}{2}\hch\wedge\hchb.
\end{align*}
Applying Proposition \ref{ellipticLp}, we infer
\begin{align*}
    \|\eta^{(s+3)}\|_{L^2(S_{u,\ub})}&\les \|(\trcht,\trchbt,\rho,\si)^{(s+2)}\|_{L^2(S_{u,\ub})}+\|(\dkbb[\mu])^{(s+1)}\|_{L^2(S_{u,\ub})}\\
    &+\|\Ga\|_{L^\infty(S_{u,\ub})}\|\Ga^{(s+2)}\|_{L^2(S_{u,\ub})}\\
    &\les\ep_0+\ep^2\les\ep_0,
\end{align*}
where we used \eqref{dkbmuest} and Propositions \ref{estL2R} and \ref{esttrchttrchbt} at the last step. The estimate for $\etab$ is similar and is left to the reader. Combining with \eqref{nullidentities}, this concludes the proof of Proposition \ref{estetaetab}.
\end{proof}
\subsubsection{Estimates for \texorpdfstring{$\hch$}{} and \texorpdfstring{$\hchb$}{}}
\begin{prop}\label{hchhchbest}
We have the following estimates:
    \begin{align*}
        \left\|(\hch,\hchb)^{(s+3)}\right\|_{L^2(S_{u,\ub})}\les\ep_0.
    \end{align*}
\end{prop}
\begin{proof}
    We have from \eqref{codazzi}
    \begin{align*}
        \sdivs\hch&=\trcht^{(1)}+\b^{(0)}+\ze^{(0)}+\Ga\c\Ga,\\
        \sdivs\hchb&=\trchbt^{(1)}+\bb^{(0)}+\ze^{(0)}+\Ga\c\Ga.
    \end{align*}
    Applying Propositions \ref{ellipticLp}, \ref{estL2R}, \ref{esttrchttrchbt} and \ref{estetaetab}, we obtain
    \begin{align*}
        \|(\hch,\hchb)^{(s+3)}\|_{L^2(S_{u,\ub})}&\les\|(\trcht,\trchbt)^{(s+3)}\|_{L^2(S_{u,\ub})}+\|(\b,\bb,\ze)^{(s+2)}\|_{L^2(S_{u,\ub})}\\
        &+\|\Ga\|_{L^\infty(S_{u,\ub})}\|\Ga^{(s+2)}\|_{L^2(S_{u,\ub})}\\
        &\les\ep_0+\ep^2\les\ep_0.
    \end{align*}
    This concludes the proof of Proposition \ref{hchhchbest}.
\end{proof}
Combining Propositions \ref{estL2R}--\ref{hchhchbest}, this concludes the proof of Theorem \ref{MM2}.
\subsection{Proof of Theorem \ref{mainstability}}
We now use Theorems \ref{MM0}--\ref{MM2} to prove Theorem \ref{mainstability}.
\begin{df}\label{bootstrapstability}
For any $u_*\in[u_0,-\frac{1}{2}]$, let $\aleph(u_*)$ be the set of spacetimes $V(u_*,1)$ associated with a double null foliation $(u,\ub)$ in which we have the following bounds:
\begin{align}
    \Ok\leq\ep,\qquad\quad\Rk\leq\ep.\label{BB2}
\end{align}
\end{df}
\begin{df}\label{defbootstability}
We denote by $\mathcal{U}$ the set of values $u_*$ such that $\aleph(u_*)\ne\emptyset$.
\end{df}
The assumptions of Theorem \ref{mainstability} and Theorem \ref{MM0} imply that
\begin{equation*}
    \Ok_{(0)}+\Rk_{(0)}\les\ep_0.
\end{equation*}
Combining with the local existence theorem, we deduce that \eqref{B2} holds if $u_*$ is sufficiently close to $u_0$. So, we have $\mathcal{U}\ne\emptyset$.\\ \\
Define $u_*$ as the supremum of the set $\mathcal{U}$. We assume by contradiction that $u_*<-\frac{1}{2}$. In particular, we may assume $u_*\in\mathcal{U}$. We consider the region $V(u_*,1)$. Applying Theorem \ref{MM1}, we obtain
\begin{equation*}
    \Rk\les\ep_0.
\end{equation*}
Then, we apply Theorem \ref{MM2} to deduce
\begin{equation*}
    \Ok\les\ep_0.
\end{equation*}
Applying local existence results,\footnote{See Theorem 1 in \cite{luk}.} we can extend $V(u_*,1)$ to $V(u_*+\nu,1)$ for a $\nu$ sufficiently small. We denote $\widetilde{\Ok}$ and $\widetilde{\Rk}$ the norms in the extended region. We have
\begin{equation*}
    \widetilde{\Ok}\les\ep_0,\qquad\qquad \widetilde{\Rk}\les\ep_0,
\end{equation*}
as a consequence of continuity. We deduce that $V(u_*+\nu,1)$ satisfies all the properties in Definition \ref{bootstrapstability}, and so $\aleph(u_*+\nu)\ne\emptyset$, which is a contradiction. Thus, we have $u_*=-\frac{1}{2}$. Moreover, we have
\begin{equation*}
    \Ok\les\ep_0,\qquad\quad \Rk\les\ep_0\qquad\mbox{ in }\;\; V\left(-\frac{1}{2},1\right).
\end{equation*}
This concludes the proof of Theorem \ref{mainstability}.
\subsection{Short-pulse annulus and barrier annulus}
We first prove the following consequence of Theorem \ref{mainstability}.
\begin{prop}\label{estCs}
    Under the assumptions of Theorem \ref{mainstability}, we have:
    \begin{align*}
    \|\g-\etabf\|_{C^s(\V,\etabf)}\les\ep_0,
\end{align*}
where $\etabf$ denotes the Minkowski metric and $\V:=V(-\frac{1}{2},1)$ denotes the transition region.
\end{prop}
\begin{proof}
We have from $\Ok\les\ep_0$ in Theorem \ref{mainstability} that
\begin{align*}
\left\|\nabs^{\leq s+2}(\a,\b,\rho,\si,\bb,\aa)\right\|_{L^2(S_{u,\ub})}&\les\ep_0,\\
\left\|\nabs^{\leq s+3}(\trcht,\trchbt,\hch,\hchb,\om,\omb,\eta,\etab,\ze,\log\Om,\bbb,\slgc)\right\|_{L^2(S_{u,\ub})}&\les\ep_0.
\end{align*}
We now suppose the following bootstrap assumption:\footnote{Recall that $\dk$ is defined in \eqref{dfdk}.}
\begin{align}
\begin{split}\label{bootCs}
\left\|\dk^{\leq s+1}(\a,\b,\rho,\si,\bb,\aa)\right\|_{L^2(S_{u,\ub})}&\leq\ep,\\
\left\|\dk^{\leq s+2}(\trcht,\trchbt,\hch,\hchb,\om,\omb,\eta,\etab,\ze,\log\Om,\bbb,\slgc)\right\|_{L^2(S_{u,\ub})}&\leq\ep,
\end{split}
\end{align}
which will be improved in the end of the proof. As a consequence of Propositions \ref{nulles} and \ref{bianchiequations}, it is sufficient to estimate the following quantities:
\begin{align*}
    \nabs^{i}\nabs_4^{j}\a,\qquad \nabs^i\nabs_4^j(\om,\etab),\qquad\nabs^{i}\nabs_3^{j}\aa,\qquad \nabs^i\nabs_3^j(\omb,\eta,\bbb).
\end{align*}
We have from Proposition \ref{bianchiequations}
\begin{align*}
    \nabs_3(r\a)=\b^{(1)}+\Ga\c\Ga^{(1)}.
\end{align*}
Applying Propositions \ref{bianchiequations} and \ref{commu} and \eqref{bootCs}, we obtain for $i+j\leq s+1$
\begin{align*}
    \nabs_3(\nabs^i\nabs_4^j(r\a))=\nabs_4^{\leq j-1}(r\a)^{(i+2)}+O(\ep^2).
\end{align*}
Integrating it along $\Hb_\ub^{(u_0,u)}$, we deduce by induction that
\begin{align*}
    \|\nabs^i\nabs_4^j\a\|_{L^2(S_{u,\ub})}\les \ep_0,\qquad \forall\; i+j\leq s+1.
\end{align*}
Similarly, we have
\begin{align*}
    \|\nabs^i\nabs_4^j(\om,\etab)\|_{L^2(S_{u,\ub})}\les \ep_0,\qquad \forall\; i+j\leq s+2.
\end{align*}
Next, we have from Proposition \ref{bianchiequations}
\begin{align*}
    \nabs_4(r\aa)=\bb^{(1)}+\Ga\c\Ga^{(1)}.
\end{align*}
Applying Propositions \ref{bianchiequations} and \ref{commu} and \eqref{bootCs}, we obtain for $i+j\leq s+1$
\begin{align*}
    \nabs_4(\nabs^i\nabs_3^j(r\aa))=\nabs_3^{\leq j-1}(r\aa)^{(i+2)}+O(\ep^2).
\end{align*}
Integrating it along $H_u^{(\de,\ub)}$, we deduce by induction that\footnote{Notice that $|\nabs^i\nabs_3^j\aa|\les\ep_0$ holds on $\Hb_\de^{(u_0,-\frac{1}{2})}$ for all $i+j\leq s+1$, which can be proved by applying Theorem \ref{shortpulsecone} and transporting the Bianchi equations for $\nabs_4(\nabs^i\nabs_3^j\aa)$ along $H_u^{(0,\de)}$.}
\begin{align*}
    \|\nabs^i\nabs_3^j\aa\|_{L^2(S_{u,\ub})}\les \ep_0,\qquad \forall\; i+j\leq s+1.
\end{align*}
Similarly, we have
\begin{align*}
    \|\nabs^i\nabs_3^j(\omb,\eta)\|_{L^2(S_{u,\ub})}\les \ep_0,\qquad \forall\; i+j\leq s+2.
\end{align*}
Finally, we have
\begin{align*}
    \nabs_4(\bbb^A)=-4\Om\ze^A,
\end{align*}
which implies from Proposition \ref{commu} that for $i+j\leq s+2$
\begin{align*}
    \nabs_4(\nabs^i\nabs_3^j\bbb)=\nabs^{\leq j-1}\bbb^{(i)}+\nabs^i\nabs^j\ze+O(\ep^2)=O(\ep_0).
\end{align*}
Integrating it along $H_u^{(\de,\ub)}$, we deduce by induction
\begin{align*}
    \|\nabs^i\nabs_3^j\bbb\|_{L^2(S_{u,\ub})}\les \ep_0,\qquad \forall\; i+j\leq s+2.
\end{align*}
Combining the above estimates, we obtain
\begin{align*}
\left\|\dk^{\leq s+1}(\a,\b,\rho,\si,\bb,\aa)\right\|_{L^2(S_{u,\ub})}&\les\ep_0,\\
\left\|\dk^{\leq s+2}(\trcht,\trchbt,\hch,\hchb,\om,\omb,\eta,\etab,\ze,\log\Om,\bbb,\slgc)\right\|_{L^2(S_{u,\ub})}&\les\ep_0,
\end{align*}
which improves \eqref{bootCs}. Combining with Proposition \ref{sobolevinequality}, we infer
\begin{align*}
    \sup_{S_{u,\ub}\subseteq\V}\left\|\dk^{\leq s}(\trcht,\trchbt,\hch,\hchb,\om,\omb,\eta,\etab,\ze,\log\Om,\bbb,\slgc)\right\|_{L^\infty(S_{u,\ub})}\les\ep_0.
\end{align*}
This concludes the proof of Proposition \ref{estCs}.
\end{proof}
\begin{thm}\label{interiorsolution}
    For any $s\in\mathbb{N}$, there exists a sufficiently large $a_0(s)>0$. For any $a>a_0$ and $0<\de\leq a^{-2}$, there exists a spacelike initial data $(\Si,g,k)$ solving \eqref{constrainteq}, endowed with a radial $r$--foliation for $r\in(0,2)$, which satisfies the following properties: \begin{enumerate}
        \item We have\footnote{For any $R>0$, we denote $B_R:=\{p\in\Si/\, r(p)<R\}$.}
        \begin{align}
        \begin{split}\label{diffge}
            (g,k)&=(e,0)\qquad\mbox{ in }\;B_{1-2\de},\\
            \|(g-e,k)\|_{H^s\times H^{s-1}(B_2\setminus\ov{B_1})}&\les a^{-1}.
        \end{split}
        \end{align}
        \item The following estimates hold on $\Si$:
        \begin{equation}\label{notrappingde}
            \tr_\slg(-\th-k)<0,\qquad\quad\tr_\slg(\th-k)>0,
        \end{equation}
        where $\th$ is defined in Definition \ref{dftrapped}.
    \end{enumerate}
    The spacelike initial data $(\Si,g,k)$ is denoted by $(\Si_{\de,a},g_{\de,a},k_{\de,a})$, or simply $\Si(\de,a)$.
\end{thm}
\begin{rk}
    The annulus $A_1:=B_2\setminus\ov{B_1}$ is called the \emph{barrier annulus} and the $B_1\setminus\ov{B_{1-2\de}}$ is called the \emph{short-pulse annulus}.
\end{rk}
\begin{proof}[Proof of Theorem \ref{interiorsolution}]
    As an immediate consequence of Theorem \ref{shortpulsecone} and Proposition \ref{estCs}, the spacetime obtained in Theorem \ref{mainstability} satisfies the following properties:
    \begin{itemize}
        \item It coincides with the Minkowski spacetime for $\ub\leq 0$.
        \item In the region $\left\{(u,\ub)\in\left[u_0,-\frac{1}{2}\right]\times [0,\de]\right\}$, we have
        \begin{align}\label{trchctrchbcnotrapping}
            |\trchc,\trchbc|\les a^{-1}.
        \end{align}
        \item In the region $\left\{(u,\ub)\in\left[u_0,-\frac{1}{2}\right]\times [\de,1]\right\}$, we have
        \begin{align}\label{g-eta}
            \|\g-\etabf\|_{C^s}\les a^{-1}.
        \end{align}
    \end{itemize}
    We now take
    \begin{align*}
        u_0=-\frac{3}{2}+\de,
    \end{align*}
    and we consider the following time-constant slice:
    \begin{align*}
        \Si:=\left\{u+\ub=-1+2\de\right\},
    \end{align*}
    see Figure \ref{fig:shortpulse+stab} for a geometric illustration. 
    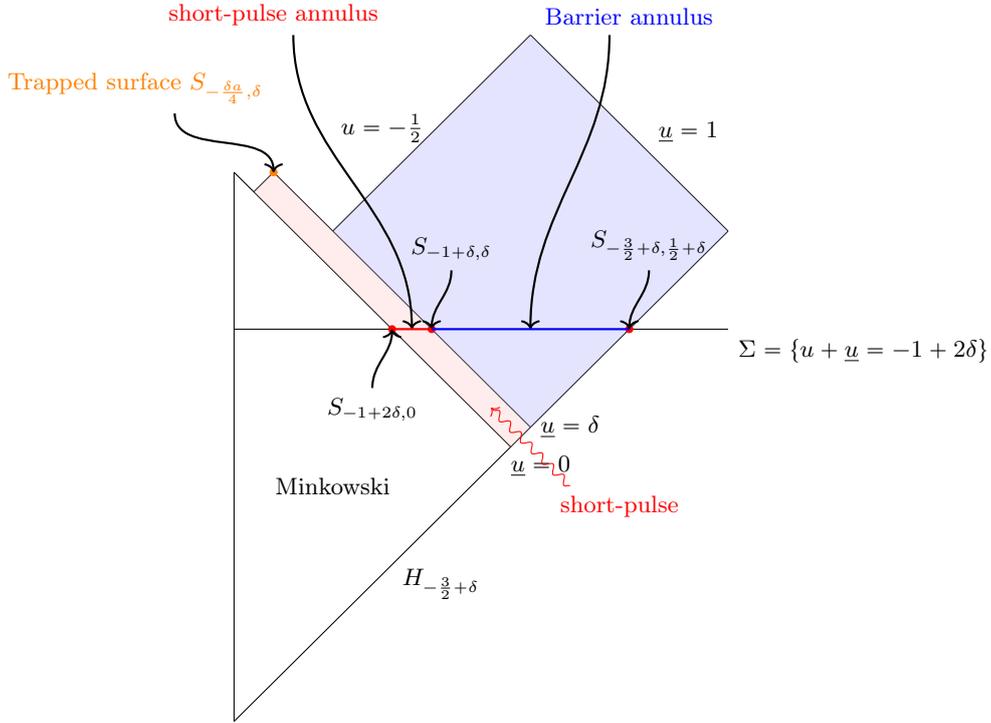
\begin{figure}[H]
        \centering
        \begin{tikzpicture}[scale=2.6, decorate]
  \coordinate (A) at (0,-2);
  \coordinate (B) at (0,0);
  \coordinate (C) at (0,0.8);
  \coordinate (D) at (0.1,0.7);
  \coordinate (E) at (0.2,0.8);
  \coordinate (F) at (1.4,-0.6);
  \coordinate (G) at (1.5,-0.5);
  \coordinate (H) at (2,0);
  \coordinate (I) at (2.5,0.5);
  \coordinate (J) at (1,0);
  \coordinate (K) at (0.5,0.5);
  \coordinate (L) at (1.5,1.5);
  \coordinate (M) at (2.5,0);
  \coordinate (N) at (0.8,0);
  \draw (A) -- (C);
  \draw (A) -- (I);
  \draw (F) -- (C);
  \draw (D) -- (E);
  \draw (G) -- (E);
  \draw (K) -- (L);
  \draw (I) -- (L);
  \draw (B) -- (M);
\fill[red!20,opacity=0.35](G) -- (E) -- (D) -- (F) -- cycle;
\fill[blue!30,opacity=0.35](G) -- (I) -- (L) -- (K) -- cycle;
  \node[right] at (1.35,-0.7) {\footnotesize $\ub = 0$};
  \node[right] at (G) {\footnotesize $\ub = \de$};
  \node[above] at (2.3,0.9) {\footnotesize $\ub = 1$};
  \node[above] at (0.75,0.9) {\footnotesize $u = -\frac{1}{2}$};
  \node[below right] at (M) {\footnotesize $\Si=\{u+\ub=-1+2\de\}$};
  \node at (0.5, -0.8) {\footnotesize Minkowski};
  \node[right] at (0.8, -1.3) {\footnotesize $H_{-\frac{3}{2}+\de}$};
  \node[below right] at (1.6,-0.8) {\footnotesize \red{short-pulse}};
\filldraw[red] (J) circle (0.5pt);
\filldraw[red] (H) circle (0.5pt);
\filldraw[red] (N) circle (0.5pt);
\filldraw[orange] (E) circle (0.5pt);
    \draw[->, thick, rounded corners=8pt] (1.1,0.3) to[out=-90, in=90] (J);
    \node[above] at (1.1,0.3) {\footnotesize $S_{-1+\de,\de}$};
    \draw[->, thick, rounded corners=8pt] (0.7,-0.3) to[out=90, in=-90] (N);
    \node[below] at (0.7,-0.3) {\footnotesize $S_{-1+2\de,0}$};
   \draw[->, thick, rounded corners=8pt] (2.1,0.3) to[out=-90, in=90] (H);
   \node[above] at (2.1,0.3) {\footnotesize $S_{-\frac{3}{2}+\de,\frac{1}{2}+\de}$};
    \draw[->, thick, rounded corners=8pt] (-0.3,1.1) to[out=-90, in=90] (E);
    \node[above] at (-0.5,1.1) {\footnotesize {\color{orange}Trapped surface $S_{-\frac{\de a}{4},\de}$}};
    \draw[->, thick, rounded corners=8pt] (0.3,1.5) to[out=-90, in=90] (0.9,0);
    \node[above] at (0.2,1.5) {\footnotesize \red{short-pulse annulus}};
    \draw[->, thick, rounded corners=8pt] (1.9,1.5) to[out=-90, in=90] (1.5,0);
    \node[above] at (2,1.5) {\footnotesize \blue{Barrier annulus}};
    \draw[->, decorate, decoration={snake, amplitude=0.5mm, segment length=2mm}, thin, red]  (1.7,-0.8) -- (1.3,-0.4);
    \draw[draw=blue, thick] (J) -- (H);
    \draw[draw=red, thick] (N) -- (J);
\end{tikzpicture}
        \caption{constant-time slice $(\Si,g,k)=(\Si_{\de,a},g_{\de,a},k_{\de,a})$.}
        \label{fig:shortpulse+stab}
    \end{figure}
    Taking $r:=\ub-u$, we obtain a radial foliation on $\Si$ for $r\in(0,2)$. Notice that $\Si$ coincides with a constant-time slice of Minkowski spacetime for $r\leq 1-2\de$. Moreover, \eqref{notrappingde} follows from \eqref{trchctrchbcnotrapping} and \eqref{chichib}. Finally, we see from \eqref{g-eta} that \eqref{diffge} holds. This concludes the proof of Theorem \ref{interiorsolution}.
\end{proof}
\section{Vacuum geometrostatic manifold}\label{secBL}
In this section, we introduce the \emph{vacuum geometrostatic manifold} and the \emph{Brill-Lindquist metric}, which plays an essential role in the construction of the desired Cauchy data.
\subsection{Brill-Lindquist metric}
\begin{prop}\label{conformalchange}
    Let $g_0$ be a Riemannian metric on a $n$--dimensional manifold. We define the conformal metric:
    \begin{align*}
        g:=U^\frac{4}{n-2}g_0,
    \end{align*}
    where $U>0$ is a conformal factor. Then, the scalar curvature transforms as follows:
    \begin{align*}
        R(g)=U^{-\frac{n+2}{n-2}}\left[R(g_0)U-\frac{4(n-1)}{n-2}\De_{g_0} U\right].
    \end{align*}
\end{prop}
\begin{proof}
    See Chapter V of \cite{SchoenYau94}.
\end{proof}
\begin{lem}\label{vanishingR}
    Let $g$ be a Riemannian metric on an open subset $\UU\subset\RRR^3$, which takes the following form:
    \begin{align*}
        g=U^4 e,
    \end{align*}
    where $e$ denotes the Euclidean metric. Then $R(g)=0$ on $\UU$ is equivalent to
    \begin{align*}
        \De_e U=0\qquad \mbox{ on }\;\UU.
    \end{align*}
\end{lem}
\begin{proof}
Taking $g_0=e$ and $n=3$ in Proposition \ref{conformalchange}, we obtain
\begin{align*}
    R(g)=U^{-5}\left(R(e)U-8\De_e U\right)=-8U^{-5}\De_e U.
\end{align*}
Thus, $R(g)=0$ is equivalent to $\De_e U=0$. This concludes the proof of Lemma \ref{vanishingR}.
\end{proof}
\begin{prop}\label{BrillLindquist}
    We introduce the following Brill-Lindquist metric:
    \begin{align}\label{gBLdf}
        \gBL:=U^4(\x)\,e:=\left(1+\sum_{I=1}^N\frac{m_I}{2|\x-\cb_I|}\right)^4e.
    \end{align}
    Then, for any fixed parameters $N\in\mathbb{N}$, $m_I>0$ and $\cb_I\in\RRR^3$, $(\gBL,0)$ is a solution of the constraint equations \eqref{constrainteq} on $\RRR^3\setminus\{\cb_I\}_{I=1}^N$. Moreover, the Riemannian manifold $(\RRR^3\setminus \{\cb_I\}_{I=1}^N,\gBL)$ is called a \emph{vacuum geometrostatic manifold}.
\end{prop}
\begin{proof}
    For any $I\in\{1,2,\dots,N\}$, $|\x-\cb_I|^{-1}$ is the fundamental solution of the Laplace equation $\De_e U=0$ on $\RRR^3\setminus\{\cb_I\}$. Combining with Lemma \ref{vanishingR}, we obtain
    \begin{equation*}
        R(\gBL)=0\qquad\mbox{ on }\; \RRR^3\setminus\{\cb_I\}_{I=1}^N.
    \end{equation*}
    This concludes the proof of Proposition \ref{BrillLindquist}.
\end{proof}
\begin{rk}\label{rkSchw}
    In the particular case $N=1$, the Brill-Lindquist metric \eqref{gBLdf} becomes the Riemannian Schwarzschild metric in isothermal coordinate
    \begin{equation}\label{Schw-isotermal}
        g=\left(1+\frac{m}{2r}\right)^4e,\qquad\quad r:=|\x|,
    \end{equation}
    which corresponds to the slice $t=0$ of the Schwarzschild metric. More precisely, we have
    \begin{equation}\label{conformalchangeschw}
        g=\left(1-\frac{2m}{\varrho}\right)^{-1}d\varrho^2+\varrho^2d\si_{\mathbb{S}^2}^2,\qquad\quad \varrho:=r\left(1+\frac{m}{2r}\right)^2.
    \end{equation}
\end{rk}
\subsection{Mean curvature of pole-centered spheres}
We now compute the mean curvature of the coordinate spheres for the Brill-Lindquist metric defined in \eqref{gBLdf}.
\begin{prop}\label{meancurvaturebl}
Let $\gBL:=U^4e$ be the Brill-Lindquist metric defined in \eqref{gBLdf}. Assume that the balls\footnote{Recall that $e$ denotes the Euclidean metric, $m_I > 0$ represents masses, and $\cb_I \in \mathbb{R}^3$ denotes the locations of black holes. The condition \eqref{disjointballs} ensures that each region $\{|\x -\cb_I| < \frac{m_I}{2}\}$ contains no other singularities except $\x=\cb_I$ and lies entirely within its own local Schwarzschild asymptotic regime. In particular, the conformal factor $U$ remains smooth on the complement of these balls.}
\begin{align}\label{disjointballs} 
\left\{|\x-\cb_I|<\frac{m_I}{2} \right\} \cap \left\{|\x-\cb_J| < \frac{m_J}{2} \right\} = \emptyset, \qquad \forall\; I \ne J,
\end{align} 
are mutually disjoint. Then, for any $R > 0$ such that $S_R:=\{\x/\, |\x|=R\}$ avoids $\{\cb_I\}_{I=1}^N$, the mean curvature $H_g$ of $S_R$ with respect to $\gBL$ is given explicitly by \begin{align}\label{Hgexplicit} 
H_g(\x) =\frac{2}{U(\x)^3R} \left[1 + \sum_{I=1}^N \frac{m_I}{2|\x-\cb_I|^3}\left(|\cb_I|^2 - R^2\right)\right]. 
\end{align} 
\end{prop}
\begin{proof} Let $S$ be a closed $2$--surface in $(\mathbb{R}^3\setminus\{\cb_I\}_{I=1}^N, \gBL) $. Then the area of $ S $ with respect to $ g = \gBL $ is given by
\begin{align*}
\mathrm{Area}_g(S) = \int_{S} d\sigma_g = \int_{S} U^4(x)\, d\sigma_e.
\end{align*}
Let $ S_t $ be a variation of $ S $ with a variational vector field $V =f\nu_e=U^2f\nu_g $, where $\nu_e$ is the outward Euclidean unit normal and $\nu_g=U^{-2}\nu_e $ is the unit normal with respect to the metric $g$. Then the first variation of area becomes
\begin{align*}
\frac{d}{dt}\mathrm{Area}_g(S_t)\bigg|_{t=0}= \left.\frac{d}{dt} \int_{S_t} U^4\, d\sigma_e \right|_{t=0}= \int_{S} \left(4U^3\, \nab U\c\nu_e +U^4 H_e\right) f\, d\sigma_e.
\end{align*}
Using the identity $ d\si_g =U^4 d\si_e $, we rewrite the expression as
\begin{align*}
\frac{d}{dt}\mathrm{Area}_g(S_t)\bigg|_{t=0} 
= \int_{S} \left(H_e + 4\frac{\nu_e(U)}{U}\right) f\, d\sigma_g.
\end{align*}
Recalling that $V=f\nu_e=U^2f\nu_g $, the variation is expressed in terms of the $g$--unit normal direction with coefficient $U^2f$. Therefore, factoring out this scaling gives
\begin{align*}
\frac{d}{dt}\mathrm{Area}_g(S_t)\bigg|_{t=0} 
= \int_{S} H_g \cdot U^2 f\, d\sigma_g,
\end{align*}
which implies, by comparing both expressions,
\begin{equation}\label{eqn:H_g}
    H_g =U^{-2} \left(H_e + 4\frac{\nu_e(U)}{U} \right).
\end{equation}
In particular, for a coordinate sphere $S=S_R:=\{|\x| = R\} $, we have
\[
H_e = \frac{2}{R}, \qquad\quad \nu_e = \frac{\x}{R}.
\]
Compute the Euclidean normal derivative of $u$: 
\begin{align*} 
\nu_e(U)(\x) &= \nab U(\x) \cdot \nu_e =\sum_{I=1}^N\left(-\frac{m_I}{2}\frac{\x-\cb_I}{|\x-\cb_I|^3}\right) \cdot \frac{\x}{R}= \sum_{I=1}^N\left(-\frac{m_I}{2R}\frac{\x\c(\x-\cb_I)}{|\x-\cb_I|^3}\right). 
\end{align*} 
Injecting it into \eqref{eqn:H_g} yields
\begin{align*} 
H_g(\x) &= \frac{1}{U(\x)^2}\left[\frac{2}{R}-\frac{4}{U(\x)}\sum_{I=1}^N\frac{m_I(\x\cdot(\x-\cb_I))}{2R|\x-\cb_I|^3}\right]\\
&= \frac{2}{RU(\x)^3}\left[U(\x)-\sum_{I=1}^N\frac{m_I(\x\cdot(\x-\cb_I))}{|\x-\cb_I|^3}\right]. 
\end{align*} 
Using the identity 
\begin{align*} 
\frac{1}{2|\x-\cb_I|}-\frac{\x\cdot(\x-\cb_I)}{|\x-\cb_I|^3} = \frac{|\cb_I|^2-|\x|^2}{2|\x-\cb_I|^3},
\end{align*}
we obtain \eqref{Hgexplicit} as stated. This concludes the proof of Proposition \ref{meancurvaturebl}.
\end{proof}
Next, we focus on spheres centered at one of the poles $\cb_J$:
\begin{prop}\label{meanconvex_local} 
Let $J \in \{1,2, \dots, N\}$ and we denote $S_R^{(J)} := \{\x/\,|\x-\cb_J| = R\}$. Suppose \begin{align}\label{disjointlocal} 
\frac{m_J}{2} \leq R < d_J, 
\end{align} 
where
\begin{equation}\label{dJ}
    d_J:=\min_{I \ne J}|\cb_I-\cb_J|
\end{equation}
is the shortest Euclidean distance from other poles to $\cb_J$, see Figure \ref{fig:meanconvex_multiple} for a geometric description. Then, $S_R^{(J)}$ is strictly mean-convex with respect to $\gBL =U^4 e$.
\end{prop}
\begin{figure}[H]
\centering
\begin{tikzpicture}[scale=1.2]

\coordinate (O) at (0,0);
\draw[fill=black] (O) circle (1pt);

\draw[dashed] (O) circle (0.6);

\draw[thick, blue] (O) circle (2);

\path[fill=blue!10, opacity=0.4] (O) circle (2.5);
\path[fill=white] (O) circle (0.6);
\path[fill=black] (O) circle (0.05);

\coordinate (P1) at (3.2, 0.5);
\coordinate (P2) at (-4.8, 1.2);
\coordinate (P3) at (1.5, -2);

\foreach \p/\name in {P1/c_1, P2/c_2, P3/c_3}{
    \draw[fill=black] (\p) circle (1pt);
    \edef\labelmath{$\name$}
    \draw[gray, dashed] (\p) circle (0.5);
    \draw[dotted] (O) -- (\p);
}

\node at (1.6,0.5) {$|\cb_1 - \cb_J|$};
\node at (0.7,-1.2) {$|\cb_3 - \cb_J|$};
\node at (-2.6,0.5) {$|\cb_2 - \cb_J|$};

\node [right=3pt] at (P1) {$\cb_1$};
\node [right=3pt] at (P2) {$\cb_2$};
\node [right=3pt] at (P3) {$\cb_3$};
\node [below left] at (0,0) {$\cb_J$};
\node [blue] at (90:2.3) {$S_R^{(J)}$};
\node at (120:1) {$R = \frac{m_J}{2}$};
\end{tikzpicture}
\caption{
Illustration of the region foliated by mean-convex coordinate spheres: the shaded annulus satisfies $\frac{m_J}{2} \leq |\x-\cb_J|<d_J$ for the coordinate sphere $S_R^{(J)}$ centered at $\cb_J = 0$. The dashed balls represent excluded regions $\{|\x - \cb_I| < \frac{m_I}{2}\}$.
}
\label{fig:meanconvex_multiple}
\end{figure}
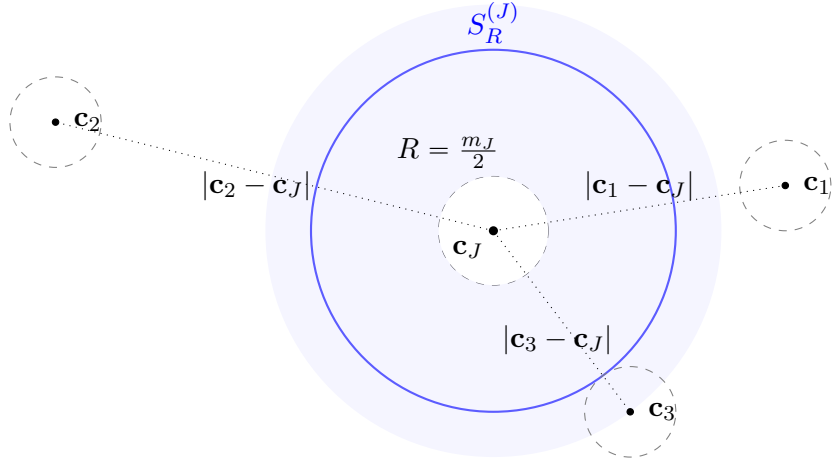
\begin{proof}
Without loss of generality, we assume that $\cb_J=0$. By Proposition \ref{meancurvaturebl}, we have
\begin{align*} 
H_g(\x) = \frac{2}{U(\x)^3R}\left[1-\frac{m_J}{2R}+\sum_{I \ne J}\frac{m_I(|\cb_I|^2-R^2)}{2|\x-\cb_I|^3}\right], \qquad \x\in S_R.
\end{align*}
Since $R<|\cb_I|$ for $I \ne J$, the summation terms are nonnegative. If $R > \frac{m_J}{2}$, clearly $1 - \frac{m_J}{2R}>0$. If $R=\frac{m_J}{2}$, then $1 - \frac{m_J}{2R}=0$, but the summation is strictly positive, ensuring strict positivity. Hence, $H_g>0$ on $S_R^{(J)}$.
\end{proof}
\subsection{Charges of Brill-Lindquist metric}
We now compute the charges of the coordinate spheres centered at one of the poles $\cb_I$ for the Brill-Lindquist metric defined in \eqref{gBLdf}.
\begin{prop}\label{prop:BL-local}
Let $\gBL$ be the Brill-Lindquist metric defined in \eqref{gBLdf}. For any $I\in\{1,2,\dots,N\}$ and $R\in[32, 64]$, we denote
\begin{equation}\label{yIdf}
\y_I:=\x-\cb_I,\qquad\quad\pr B_R^{(I)}:=\big\{|\y_I|=R\big\},\qquad\quad\nu_I:=\frac{\y_I}{R}.
\end{equation}
Assume that the total mass $M$ and the minimum separation distance of the pole $\cb_I$, denoted by $d_I$, satisfies the following small-mass and large-separation condition:\footnote{Notice that \eqref{dfMd} implies that $d_I^{-1}\gg 1$.}
\begin{equation}\label{dfMd}
M:=\sum_{I=1}^{N}m_I\ll 1,
\qquad\quad d_I:=\min\limits_{J\ne I}|\cb_I-\cb_J|\gg m_I^{-1}.
\end{equation}
Then, we have for $I\in\{1,2,\dots,N\}$ and $l=1,2,3$:
\begin{align}
\E\left[(\gBL,0);\pr B_R^{(I)}\right]&=8\pi m_I+O(m_IM+Md_I^{-1}),\label{BL-EE}\\
\P_l\left[(\gBL,0);\pr B_R^{(I)}\right]&=0,\label{BL-PP}\\
\C_l\left[(\gBL,0);\pr B_R^{(I)}\right]&=(\cb_I)_l\big(8\pi m_I+O(m_IM+Md_I^{-1})\big)+O(m_IM+Md_I^{-1}),\label{BL-CC}\\
\J_l\left[(\gBL,0);\pr B_R^{(I)}\right]&=0.\label{BL-JJ}
\end{align}
\end{prop}
\begin{proof}
For $J\ne I$ and $d_I^{-1}\ll 1\leq R$, we have from \eqref{yIdf}
\begin{align*}
\frac{m_J}{2|\x-\cb_J|}&=\frac{m_J}{2\left|\y_I+\cb_I-\cb_J\right|}\\
&=\frac{m_J}{2}\big\langle\y_I+\cb_I-\cb_J,\y_I+\cb_I-\cb_J\big\rangle^{-\frac{1}{2}}\\
&=\frac{m_J}{2}\Big(|\cb_I-\cb_J|^2+2\,\langle\y_I,\cb_I-\cb_J\rangle+|\y_I|^2\Big)^{-\frac{1}{2}}\\
&=\frac{m_J}{2|\cb_I-\cb_J|}(1+O(d_I^{-1}))^{-\frac{1}{2}}\\
&=\frac{m_J}{2|\cb_I-\cb_J|}+O\left(\frac{m_J}{d_I^2}\right),
\end{align*}
which implies
\begin{align*}
    \nu_I\left(\frac{m_J}{2|\x-\cb_J|}\right)=O\left(\frac{m_J}{d_I^2}\right).
\end{align*}
Combining the above estimates, we have on $\pr B_R^{(I)}=\{|\y_I|=R\}$
\begin{align}\label{U}
    U=1+\frac{m_I}{2|\y_I|}+\sum_{J\ne I}\frac{m_J}{2|\x-\cb_J|}=1+\frac{m_I}{2R}+O\left(\frac{M}{d_I}\right).
\end{align}
We also have on $\pr B_R^{(I)}$
\begin{align}\label{nuU}
    \nu_I(U)=\nu_I\left(\frac{m_I}{2|\y_I|}\right)+\sum_{J\ne I}\nu_I\left(\frac{m_J}{2|\x-\cb_J|}\right)=-\frac{m_I}{2R^2}+O\left(\frac{M}{d_I^2}\right).
\end{align}
Applying \eqref{U} and \eqref{nuU}, we obtain on $\pr B_R^{(I)}$
\begin{align}
\begin{split}\label{U3nu}
\nu_I(U^4)&=4U^{3}\nu_I(U)\\
&=4\left(1+\frac{m_I}{2R}+O\left(\frac{M}{d_I}\right)\right)^3\left(-\frac{m_I}{2R^{2}}+O\left(\frac{M}{d_I^2}\right)\right) \\
&=-\frac{2m_I}{R^2}+O(m_IM)+O(Md_I^{-1}).
\end{split}
\end{align}
We now compute from Definition \ref{chargesMOT}, \eqref{gBLdf} and \eqref{U3nu}
\begin{align*}
\E\left[(\gBL,0);\partial B_R^{(I)}\right]&=\frac{1}{2}\int_{\pr B_R^{(I)}}\left(\pr_i g_{ij}-\pr_j g_{ii}\right)(\nu_I)^{j}dS\\
&=\frac{1}{2}\int_{\pr B_R^{(I)}}\big(\pr_j(U^4)-3\pr_j(U^4)\big)(\nu_I)^{j}dS\\
&=-\int_{\pr B_R^{(I)}}\nu_I(U^4)dS\\
&=\int_{\pr B_R^{(i)}}\frac{2m_I}{R^2}+O(m_IM)+O(Md_I^{-1})dS\\
&=8\pi m_I+O(m_IM)+O(Md_I^{-1}),
\end{align*}
which implies \eqref{BL-EE}. Next, we have from \eqref{gBLdf}
\begin{align*}
\pr_i g_{ij}=\pr_j(U^4),\quad\; \pr_j g_{ii}=3\pr_j(U^4),\quad\; (g-e)_{ij}=(U^{4}-1)\delta_{ij},\quad\; (g-e)_{ii}=3(U^4-1).
\end{align*}
We then compute from Definition \ref{chargesMOT} and \eqref{U3nu}
\begin{align*}
&\C_l\left[(\gBL,0);\pr B_R^{(I)}\right]-(\cb_I)_l\E\left[(\gBL,0);\pr B_R^{(I)}\right]\\
=\;&\frac{1}{2}\int_{\pr B_R^{(I)}}\big[(\y_I)_l\left(\pr_i g_{ij}-\pr_j g_{ii}\right)-\de_{il}(g-e)_{ij}+\de_{jl}(g-e)_{ii}\big](\nu_I)^j dS\\
=\;&\frac{1}{2}\int_{\pr B_R^{(I)}}\big[-2(\y_I)_l\pr_j(U^4)+2\de_{jl}(U^4-1)\big](\nu_I)^j dS\\
=\;&\int_{\pr B_R^{(I)}}\left[-\nu_I(U^4)(\y_I)_l+(U^4-1)\frac{(\y_I)_l}{R}\right]dS\\
=\;&\int_{\pr B_R^{(I)}}\left[\frac{4m_I(\y_I)_l}{R^2}+O(m_IM)+O(Md_I^{-1})\right]dS\\
=\;&O(m_IM)+O(Md_I^{-1}),
\end{align*}
where we used the following identity at the last step:
$$
\int_{\partial B_R^{(I)}}(\y_I)_ldS=0,\qquad\forall\; l=1,2,3.
$$
Combining with \eqref{BL-EE}, we infer
\[
\C_l\left[(g,0);\pr B_R^{(I)}\right]=(\cb_I)_l\left(8\pi m_I+O(m_IM)+O(Md_I^{-1})\right)+O(m_IM)+O(Md_I^{-1}),
\]
which implies \eqref{BL-CC}. Finally, \eqref{BL-PP} and \eqref{BL-JJ} follow immediately from Definition \ref{chargesMOT} and the fact that $k=0$. This concludes the proof of Proposition \ref{prop:BL-local}.
\end{proof}
\begin{cor}\label{cor:BL-annulus}
Under the hypotheses of Proposition \ref{prop:BL-local}, we have for $l=1,2,3$:
\begin{align}
\E\left[(\gBL,0);A_{32}^{(I)}\right]&=8\pi m_I+O(m_IM+Md_I^{-1}),\label{annulus-EE}\\
\P_l\left[(\gBL,0);A_{32}^{(I)}\right]&=0,\label{annulus-PP}\\
\C_l\left[((\gBL,0);A_{32}^{(I)}\right]&=(\cb_I)_{l}\big(8\pi m_I+O(m_IM+Md_I^{-1})\big)+O(m_IM+Md_I^{-1}),\label{annulus-CC}\\
\J_l\left[(\gBL,0);A_{32}^{(I)}\right]&=0,\label{annulus-JJ}
\end{align}
where we denote
\[
A_{32}^{(I)}:=A_{32}\left(\cb_I\right)=B_{64}\left(\cb_I\right)\setminus\overline{B_{32}\left(\cb_I\right)},\qquad\forall\; I\in\{1,2,\dots,N\}.
\]
\end{cor}
\begin{proof}
We have from Definition \ref{ADMannulus} and \eqref{BL-EE}
\begin{align*}
\E\left[(g,k);A_{32}^{(I)}\right]&=\int_0^\infty\eta_{32}(r')\E\left[(g,k);\pr B_{r'}^{(I)}\right]\,dr'\\
&=\int_0^\infty \eta_{32}(r')\big[8\pi m_I+O(m_IM+Md_I^{-1})\big]dr'\\
&=8\pi m_I+O(m_IM+Md_I^{-1}),
\end{align*}
which implies \eqref{annulus-EE}. Similarly, we have from Definition \ref{ADMannulus} and \eqref{BL-CC}
\begin{align*}
\C_l\left[(g,k);A_{32}^{(I)}\right]&=\int_0^\infty\eta_{32}(r')\C\left[(g,k);\pr B_{r'}^{(I)}\right]\,dr'\\
&=\int_0^\infty \eta_{32}(r')\Big(\big[8\pi m_I+O(m_IM+Md_I^{-1})\big](\cb_I)_{l}+O(m_IM+Md_I^{-1})\Big)dr'\\
&=(\cb_I)_{l}\big(8\pi m_I+O(m_IM+Md_I^{-1})\big)+O(m_IM+Md_I^{-1}),
\end{align*}
which implies \eqref{annulus-CC}. Finally, \eqref{annulus-PP} and \eqref{annulus-JJ} follow from \eqref{BL-PP} and \eqref{BL-JJ}. This concludes the proof of Corollary \ref{cor:BL-annulus}.
\end{proof}
\subsection{Sobolev norms in annuli}
We now prove the following proposition, which controls the size of $\gBL-e$.
\begin{prop}\label{prop:BL-Sobolev}
Under the hypotheses of Proposition \ref{prop:BL-local}, we have for any $I\in\{1,2,\dots,N\}$ and $s\in\mathbb{N}$:
\begin{equation}\label{eq:Sobolev-bound}
\|\gBL-e\|_{H^{s}(A_{32}^{(I)})}\les m_I.
\end{equation}
\end{prop}
\begin{proof}
Notice that we have
$$
32\leq r:=|\mathbf{y}_I|\leq 64\qquad \mbox{ on }\; A_{32}^{(I)}.
$$
Thus, we have from \eqref{gBLdf}
\begin{align*}
    (g-e)_{ij}=(U^4-1)\de_{ij}=\frac{2m_I}{r}\de_{ij}+O\left(\frac{M}{d_I}\right),
\end{align*}
which implies for all $0\leq |\a|\leq s$
\begin{align*}
    \pr^\a(g-e)=O\left(\frac{m_I}{r^{1+|\a|}}\right)+O\left(\frac{M}{d_I}\right)=O(m_I)+O(Md_I^{-1}),\qquad \mbox{ on }\; A_{32}^{(I)}.
\end{align*}
Thus, we infer for $d_I\gg m_I^{-1}$
\begin{align*}
    \|g-e\|_{H^s(A_{32}^{(I)})}^2\les\int_{32}^{64}dr\int_{\pr B_r}|\pr^\a(g-e)|^2dS\les m_I^2+\frac{M^2}{d_I^2}\les m_I^2.
\end{align*}
This concludes the proof of Proposition \ref{prop:BL-Sobolev}.
\end{proof}
\section{Gluing construction of Cauchy initial data}\label{secgluing}
\subsection{Statement of the main theorem}
The goal of this section is to prove the following theorem, which proves the existence of the desired Cauchy initial data.
\begin{thm}\label{mainCauchy}
Let $N\in\mathbb{N}$ and $s\geq 3$. For any Brill-Lindquist metric $\gBL$ defined in \eqref{gBLdf} and satisfying the small-mass and large-separation condition \eqref{dfMd}, there exists a $2N$--parameter $\{(\de_I,a_I)\}_{I=1}^N$ family of triplets $(\Si,g,k)$, which are solutions of \eqref{constrainteq}, satisfying the following properties:
\begin{enumerate}
    \item For any $I\in\{1,2,\dots,N\}$, we have
    \begin{align*}
        (g,k)&=(\gBL,0)\qquad \mbox{ in }\;B_{32}^c(\cb_I),\\
        (g,k)&=(e,0)\qquad\quad\;\mbox{ in }\;B_{1-2\de_I}(\cb_I),
    \end{align*}
    and the following estimate hold:
    \begin{equation}\label{64-1control}
        \|(g-e,k)\|^2_{H^s\times H^{s-1}(B_{64}(\cb_I)\setminus\ov{B_{1}(\cb_I)})}\les m_I.
    \end{equation}
    \item For any $I\in\{1,2,\dots,N\}$, a trapped surface will form in $D^+(B_1(\cb_I))$, the future domain of dependence of $B_1(\cb_I)$.
    \item $(\Si,g,k)$ is free of trapped surfaces.
    \item For any $I\in\{1,2,\dots,N\}$, the coordinate ball $B_{d_I-32}(\cb_I)$ is free of MOTS with $d_I$ defined in \eqref{dJ}.
    \end{enumerate}
    See Figure \ref{FinalID} for a geometric illustration of $(\Si,g,k)$ in the particular case $N=3$.
\begin{figure}[H]
\centering
\tdplotsetmaincoords{70}{120} 
\begin{tikzpicture}[scale=0.46, tdplot_main_coords]
\fill[blue!20, opacity=0.3] (-13,-12,0) -- (12,-12,0) -- (12,12,0) -- (-13,12,0) -- cycle; 
\node at (2,-7,0) {\scalebox{1}{$\gBL$}};
\newcommand{\AnnulusUnit}[4]{%
    \begin{scope}[shift={(#1,#2,0)}, scale=#3]
\begin{scope}
    \fill[blue!40, opacity=0.8] (0,0,0) ellipse (5 and 5);
    \clip (0,0,0) ellipse (2 and 2);
    \fill[white] (0,0,0) ellipse (2 and 2); 
\end{scope}
\begin{scope}
    \fill[red!90, opacity=0.8] (0,0,0) ellipse (2 and 2);
    \clip (0,0,0) ellipse (1.75 and 1.75);
    \fill[white] (0,0,0) ellipse (1.75 and 1.75); 
\end{scope}
\fill[black] (0,0,0) circle (2pt);
\node[left] at (0,0,0) {\scalebox{0.7}{$\cb_{#4}$}};
\draw (0,0,0) ellipse (2 and 2); 
\draw (0,0,0) ellipse (1.75 and 1.75); 
\draw (0,0,0) ellipse (5 and 5);
\draw (0,0,0) ellipse (8 and 8);
\draw[->, thin, rounded corners=4pt] (0,-1,8) to[out=-90, in=90] (0,1,0);
\node[above] at (0.1,-1.2,8) {\scalebox{0.7}{Euclidean}};
\draw[->, thin, rounded corners=4pt] (0,4,11) to[out=-90, in=90] (0,1.9,0);
\node[above] at (0,4,11) {\scalebox{0.7}{\red{Short-pulse annulus}}};
\draw[->, thin, rounded corners=4pt] (0,-4,5) to[out=-90, in=90] (0,-3,0);
\node[above] at (0,-4,5) {\scalebox{0.7}{\blue{Barrier annulus}}};
\draw[->, thin, rounded corners=4pt] (0,-7,2) to[out=-90, in=90] (0,-6,0);
\node[above] at (0,-8,2) {\scalebox{0.7}{Gluing region}};
\draw[dashed] (0,0,-0.25) ellipse (2*11/12 and 2*11/12); 
\foreach \angle in {-40,110}
{\draw[dashed] (0,0,-3) -- ({2*cos(\angle)},{2*sin(\angle)},0);}
\end{scope}
}
\newcommand{\AnnulusUnitplain}[4]{%
    \begin{scope}[shift={(#1,#2,0)}, scale=#3]
\begin{scope}
    \fill[blue!40, opacity=0.8] (0,0,0) ellipse (5 and 5);
    \clip (0,0,0) ellipse (2 and 2);
    \fill[white] (0,0,0) ellipse (2 and 2); 
\end{scope}
\begin{scope}
    \fill[red!90, opacity=0.8] (0,0,0) ellipse (2 and 2);
    \clip (0,0,0) ellipse (1.75 and 1.75);
    \fill[white] (0,0,0) ellipse (1.75 and 1.75); 
\end{scope}
\fill[black] (0,0,0) circle (2pt);
\node[left] at (0,0.8,0) {\scalebox{0.5}{$\cb_{#4}$}};
\draw (0,0,0) ellipse (2 and 2); 
\draw (0,0,0) ellipse (1.75 and 1.75); 
\draw (0,0,0) ellipse (5 and 5);
\draw (0,0,0) ellipse (8 and 8);
\draw[dashed] (0,0,-0.25) ellipse (2*11/12 and 2*11/12); 
\foreach \angle in {-40,110}
{\draw[dashed] (0,0,-3) -- ({2*cos(\angle)},{2*sin(\angle)},0);}
\end{scope}
}
\AnnulusUnit{5}{5.4}{0.73}{1}
\AnnulusUnitplain{-8}{6}{0.45}{2}
\AnnulusUnitplain{-7}{-8}{0.36}{3}
 \draw[->] (9,-8,0) -- (9,-8,9) node[anchor=south]{\scalebox{0.5}{$t$}};
\end{tikzpicture}
\caption{Cauchy initial data in Theorem \ref{mainCauchy}.}
\label{FinalID}
\end{figure}
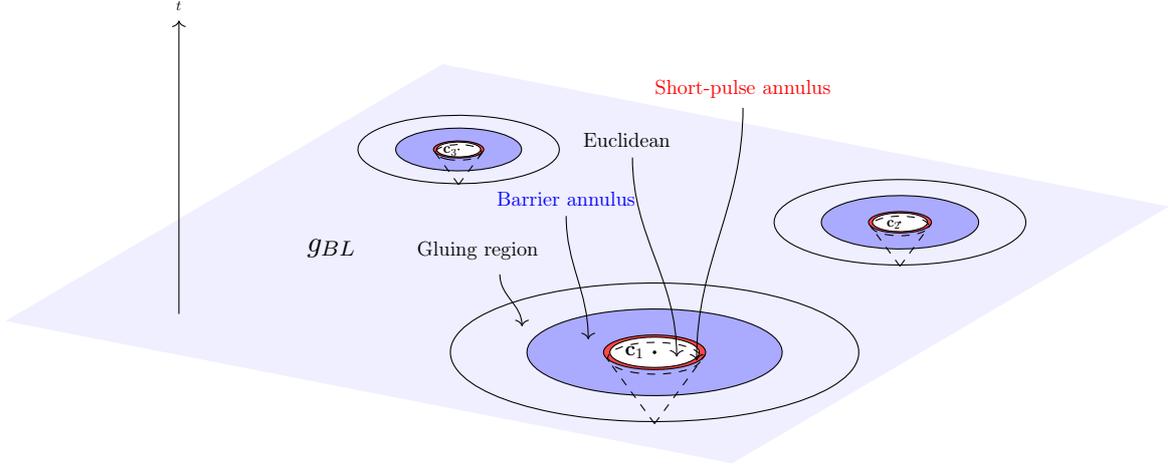
\end{thm}
\begin{rk}\label{rkBH}
Properties 2 and 3 in Theorem \ref{mainCauchy} implies that even if the initial data are free of trapped surfaces, there are $N$ trapped surfaces that will form in $D^+(B_1(\cb_I))$, $I=1,2,\dots,N$, respectively. Assuming the weak cosmic censorship conjecture, we have that $N$ black holes will form in the future of $(\Si,g,k)$ in a finite time.
\end{rk}
\begin{rk}
Property 4 in Theorem \ref{mainCauchy} demonstrates the generic nature of the result, i.e. the formation of trapped surface in the future is not a perturbation of a MOTS.
\end{rk}
\begin{prop}[External Stability of the Constructed Initial Data]\label{rkShenKerr}
Let $(\Si, g, k)$ be the initial data set constructed in Theorem \ref{mainCauchy}. Then the external region of its future development,
\[
D^+(\Sigma \setminus B_R),
\]
is future asymptotically stable and converges to a Kerr spacetime of mass $M$, for any sufficiently large $R \gg \max\limits_{1 \leq I \leq N} |\cb_I|$.
\end{prop}
\begin{proof}
From Theorem~\ref{mainCauchy}, the initial data $(\Sigma, g, k)$ satisfy
\[
g = \left(1 + \frac{M}{2|\x|} \right)^4 e + O\left( \frac{M}{|\x|^2} \right), \qquad k = 0,
\]
in the region $\Sigma \setminus B_R$, where $R$ is chosen large enough to contain all centers $\cb_I$. That is, the data match a spatial Schwarzschild metric of mass $M$ up to $O(M/|\x|^2)$ decay. Using the Riemannian Schwarzschild metric \eqref{Schw-isotermal} in isothermal coordinate, this becomes
\[
g = g_M + O\left( \frac{M}{|\x|^2} \right), \qquad k = 0,
\]
where $g_M = \left(1 + \frac{M}{2|\x|} \right)^4 e$ is the spatial Schwarzschild metric of mass $M$.

By the main result of \cite{ShenKerr}, the Einstein vacuum equations evolve such asymptotically Schwarzschild data into a spacetime that is future asymptotically close to a Kerr solution. In particular, the external domain of dependence $D^+(\Sigma \setminus B_R)$ is future asymptotically stable and converges to the external region of a Kerr black hole with mass $M$.
\end{proof}
\begin{que}[Interior Evolution]
While the external region $D^+(\Sigma \setminus B_R)$ evolves to the exterior region of a Kerr spacetime, the future evolution of the interior region $J^+(\Sigma \cap B_R)$ remains unknown.\footnote{Here, $J^+(\Si\cap B_R)$ denotes the causal future of $\Si\cap B_R$.} In particular:
\begin{itemize}
  \item How do the individual black holes interact with each other during the evolution?
  \item Is the full future development of $(\Si, g, k)$ globally asymptotically Kerr?
\end{itemize}
Understanding the nonlinear dynamics and potential mergers in the interior region is a fundamental open problem in mathematical general relativity. These questions are central to the relativistic $N$--body problem and are closely related to the \emph{Final State Conjecture} (FSC), which posits that generic asymptotically flat vacuum data evolve, in the large, toward a configuration of finitely many Kerr black holes plus radiative decay. See the Introduction of \cite{Ksurvey} for a broad overview of this conjecture and its connection to black hole stability, rigidity, and merger dynamics.
\end{que}
\begin{rk}\label{rkcountable}
Theorem \ref{mainCauchy} shows that for any $N\in\mathbb{N}$, there exists a family of triplets $(\Si,g,k)$, which solves \eqref{constrainteq} and evolves to $N$ trapped surfaces. The proof also applies for the Brill--Lindquist metric \eqref{gBLdf} with countably many poles $\{\cb_{I}\}_{I=1}^\infty$ under the following small-mass and large-separation condition:
\begin{align*}
    M:=\sum_{I=1}^\infty m_I\ll 1,\qquad\quad d_I:=\min_{J\ne I} |\cb_I-\cb_J|\gg  m_I^{-1}\gg 1.
\end{align*}
More precisely, there exists a family of solutions of \eqref{constrainteq}, which evolves to countably many trapped surfaces. As in Remark \ref{rkBH}, we see that countably many black holes will form in a finite time by assuming the weak cosmic censorship conjecture. However, the external stability in Proposition \ref{rkShenKerr} does not apply in this countable case. We also mention the work \cite{ACP}, which constructs scalar-flat initial data with infinitely many minimal spheres using a gluing scheme based on the Brill--Lindquist metric.
\end{rk}
The proof of Theorem \ref{mainCauchy} is given in the following steps:
\begin{itemize}
    \item In Section \ref{ssecint}, we calculate the charges of the barrier annulus constructed in Theorem \ref{interiorsolution}. We also estimate the Sobolev norms of $(g-e,k)$ on the barrier annulus.
    \item In Section \ref{ssecout}, we summarize the properties of the Brill-Lindquist metric $\gBL$ introduced in Section \ref{secBL}, which will be used in the next section.
    \item In Section \ref{ssecglue}, we perform surgery for $\gBL$ near $\{\cb_I\}_{I=1}^N$. More precisely, for any $I\in\{1,2,\dots,N\}$, we remove the ball $B_{32}(\cb_I)$ and then apply Theorem \ref{thm:MOT1.7} to patch the constant-time slices $\Si(\de_I,a_I)$ constructed in Theorem \ref{interiorsolution}. This constructs the desired Cauchy data in Theorem \ref{mainCauchy}.
    \item In Section \ref{ssecNtrap}, we show that, for the Cauchy data constructed in Section \ref{ssecglue}, $N$ trapped surfaces will form in $D^+(B_1(\cb_I))$, $I\in\{1,2,\dots,N\}$, the future domain of dependence of balls $B_1(\cb_I)$.
    \item In Section \ref{ssecfree}, we prove that the Cauchy data constructed in Section \ref{ssecglue} are free of trapped surfaces and MOTS.
\end{itemize}
\subsection{Interior solutions}\label{ssecint}
\begin{prop}\label{ginkin}
Let $(\Si,g,k):=(\Si_{\de,a},g_{\de,a},k_{\de,a})$ be the constant-time slice constructed in Theorem \ref{interiorsolution} with the parameter $(\de,a)$. Then, we have
    \begin{align}\label{interiorQ}
        \Q[(g,k);A_1]=O(a^{-1}),
    \end{align}
    where $\Q[(g,k);A_1]$ is defined in \eqref{eq:charge-avg} and $A_1:=B_2\setminus\ov{B_1}$. Moreover, we have
    \begin{align}\label{interiorSobolev}
        \|(g-e,k)\|_{H^s\times H^{s-1}(A_1)}\les a^{-1}.
    \end{align}
\end{prop}
\begin{proof}
We have from \eqref{g-eta}
    \begin{align}\label{interiorCs}
        \|(g-e,k)\|_{C^s\times C^{s-1}(A_1)}\les a^{-1},
    \end{align}
which implies \eqref{interiorSobolev}. Next, we have from Definition \ref{chargesMOT} that for all $r\in(1,2)$
    \begin{align*}
        \E[(g,k);\pr B_r]&=\frac{1}{2}\int_{\pr B_r}(\pr_ig_{ij}-\pr_jg_{ii})\nu^jdS,\\
        \P_i[(g,k);\pr B_r]&=\int_{\pr B_r}(k_{ij}-\de_{ij}\tr_ek)\nu^jdS,\\
        \C_l[(g,k);\pr B_r]&=\frac{1}{2}\int_{\pr B_r}\big(x_l\pr_ig_{ij}-x_l\pr_jg_{ii}-\de_{il}(g-e)_{ij}+\de_{jl}(g-e)_{ii}\big)\nu^jdS,\\
        \J_l[(g,k);\pr B_r]&=\int_{\pr B_r}(k_{ij}-\de_{ij}\tr_e k)Y_l^i\nu^jdS.
    \end{align*}
    Applying \eqref{interiorCs}, we obtain for $\Q$ defined in \eqref{dfQ}
    \begin{align*}
        \Q[(g,k);\pr B_r]=O(a^{-1}),\qquad \forall\; r\in(1,2).
    \end{align*}
    Combining with \eqref{eq:charge-avg}, we deduce \eqref{interiorQ}. This concludes the proof of Proposition \ref{ginkin}.
\end{proof}
\subsection{Exterior solution}\label{ssecout}
\begin{prop}\label{exteriorQ}
Let $\gBL$ be the Brill-Lindquist metric defined in \eqref{gBLdf} satisfying the small-mass and large-separation condition \eqref{dfMd}. Let $(x^1,x^2,x^3)$ be a coordinate system centered at $\cb_I$.\footnote{In other word, we shift the center of the coordinate system $(x^1,x^2,x^3)$ to $\cb_I$.} Then, we have for $I\in\{1,2,\dots,N\}$:
\begin{align}
\begin{split}\label{BLcharges}
\E\left[(\gBL,0);A_{32}^{(I)}\right]&=8\pi m_I+O(m_IM+Md_I^{-1}),\\
\P_l\left[(\gBL,0);A_{32}^{(I)}\right]&=0,\\
\C_l\left[((\gBL,0);A_{32}^{(I)}\right]&=O(m_IM+Md_I^{-1}),\\
\J_l\left[(\gBL,0);A_{32}^{(I)}\right]&=0.
\end{split}
\end{align}
Moreover, we have the following Sobolev estimates:
\begin{equation}\label{Sobolev-bound}
\|(\gBL-e,0)\|_{H^{s}\times H^{s-1}(A_{32}^{(I)})}\les m_I.
\end{equation}
\end{prop}
\begin{proof}
    Notice that \eqref{BLcharges} follows from Corollary \ref{cor:BL-annulus} the fact that $\cb_I=\mathbf{0}$ in this coordinate system. Moreover, \eqref{Sobolev-bound} follows directly from Proposition \ref{prop:BL-Sobolev}.
\end{proof}
\subsection{Gluing argument}\label{ssecglue}
We are now ready to construct the Cauchy initial data in Theorem \ref{mainCauchy}.
\begin{prop}\label{propglue}
For all $N\in\mathbb{N}$ and $s\geq 3$, there exists an initial data $(g,k)$ solving \eqref{constrainteq} and satisfying for all $I\in\{1,2,\dots,N\}$:
\begin{align*}
    (g,k)&=(\gBL,0)\qquad \mbox{ in }\;B_{32}^c(\cb_I),\\
    (g,k)&=(e,0)\qquad\quad\;\mbox{ in }\;B_{1-2\de_I}(\cb_I),\\
    \|(g-e,k)\|^2_{H^s\times H^{s-1}(B_{64}(\cb_I)\setminus \ov{B_{1}(\cb_I)})}&\les m_I.
\end{align*}
\end{prop}
\begin{proof}
The idea of the proof is to perform surgery for the Brill-Lindquist metric $\gBL$ near $\cb_I$. More precisely, we apply the obstruction-free gluing theorem for annuli, as established in \cite{MOT} (restated here as Theorem \ref{thm:MOT1.7}), to patch the constant-time slices $(\Si_{\de_I,a_I},g_{\de_I,a_I},k_{\de_I,a_I})$ constructed in Theorem \ref{interiorsolution}. The proof is divided into 4 steps.\\ \\
{\bf Step 1.} Difference of interior charges and exterior charges.\\ \\
We denote
\begin{align*}
    (g_{in},k_{in}):=(g_{\de_1,a_1},k_{\de_1,a_1}),\qquad\quad (g_{out},k_{out})=(\gBL,0),
\end{align*}
where $(\Si_{\de_1,a_1},g_{\de_1,a_1},k_{\de_1,a_1})$ is the constant-time slice constructed in Theorem \ref{interiorsolution}. Let $(x^1,x^2,x^3)$ be a coordinate system centered at $\cb_1$. We have from Proposition \ref{exteriorQ}
\begin{align*}
\E\left[(g_{out},k_{out});A_{32}^{(1)}\right]&=8\pi m_1+O(m_1M+Md_1^{-1}),\\
\P_l\left[(g_{out},k_{out});A_{32}^{(1)}\right]&=0,\\
\C_l\left[((g_{out},k_{out});A_{32}^{(1)}\right]&=O(m_1M+Md_1^{-1}),\\
\J_l\left[(g_{out},k_{out});A_{32}^{(1)}\right]&=0.
\end{align*}
We also have from \eqref{interiorQ}
\begin{align*}
    \Q[(g_{in},k_{in});A_1]=O(a_1^{-1}),\qquad\quad \Q=(\E,\P,\C,\J).
\end{align*}
Thus, we obtain\footnote{Recall that $A\simeq B$ means $B\les A\les B$.}
\begin{align}\label{DeEestimate}
    \begin{split}
    \De\E&=\E\left[(g_{out},k_{out});A_{32}^{(1)}\right]-\E\left[(g_{in},k_{in});A_1^{(1)}\right]\\
    &=8\pi m_1+O(m_1M+Md_1^{-1}+a_1^{-1})\\
    &\simeq m_1,
    \end{split}
\end{align}
where we used $M\ll 1$, $d_1^{-1}\ll m_1$ and $a_1^{-1}\ll m_1$. We also have
\begin{align}\label{DePCJestimate}
    (\De\P,\De\C,\De\J)=O(m_1M+Md_1^{-1}+a_1^{-1}).
\end{align}
As an immediate consequence of \eqref{DeEestimate} and \eqref{DePCJestimate}, we have
\begin{align*}
    \De\E\gg |\De\P|,\qquad\quad\De\E\ll 1,\qquad\quad|\De\C|+|\De\J|\ll\De\E.
\end{align*}
These imply that the conditions \eqref{eq:obs-free-unit:EP}--\eqref{eq:obs-free-unit:CJ} in Theorem \ref{thm:MOT1.7} hold.\footnote{Let $\Ga=2$ in Theorem \ref{thm:MOT1.7}, the constant $\eps_o$ become uniform constant which only depends on $s$.} \\ \\
{\bf Step 2.} Estimates for Sobolev norms.\\ \\
We have from \eqref{interiorSobolev} and \eqref{Sobolev-bound}
\begin{align*}
    \|(g_{in}-e,k_{in})\|_{H^s\times H^{s-1}(A_1^{(1)})}&\les a_1^{-1},\\
    \|(g_{out}-e,k_{out})\|_{H^s\times H^{s-1}(A_{32}^{(1)})}&\les m_1.
\end{align*}
Thus, we obtain for $a_1^{-1}\ll m_1\ll 1$
\begin{align*}
    \|(g_{in}-e,k_{in})\|_{H^s\times H^{s-1}(A_1^{(1)})}^2+\|(g_{out}-e,k_{out})\|_{H^s\times H^{s-1}(A_{32}^{(1)})}^2\les a_1^{-2}+m_1^2\ll m_1\simeq \De\E,
\end{align*}
which implies that \eqref{eq:obs-free-conc} in Theorem \ref{thm:MOT1.7} holds.\\ \\
{\bf Step 3.} Gluing construction near $\cb_1$.\\ \\
Since the conditions \eqref{eq:obs-free-unit:EP}--\eqref{eq:obs-free-conc} in Theorem \ref{thm:MOT1.7} hold, we have from Theorem \ref{thm:MOT1.7} that there exists a solution of \eqref{constrainteq}
$$
(\gt,\kt)\in H^{s}\times H^{s-1}\left(B_{64}(\cb_1)\setminus\ov{B_1(\cb_1)}\right),
$$
which satisfies\footnote{Recall that $A_r^{(I)}=B_{2r}(\cb_I)\setminus\ov{B_r(\cb_I)}$.}
\begin{align*}
    (\gt,\kt)=(g_{\de_1,a_1},k_{\de_1,a_1}) \quad\mbox{ in }\;A_1^{(1)},\qquad (\gt,\kt)=(\gBL,0) \quad\mbox{ in }\; A_{32}^{(1)},
\end{align*}
and the following estimate hold:
\begin{equation}
    \|(\gt-e,\kt)\|_{H^s\times H^{s-1}\left(B_{64}(\cb_1)\setminus\ov{B_1(\cb_1)}\right)}^2\les\De\E\simeq m_1.
\end{equation}
We then define $(g^1,k^1)$ as follows:
        \begin{align*}
            (g^1,k^1)&=(\gBL,0)\qquad\qquad\;\mbox{ in }\;B_{32}^c(\cb_1),\\
            (g^1,k^1)&=(\gt,\kt)\qquad\quad\qquad\;\,\mbox{ in }\;B_{64}(\cb_1)\setminus\ov{B_1(\cb_1)},\\
            (g^1,k^1)&=(g_{\de_1,a_1},k_{\de_1,a_1})\quad\;\;\;\,\mbox{ in }\;B_{1}(\cb_1),
        \end{align*}
which is well defined in $\RRR^3\setminus\{\cb_I\}_{I=2}^N$ and solves \eqref{constrainteq}. Moreover, we have by construction
\begin{align*}
    (g^1,k^1)&=(\gBL,0)\qquad \mbox{ in }\;B_{32}^c(\cb_1),\\
    (g^1,k^1)&=(e,0)\qquad\quad\;\;\,\mbox{ in }\;B_{1-2\de_1}(\cb_1),\\
    \|(g^1-e,k^1)\|^2_{H^s\times H^{s-1}(B_{64}(\cb_1)\setminus \ov{B_{1}(\cb_1)})}&\les m_1.
\end{align*}
{\bf Step 4.} Gluing construction of $(g,k)$.\\ \\
We now construct the desired Cauchy data $(g,k)$ by induction. Assume that there exists $(g^J,k^J)$ solving \eqref{constrainteq} in $\RRR^3\setminus\{\cb_I\}_{I=J+1}^N$ and satisfies for all $I\in\{1,2,\dots,J\}$:
\begin{align*}
    (g^J,k^J)&=(\gBL,0)\qquad \mbox{ in }\;B_{32}^c(\cb_I),\\
    (g^J,k^J)&=(e,0)\qquad\quad\;\mbox{ in }\;B_{1-2\de_I}(\cb_I),\\
    \|(g^J-e,k^J)\|^2_{H^s\times H^{s-1}(B_{64}(\cb_I)\setminus \ov{B_{1}(\cb_I)})}&\les m_I.
\end{align*}
Then, we introduce a coordinate system $(x^1,x^2,x^3)$ centered at $\cb_{J+1}$. We denote
\begin{align*}
    (g_{in},k_{in}):=(g_{\de_{J+1},a_{J+1}},k_{\de_{J+1},a_{J+1}}),\qquad\quad (g_{out},k_{out}):=(g^J,k^J).
\end{align*}
Notice from $d_J\gg 1$ that we have $(g_{out},k_{out})=(\gBL,0)$ in $A_{32}^{(J+1)}$, proceeding as in Step 1, we infer
\begin{align*}
    \De\E&\simeq m_{J+1},\qquad\quad|\De\P,\De\C,\De\J|\les m_{J+1}M+Md_{J+1}^{-1}+a_{J+1}^{-1}.
\end{align*}
Thus, we obtain that \eqref{eq:obs-free-unit:EP}--\eqref{eq:obs-free-unit:CJ} in Theorem \ref{thm:MOT1.7} hold. Similarly, proceeding as in Step 2, we have
\begin{align*}
    \|(g_{in}-e,k_{in})\|_{H^s\times H^{s-1}(A_1^{(J+1)})}&\les a_{J+1}^{-1},\\
    \|(g_{out}-e,k_{out})\|_{H^s\times H^{s-1}(A_{32}^{(J+1)})}&\les m_{J+1},
\end{align*}
which implies \eqref{eq:obs-free-unit:data} in Theorem \ref{thm:MOT1.7} holds. Applying Theorem \ref{thm:MOT1.7} and proceeding as in Step 3, we deduce the existence of $(g^{J+1},k^{J+1})$, which solves \eqref{constrainteq} and satisfies for all $I\in\{1,2,\dots,J+1\}$:
\begin{align*}
    (g^{J+1},k^{J+1})&=(\gBL,0)\qquad\;\; \mbox{ in }\;B_{32}^c(\cb_I),\\
    (g^{J+1},k^{J+1})&=(e,0)\qquad\quad\;\;\;\mbox{ in }\;B_{1-2\de_I}(\cb_I),\\
    \|(g^{J+1}-e,k^{J+1})\|^2_{H^s\times H^{s-1}(B_{64}(\cb_I)\setminus \ov{B_{1}(\cb_I)})}&\les m_I.
\end{align*}
We have by induction that there exists $(g,k):=(g^N,k^N)$ solving \eqref{constrainteq} and satisfies for all $I\in\{1,2,\dots,N\}$:
\begin{align*}
    (g,k)&=(\gBL,0)\qquad\mbox{ in }\;B_{32}^c(\cb_I),\\
    (g,k)&=(e,0)\qquad\quad\;\mbox{ in }\;B_{1-2\de_I}(\cb_I),\\
    \|(g-e,k)\|^2_{H^s\times H^{s-1}(B_{64}(\cb_I)\setminus \ov{B_1(\cb_I)})}&\les m_I.
\end{align*}
This concludes the proof of Proposition \ref{propglue}.
\end{proof}
\subsection{Formation of multiple trapped surfaces}\label{ssecNtrap}
\begin{prop}\label{multipletrapped}
Let $(\Si,g,k)$ be the Cauchy initial data constructed in Proposition \ref{propglue}. Then, for any $I\in\{1,2,\dots,N\}$, a trapped surface will form in the future domain of dependence of $B_1(\cb_I)$.
\end{prop}
\begin{proof}
    For any $I\in\{1,2,\dots,N\}$, the ball $B_1(\cb_I)$ is obtained from the evolution of the characteristic initial data given in Theorem \ref{shortpulsecone}. The future domain of dependence $D^+(B_1(\cb_I))$ is causally independent of $B_1^c(\cb_I)$. Thus, we have from Theorem \ref{shortpulsecone} that a trapped surface will form in $D^+(B_1(\cb_I))$. Hence, $N$ trapped surfaces will form, respectively, in $D^+(B_1(\cb_I))$ for $I\in\{1,2,\dots,N\}$.
\end{proof}
\begin{rk}
    The conclusion of Proposition \ref{multipletrapped} also holds if we replace the lower bound condition \eqref{geqa} by the following anisotropic lower bound condition:
        \begin{align}\label{geqaaniso}
        \sup_{(x^1,x^2)\in\mathbb{S}^2}\int_0^\de |u_\infty|^2|\hch|^2(u_\infty,\ub',x^1,x^2)d\ub'\geq \de a.
    \end{align}
    In fact, the lower bound condition \eqref{geqa} is only used to ensure the formation of trapped surfaces. However, as a consequence of Theorem 3 in \cite{KLR}, the weaker condition \eqref{geqaaniso} also implies the formation of trapped surfaces.
\end{rk}
\subsection{Free of trapped surfaces}\label{ssecfree}
The following comparison principle for mean curvature plays an essential role in the proof of Proposition \ref{notrapped}.
\begin{lem}\label{meancurvaturecomparison}
    Let $S_1$ and $S_2$ be two closed $2$--surfaces in $(\Si,g)$. Assume that $S_1$ is enclosed in $S_2$ and they are tangent at $p\in S_1\cap S_2$. Then, we have
    \begin{align*}
        H_1(p)\geq H_2(p)
    \end{align*}
    where $H_i$ denotes the mean curvature of $S_i$ for $i=1,2$.
\end{lem}
\begin{proof}
Let $\nu$ be the common outward unit normal at $p$. In a local coordinate chart around $p$, we write $S_i$ as the graph of a function $u_i:U \subset \Si \to \RRR$ over the common tangent plane, so that $u_1 \geq u_2$, with equality and zero gradient at $p$. Recall that the mean curvature at $p$ is given by
\[H_i(p) = \div_g \left( \frac{\nab u_i}{\sqrt{1+|\nab u_i|^2}}\right)(p).\]
Since $\nab u_1(p) = \nab u_2(p) = 0$, this reduces to $H_i(p)=\De_g u_i(p)$, where $\De_g$ is the Laplacian with respect to $g$. As $u_1-u_2\geq 0$ with a local minimum at $p$, we have
\begin{equation*}
\De_g u_1(p) \geq \De_g u_2(p).
\end{equation*}
This concludes the proof of Lemma \ref{meancurvaturecomparison}.
\end{proof}
We are now ready to prove the following proposition.
\begin{prop}\label{notrapped}
    The initial data $(\Si,g,k)$ constructed in Proposition \ref{propglue} are free of trapped surfaces. Moreover, $B_{d_I-32}(\cb_I)$ are free of MOTS for any $I\in\{1,2,\dots,N\}$.
\end{prop}
Motivated by Conjecture 3.1 in \cite{SormaniStavrov}, which states that in (vacuum) geometrostatic manifolds with small mass and large separation, the only outermost minimal surfaces are those enclosing individual poles, we pose the following question regarding the global geometry of our constructed initial data:
\begin{que}[Absence of MOTS in the constructed initial data]
Is it true that the initial data $(\Si, g, k)$ constructed in Proposition~\ref{propglue} contains no marginally outer trapped surfaces (MOTS) anywhere?
\end{que}
\begin{proof}[Proof of Proposition \ref{notrapped}] The proof is divided into 2 steps.\\ \\
{\bf Step 1.} Absence of trapped surfaces and MOTS in $B_{d_I-32}(\cb_I)$.\\ \\
Let $I\in\{1,2,\dots,N\}$. We have from Proposition \ref{propglue}
\begin{align*}
     \|(g-e,k)\|^2_{H^s\times H^{s-1}(B_{64}(\cb_I)\setminus \ov{B_{1}(\cb_I)})}\les m_I.
\end{align*}
Consider the radial foliation given by $r=|\x|$, we obtain for $s\geq 3$
\begin{align*}
    |k|\les M^\frac{1}{2},\qquad\quad \left|\tr_\slg\th-\frac{2}{r}\right|\les M^\frac{1}{2}\qquad\mbox{ on }\;\pr B_r(\cb_I),\quad\forall\; r\in(1,64),
\end{align*}
where $\tr_\slg\th$ denotes the mean curvature of $\pr B_r(\cb_I)$. Thus, we obtain for $M\ll 1$
\begin{align*}
    \tr_\slg\th>\frac{1}{40}>0\qquad\mbox{ on }\;\pr B_r(\cb_I),\quad\forall\; r\in(1,64).
\end{align*}
We also have from \eqref{trchctrchbcnotrapping} and \eqref{chichib} that for $a_I\gg 1$ 
\begin{align*}
    \tr_\slg\th=\trch-\trchb>1\qquad\mbox{ on }\;\pr B_r(\cb_I),\quad\forall\; r\in(0,1].
\end{align*}
Combining the above inequalities, we infer
\begin{align}
    \tr_\slg\th>\frac{1}{40}>0\qquad\mbox{ on }\;\pr B_r(\cb_I),\quad\forall\; r\in(0,64).
\end{align}
Next, recalling from Proposition \ref{propglue} that $(g,k)=(\gBL,0)$ in $B_{d_I-32}(\cb_I)\setminus\ov{B_{32}(\cb_I)}$, we have from Proposition \ref{meanconvex_local}
\begin{align}\label{extnomots}
    \tr_\slg\th>0\qquad\mbox{ on }\;\pr B_r(\cb_I),\quad\forall\; r\in(32,d_I-32).
\end{align}
For any $S$ be a compact $2$--surface embedded in $\Si$ satisfying $S\subseteq B_{d_I-32}(\cb_I)$, we denote $B_{r_S}(\cb_I)$ be the innermost ball centered at $\cb_I$ which includes $S$. Then, $S$ and $S_{r_S}:=\pr B_{r_S}(\cb_I)$ tangent at some point $p$. Applying Lemma \ref{meancurvaturecomparison}, we obtain
    \begin{align}\label{meancurvaturecom}
        \tr_\slg\th_{S}\big|_p\geq \tr_\slg\th_{S_r}\big|_p.
    \end{align}
In the case $r_S<64$, we have from \eqref{meancurvaturecom} and the fact that $|k|\les M^\frac{1}{2}$
\begin{align}\label{trch>0}
    \tr_\slg(\th_S-k_S)\big|_p>\frac{1}{50}>0.
\end{align}
In the case $r_S\in(32,d_I-32)$, we have from \eqref{extnomots}, \eqref{meancurvaturecom} and the fact that $k|_p=0$
\begin{align*}
    \tr_\slg(\th_S-k_S)\big|_p\geq\tr_\slg(\th_{S_r})>0.
\end{align*}
Thus, we deduce that $S$ is neither a trapped surface nor a MOTS.\\ \\
{\bf Step 2.} Absence of trapped surfaces in $\Si_{ext}:=\bigcup_{I=1}^N B_{32}^c(\cb_I)$.\\ \\
We have by construction in Proposition \ref{propglue}
    \begin{align*}
        (g,k)=(\gBL,0) \qquad\mbox{ in }\; \Si_{ext}.
    \end{align*}
    Let $S\subseteq\Si$ be a compact, embedded smooth $2$--surface satisfying $S\cap\Si_{ext}\ne\emptyset$. Notice that we have
    \begin{align*}
        \tr_\slg(-\th-k)\tr_\slg(\th-k)=(\tr_\slg k)^2-(\tr_\slg\th)^2,
    \end{align*}
    which implies at $p\in S\cap\Si_{ext}$
    \begin{align}\label{extnotrapping}
        \tr_\slg(-\th-k)\tr_\slg(\th-k)\big|_p=-(\tr_\slg\th)^2\big|_p+(\tr_\slg k)^2\big|_p=-(\tr_\slg\th)^2\big|_p\leq 0,
    \end{align}
    where we used the fact that $k|_p$. Note that \eqref{extnotrapping} contradicts \eqref{dftrapped} at $p\in S$. Thus, we deduce that $S$ is not a trapped surface. This concludes the proof of Proposition \ref{notrapped}.
\end{proof}
Combining Propositions~\ref{propglue}, \ref{multipletrapped}, and \ref{notrapped}, this concludes the proof of Theorem~\ref{mainCauchy}.

\vspace{0.1cm}
\small{Dawei Shen: Department of Mathematics, Columbia University, New York, NY, 10027. \\
Email: \textit{ds4350@columbia.edu}\\ \\
Jingbo Wan: Department of Mathematics, Columbia University, New York, NY, 10027.\\
Email: \textit{jingbowan@math.columbia.edu}}

\begin{thebibliography}{30}
\addcontentsline{toc}{section}{References}
\bibitem{An12} X.~An, \emph{Formation of Trapped Surfaces from Past Null infinity}, arXiv:1207.5271.
\bibitem{An} X.~An, \emph{A Scale-Critical Trapped Surface Formation Criterion: A New Proof Via Signature for Decay Rates}, Ann. PDE \textbf{8} (1), Art. 3, 2022.
\bibitem{An24} X.~An, \emph{Naked singularity censoring with anisotropic apparent horizon}, To appear in Ann. Math.
\bibitem{AnAth} X.~An and N.~Athanasiou, \emph{A scale-critical trapped surface formation criterion for the Einstein-Maxwell system}, J. Math. Pures Appl. \textbf{167} (9), 294--409, 2022.
\bibitem{AnLuk} X.~An and J.~Luk, \emph{Trapped surfaces in vacuum arising dynamically from mild incoming radiation}, Adv. Theor. Math. Phys. \textbf{21} (1), 1--120, 2017.
\bibitem{ACP} J.~Anderson, J.~Corvino and F.~Pasqualotto, \emph{Multi-localized time-symmetric initial data for the Einstein vacuum equations}, J. Reine Angew. Math. \textbf{808}, 67--110, 2024.
\bibitem{Are1} S.~Aretakis, \emph{The characteristic gluing problem and conservation laws for the wave equation on null hypersurfaces}, Ann. PDE \textbf{3}, Paper No. 3, 56 pp, 2017.
\bibitem{Are2} S.~Aretakis, \emph{On a foliation-covariant elliptic operator on null hypersurfaces}, Int. Math. Res. Not. IMRN \textbf{2015} (24), 6433--6469.
\bibitem{ACR23} S.~Aretakis, S~Czimek and I~Rodnianski, \emph{Characteristic gluing
to the Kerr family and application to spacelike gluing}, Commun. Math. Phys., \textbf{403} (1), 275--327, 2023.
\bibitem{ACR2} S.~Aretakis, S.~Czimek and I.~Rodnianski, \emph{The characteristic gluing problem for the Einstein vacuum equations: linear and nonlinear analysis}, Ann. Henri Poincaré \textbf{25} (6), 2024.
\bibitem{ACR3} S.~Aretakis, S.~Czimek and I.~Rodnianski, \emph{The characteristic gluing problem for the Einstein equations and applications}, Duke Math. J. \textbf{174} (2), 355--402, 2025.
\bibitem{AMY} N.~Athanasiou, P.~Mondal and S.-T.~Yau, \emph{Formation of trapped surfaces in the Einstein-Yang-Mills system}, arXiv:2302.06915.
\bibitem{Bartnik} R.~Bartnik, \emph{The mass of an asymptotically flat manifold}, Commun. Pure Appl. Math. \textbf{39}, 661--693, 1986.
\bibitem{Bieri} L.~Bieri, \emph{An extension of the stability theorem of the Minkowski space in general relativity}, J. Diff. Geom. Vol. \textbf{86} (1), 17--70, 2010.
\bibitem{BL} D.~Brill and R.~Lindquist, \emph{Interaction energy in geometrostatics}, Phys. Rev. \textbf{131} (1963), 471--476.
\bibitem{Bogovskii} M.~E.~Bogovskii, \emph{Solution of the first boundary value problem for the equation of continuity of an incompressible medium}, Dokl. Akad. Nauk SSSR \textbf{248} (1979), 1037–1040.
\bibitem{Cor} J.~Corvino, \emph{Scalar curvature deformation and a gluing construction for the Einstein constraint equations}, Commun. Math. Phys. \textbf{214}, 137--189, 2000.
\bibitem{CS} A.~Carlotto and R.~Schoen, \emph{Localizing solutions of the Einstein constraint equations}, Invent. Math. \textbf{205}, no. 3, 559--615, 2016.
\bibitem{CSglue}  J.~Corvino and R.~Schoen, \emph{On the asymptotics for the vacuum Einstein constraint equations}, J. Diff. Geom. Vol. \textbf{73} (2), 185--217, 2006.
\bibitem{Chen-K} X.~Chen and S.~Klainerman, \emph{Formation of trapped surfaces in geodesic foliation}, arXiv:2409.14582.
\bibitem{cb} Y.~Choquet-Bruhat, \emph{Th\'eor\`eme d’existence pour certains syst\`emes d’\'equations aux d\'eriv\`ees partielles non lin\'eaires}, Acta Math. \textbf{88}, 141--225, 1952.
\bibitem{cbg} Y.~Choquet-Bruhat and R.~Geroch \emph{Global aspects of the Cauchy problem in general relativity}, Comm. Math. Phys. \textbf{14}, 329--335, 1969.
\bibitem{Chr99} D.~Christodoulou, \emph{On the global initial value problem and the issue of singularities}, Class. Quantum Grav. \textbf{16}, A23--A35, 1999.
\bibitem{Chr} D.~Christodoulou, \textit{The formation of Black Holes in General Relativity}, EMS Monographs in Mathematics, 2009.
\bibitem{Ch-Kl} D.~Christodoulou and S.~Klainerman, \emph{The Global Nonlinear Stability of Minkowski Space}, Princeton Mathematical Series \textbf{41}, 1993.
\bibitem{CCI} P.~T.~Chruściel, J.~Corvino and J.~Isenberg, \emph{Construction of N-body initial data sets in general relativity}, Commun. Math. Phys., \textbf{304} (3), 637--647, 2011.
\bibitem{CR} S.~Czimek and I.~Rodnianski, \emph{Obstruction-free gluing for the Einstein equations}, arXiv:2210.09663.
\bibitem{Da-Ro} M.~Dafermos and I.~Rodnianski, \emph{A new physical-space approach to decay for the wave equation with applications to black hole spacetimes}, XVIth International Congress on Mathematical Physics, World Sci. Publ., Hackensack, NJ, 2010, 421--432.
\bibitem{Dilts} J.~Dilts, J.~Isenberg, R.~Mazzeo and C.~Meier, \emph{Non-CMC solutions of the Einstein constraint equations on asymptotically Euclidean manifolds}, Class. Quantum Grav. \textbf{31} (6), 065001, 10 pp., 2014.
\bibitem{FST} A.~Fang, J.~Szeftel and A.~Touati, \emph{Initial data for Minkowski stability with arbitrary decay}, arXiv:2401.14353, To appear in Adv. Theor. Math. Phys..
\bibitem{FSTKerr} A.~Fang, J.~Szeftel, A.~Touati, \emph{Spacelike initial data for black hole stability}, arXiv:2405.02071.
\bibitem{graf} O.~Graf, \emph{Global nonlinear stability of Minkowski space for spacelike-characteristic initial data}, arXiv:2010.12434, To appear in M\'emoires de la SMF.
\bibitem{Hintz} P.~Hintz, \emph{Exterior stability of Minkowski space in generalized harmonic gauge}, Arch. Rational Mech. Anal. \textbf{247} (99), 2023.
\bibitem{Hin1} P.~Hintz, \emph{Gluing small black holes into initial data sets}, Commun. Math. Phys. \textbf{405} (5), 114 (2024).
\bibitem{Hin2} P.~Hintz, \emph{Underdetermined elliptic PDE on asymptotically Euclidean manifolds, and generalizations}, arXiv:2309.05655.
\bibitem{Hin3} P.~Hintz, \emph{Gluing small black holes along timelike geodesics I: Formal solution}, arXiv:2304.05988.
\bibitem{Hin4} P.~Hintz, \emph{Gluing small black holes along timelike geodesics II: uniform analysis on glued spacetimes}, arXiv:2408.06712.
\bibitem{Hin5} P.~Hintz, \emph{Gluing small black holes along timelike geodesics III: construction of true solutions and extreme mass ratio mergers}, arXiv:2408.06715.
\bibitem{hv} P.~Hintz and A.~Vasy, \emph{Stability of Minkowski space and polyhomogeneity of the metric}, Ann. PDE \textbf{6} (1): Art. 2, 146pp, 2020.
\bibitem{HKKZ} S.~Hirsch, D.~Kazaras, M.~A.~Khuri and Y.~Zhang, \emph{Spectral torical band inequalities and generalizations of the Schoen–Yau black hole existence theorem}, Int. Math. Res. Not., 2023.
\bibitem{holzegel} G.~Holzegel, \emph{Ultimately Schwarzschildean spacetimes and the black hole stability problem}, arXiv:1010.3216.
\bibitem{huneau} C.~Huneau, \emph{Stability of Minkowski spacetime with a translation space-like Killing field}, Ann. PDE \textbf{4} (1): Art. 12, 147 pp, 2018.
\bibitem{Isenberg} J.~Isenberg, \emph{Constant mean curvature solutions of the Einstein constraint equations on closed manifolds}, Class. Quantum Grav. \textbf{12}, 2249--2274, 1995.
\bibitem{KU} C.~Kehle and R.~Unger, \emph{Gravitational collapse to extremal black holes and the third law of black hole thermodynamics}, J. Eur. Math. Soc., 2025.
\bibitem{Ksurvey} S.~Klainerman, \emph{The black hole stability problem}, C. R. M\'ecanique \textbf{353} (2025), no. G1, 555--581.
\bibitem{KLR} S.~Klainerman, J.~Luk and I.~Rodnianski, \emph{A fully anisotropic mechanism for formation of trapped surface in vacuum}, Invent. Math., \textbf{198} (1), 1--26, 2014.
\bibitem{kn} S.~Klainerman and F.~Nicol\`o, \emph{The Evolution Problem in General Relativity}, Progress in Mathematical Physics, Vol. \textbf{25}, 2003.
\bibitem{kr} S.~Klainerman and I.~Rodnianski, \emph{On the formation of trapped surfaces}, Acta Math \textbf{208}, 211--333, 2012.
\bibitem{Lich} A.~Lichnerowicz, \emph{L'int\'egration des \'equations de la gravitation relativiste et le probl\`eme des $n$ corps}, J. Math. Pures Appl. (9) \textbf{23}, 37--63, 1944.
\bibitem{LL} J.~Li and J.~Liu, \emph{Instability of spherical naked singularities of a scalar field under gravitational perturbations}, J. Differ. Geom. \textbf{120} (1): 97--197, 2022.
\bibitem{LY} J.~Li and P.~Yu, \emph{Construction of Cauchy data of vacuum Einstein field equations evolving to black holes}, Ann. of Math., \textbf{181} (2), 699--768, 2015.
\bibitem{lr1} H.~Lindblad and I.~Rodnianski, \emph{Global existence for the Einstein vacuum equations in wave coordinates}, Comm. Math. Phys. \textbf{256} (1), 43–110, 2005.
\bibitem{lr2} H.~Lindblad and I.~Rodnianski, \emph{The global stability of Minkowski spacetime in harmonic gauge}, Ann. of Math. (2), \textbf{171} (3), 1401--1477, 2010.
\bibitem{luk} J.~Luk, \emph{On the Local Existence for the Characteristic Initial Value Problem in General Relativity}, Int. Math. Res. Not. no. 20, 4625--4678, 2018.
\bibitem{LM} J.~Luk and G.~Moschidis, \emph{On the non-existence of trapped surfaces under low-regularity bounds}, Pure Appl. Math. Q. \textbf{20} (4), 1463--1504, 2024.
\bibitem{MOT} Y.~Mao, S.-J. Oh and Z.~Tao, \emph{Initial data gluing in the asymptotically flat regime via solution operators with prescribed support properties}, arXiv:2308.13031.
\bibitem{MaoTao} Y.~Mao and Z.~Tao, \emph{Localized initial data for Einstein equations}, arXiv:2210.09437.
\bibitem{OhTataru} S.-J.~Oh and D.~Tataru, \emph{The hyperbolic Yang–Mills equation for connections in an arbitrary topological class}, Commun. Math. Phys. \textbf{365} (2), 685–739, 2019.
\bibitem{Penrose} R.~Penrose, \emph{Gravitational collapse and space-time singularities}, Phys. Rev. Lett. \textbf{14}, 57--59, 1965.
\bibitem{PenroseWCC} R.~Penrose, \emph{Gravitational collapse: The role of general relativity}, Riv. Nuovo Cim. \textbf{1}, 252--276, 1969.
\bibitem{SYbh} R.~M.~Schoen and S.-T.~Yau, \emph{The existence of a black hole due to condensation of matter}, Commun. Math. Phys. \textbf{90}, 575–579, 1983.
\bibitem{SchoenYau94} R.~M.~Schoen and S.-T.~Yau, \emph{Lectures on Differential Geometry}, International Press, 1994.
\bibitem{Shen22} D.~Shen, \emph{Stability of Minkowski spacetime in exterior regions}, Pure Appl. Math. Q. \textbf{20} (2), 757--868, 2024.
\bibitem{ShenKerr} D.~Shen, \textit{Kerr stability in external regions}, Ann. PDE \textbf{10}, Art. 9, 2024.
\bibitem{Shen23} D.~Shen, \emph{Global stability of Minkowski spacetime with minimal decay}, arXiv:2310.07483.
\bibitem{Shen24} D.~Shen, \emph{Exterior stability of Minkowski spacetime with borderline decay}, arXiv:2405.00735.
\bibitem{ShenWan} D.~Shen and J.~Wan, \emph{Formation of trapped surfaces for the Einstein-Maxwell-charged scalar field system}, arXiv:2504.19976.
\bibitem{Smulevici} J.~Smulevici, \emph{The stability of Minkowski space and its influence on the mathematical analysis of General Relativity}, C. R. M\'ecanique \textbf{353} (2025), no. G1, 519--542.
\bibitem{SormaniStavrov} C.~Sormani and I.~Stavrov Allen, \emph{Geometrostatic manifolds of small ADM mass}, Comm. Pure Appl. Math. \textbf{72} (2019), no.~6, 1243--1287.
\bibitem{Yau02} S.-T.~Yau, \emph{Geometry of three manifolds and existence of black hole due to boundary effect}, Adv. Theor. Math. Phys. \textbf{5}, 755--767, 2001.
\bibitem{Yucmp} P.~Yu, \emph{Energy estimates and gravitational collapse}, Comm. Math. Phys. \textbf{317} (2), 275--316, 2013.
\bibitem{Yu} P.~Yu, \emph{Dynamical formation of black holes due to the condensation of matter field}, arXiv:1105.5898.
\bibitem{Zhao} P.~Zhao, D.~Hilditch and J.~A.~Valiente Kroon, \emph{Trapped surface formation for the Einstein-scalar system}, arXiv:2304.01695.
\end{thebibliography}
\end{document}